%% file: main.tex
\def\draft{0}
\def\doubleblind{0}
\documentclass[11pt,letterpaper]{article}
\usepackage[margin=1in]{geometry}
\usepackage[utf8]{inputenc}
\usepackage{amsthm}
\usepackage{amsthm, amsmath, amssymb, mathtools}
\usepackage{bbm,bm}
\usepackage{graphicx}
\usepackage{color,xcolor}
\usepackage{paralist,enumitem}
\usepackage{mathrsfs}
\usepackage[normalem]{ulem}
\usepackage{tabularx}
\usepackage{tikz}
\usepackage{caption,comment}

\usepackage{Files/singer-macros}

\usepackage{algorithmicx}
\usepackage{algorithm}
\usepackage{algpseudocode}
\newcounter{algsubstate}
\renewcommand{\thealgsubstate}{\alph{algsubstate}}

\algnewcommand\algorithmicinput{\textbf{Input:}}
\algnewcommand\Input{\item[\algorithmicinput]}

\algnewcommand\algorithmicoutput{\textbf{Output:}}
\algnewcommand\Output{\item[\algorithmicoutput]}

\algnewcommand\algorithmicgoal{\textbf{Goal:}}
\algnewcommand\Goal{\item[\algorithmicgoal]}

\newcommand{\alglinenoNew}[1]{\newcounter{ALG@line@#1}}
\newcommand{\alglinenoPop}[1]{\setcounter{ALG@line}{\value{ALG@line@#1}}}
\newcommand{\alglinenoPush}[1]{\setcounter{ALG@line@#1}{\value{ALG@line}}}

\usepackage[colorlinks=true, allcolors=blue]{hyperref}
\usepackage[capitalize,nameinlink]{cleveref}

\newcommand{\vnote}[1]{\ifnum\draft=1\textcolor{orange}{[\textbf{Santhoshini:} #1]}\fi}
\newcommand{\mnote}[1]{\ifnum\draft=1\textcolor{red}{[\textbf{Madhu:} #1]}\fi}
\newcommand{\snote}[1]{\ifnum\draft=1\textcolor{teal}{[\textbf{Noah:} #1]}\fi}
\newcommand{\rnote}[1]{\ifnum\draft=1\textcolor{brown}{[\textbf{Raghuvansh:} #1]}\fi}

\newcommand{\blind}[2]{{\ifnum\draft=1\color{purple}\fi \ifnum\doubleblind=1#2\fi\ifnum\doubleblind=0#1\fi\ifnum\doubleblind=2$\{$ #1 $\vert$ #2 $\}$\fi}}

\input{Files/thmmacros}
\input{Files/macros}

\allowdisplaybreaks

\title{Improved Streaming Algorithms for Maximum Directed Cut via Smoothed Snapshots}

\ifnum\doubleblind=0
\author{Raghuvansh R. Saxena\thanks{Microsoft Research. Email: \texttt{raghuvansh.saxena@gmail.com}} 
\and Noah G. Singer\thanks{Department of Computer Science, Carnegie Mellon University, Pittsburgh, PA, USA. Supported by an NSF Graduate Research Fellowship (Award DGE2140739). Email: \texttt{ngsinger@cs.cmu.edu}.} 
\and Madhu Sudan\thanks{School of Engineering and Applied Sciences, Harvard University, Cambridge, Massachusetts, USA. Supported in part by a Simons Investigator Award and NSF Award CCF 2152413. Email: \texttt{madhu@cs.harvard.edu}.}
\and Santhoshini Velusamy\thanks{School of Engineering and Applied Sciences, Harvard University, Cambridge, Massachusetts, USA. Supported in part by a Google Ph.D. Fellowship, a Simons Investigator Award to Madhu Sudan, and NSF Award CCF 2152413. Email: \texttt{svelusamy@g.harvard.edu}.}}
\fi

\date{}
\begin{document}
\maketitle

\begin{abstract}

 We give an $\widetilde{O}(\sqrt{n})$-space single-pass $0.483$-approximation streaming algorithm for estimating the maximum directed cut size ($\mdcut$) in a directed graph on $n$ vertices. This improves over an $O(\log n)$-space $4/9 < 0.45$ approximation algorithm  due to Chou, Golovnev, and Velusamy (FOCS 2020), which was known to be optimal for $o(\sqrt{n})$-space algorithms. $\mdcut$ is a special case of a {\em constraint satisfaction problem} (CSP). In this broader context, we give the {\em first} CSP for which algorithms with $\widetilde{O}(\sqrt{n})$ space can provably outperform $o(\sqrt{n})$-space algorithms. 
 
 The key technical contribution of our work is development of the notions of a first-order snapshot of a (directed) graph and of estimates of such snapshots. These snapshots can be used to simulate certain (non-streaming) $\mdcut$ algorithms, including the ``oblivious'' algorithms introduced by Feige and Jozeph (Algorithmica, 2015), who showed that one such algorithm achieves a 0.483-approximation.

Previous work \blind{of the authors}{by Saxena, Singer, Sudan, and Velusamy} (SODA 2023) studied the restricted case of bounded-degree graphs, and observed that in this setting, it is straightforward to estimate the snapshot with $\ell_1$ errors and this suffices to simulate oblivious algorithms. But for unbounded-degree graphs, even defining an achievable and sufficient notion of estimation is subtle. We describe a new notion of snapshot estimation and prove its sufficiency using careful smoothing techniques, and then develop an algorithm which sketches such an estimate via a delicate process of intertwined vertex- and edge-subsampling.

Prior to our work, the only streaming algorithms for {\em any} CSP on general instances were based on generalizations of the $O(\log n)$-space algorithm for $\mdcut$, and can roughly be characterized as based on ``zeroth'' order snapshots. Our work thus opens the possibility of a new class of algorithms for approximating CSPs by demonstrating that more sophisticated snapshots can outperform cruder ones in the case of $\mdcut$. 

\end{abstract}

\newpage 
\setcounter{tocdepth}{2}
{\small \tableofcontents}

\input{Files/01-intro}
\input{Files/02-prelims}
\input{Files/03-snapshots}
\input{Files/04-algorithm}

\input{Files/05-alg-proofs}

\input{Files/06-snapshot-proofs}

\input{Files/07-smoothing-proofs}


\bibliographystyle{alpha}
\bibliography{csps}

\appendix
\input{Files/AA-appendix-proofs}

\end{document}

%% file: Files/thmmacros.tex
\usepackage{amsthm}
\usepackage{thmtools,thm-restate}

\numberwithin{equation}{section}
\declaretheoremstyle[bodyfont=\it,qed=\qedsymbol]{noproofstyle}

\declaretheorem[name=Observation,numbered=no]{observation*}

\declaretheorem[numberlike=equation]{fact}

\declaretheorem[numberlike=equation]{theorem}

\declaretheorem[name=Theorem,numbered=no]{theorem*}

\declaretheorem[numberlike=equation]{lemma}
\declaretheorem[name=Lemma,numbered=no]{lemma*}

\declaretheorem[numberlike=equation]{corollary}
\declaretheorem[name=Corollary,numbered=no]{corollary*}

\declaretheorem[numberlike=equation]{proposition}
\declaretheorem[name=Proposition,numbered=no]{proposition*}

\declaretheorem[numberlike=equation]{claim}
\declaretheorem[name=Claim,numbered=no]{claim*}

\declaretheorem[name=Conjecture,numbered=no]{conjecture*}

\declaretheorem[name=Question,numbered=no]{question*}

\declaretheoremstyle[bodyfont=\it]{defstyle} 

\declaretheorem[numberlike=equation,style=defstyle]{definition}
\declaretheorem[unnumbered,name=Definition,style=defstyle]{definition*}

\declaretheorem[unnumbered,name=Example,style=defstyle]{example*}

\declaretheorem[unnumbered,name=Notation=defstyle]{notation*}

\declaretheorem[unnumbered,name=Construction,style=defstyle]{construction*}

\declaretheoremstyle[]{rmkstyle}

\crefname{claim}{Claim}{Claims}
\crefname{fact}{Fact}{Facts}

%% file: Files/macros.tex
\newcommand{\mcut}{\mathsf{Max}\text{-}\mathsf{CUT}}
\newcommand{\mdcut}{\mathsf{Max}\text{-}\mathsf{DICUT}}

\newcommand{\Win}{\mathsf{Win}}
\newcommand{\Ind}{\mathsf{ind}}
\newcommand{\bi}{\mathsf{b}\text{-}\mathsf{ind}}
\newcommand{\di}{\mathsf{d}\text{-}\mathsf{ind}}
\newcommand{\dbi}{\mathsf{db}\text{-}\mathsf{ind}}

\newcommand{\abi}{\hat{\mathsf{b}}\text{-}\mathsf{ind}}
\newcommand{\adi}{\hat{\mathsf{d}}\text{-}\mathsf{ind}}

\newcommand{\M}{\mathbb{M}}
\newcommand{\A}{\mathbb{A}}
\newcommand{\MS}{\M_\Delta}
\newcommand{\MN}{\M_{\geq 0}}
\newcommand{\AS}{\A_\Delta}
\newcommand{\AN}{\A_{\geq 0}}

\newcommand{\sstar}{(\star,\star)}

\newcommand{\Snap}{\mathsf{Snap}}
\newcommand{\SnapG}{\Snap_{\CG,\vect}}
\newcommand{\SnapGw}{\SnapG^{\sim w}}
\newcommand{\SnapH}{\Snap_{\CH,\vect}}

\newcommand{\Proj}{\mathsf{Proj}}

\newcommand{\RSnap}{\mathsf{RSnap}}
\newcommand{\RSnapG}{\RSnap_{\CG,\vecd,\vect}}

\newcommand{\Adj}{\mathsf{AdjMat}}
\newcommand{\AdjG}{\Adj_\CG}
\newcommand{\AdjH}{\Adj_\CH}
\newcommand{\AdjHZ}{\Adj_{\CH_0}}
\newcommand{\AdjK}{\Adj_\CK}
\newcommand{\AdjL}{\Adj_\CL}
\newcommand{\AdjLP}{\Adj_{\CL'}}

\newcommand{\Obl}{\mathcal{A}}

\newcommand{\val}{\mathsf{val}}
\newcommand{\bias}{\mathsf{bias}}
\renewcommand{\d}{\mathsf{deg}}
\renewcommand{\deg}{\d}
\newcommand{\dout}{\mathsf{deg}\text{-}\mathsf{out}}
\newcommand{\din}{\mathsf{deg}\text{-}\mathsf{in}}

\newcommand{\1}{\mathbbm{1}}
\newcommand{\Bern}{\mathsf{Bern}}

\newcommand{\orig}{\mathtt{orig}}
\newcommand{\RefinePart}{\mathtt{RefinePart}}
\newcommand{\RefineAlg}{\mathtt{RefineAlg}}

\newcommand{\eStored}{\mathtt{eStored}}
\newcommand{\vStored}{\mathtt{vStored}}

\newcommand{\vCutoff}{\mathtt{vCutoff}}
\newcommand{\eCutoff}{\mathtt{eCutoff}}

\newcommand{\bEst}{\mathtt{bEst}}
\newcommand{\dEst}{\mathtt{dEst}}
\newcommand{\eEst}{\mathtt{eEst}}
\newcommand{\vEst}{\mathtt{vEst}}
\newcommand{\AEst}{\mathtt{AEst}}

\newcommand{\nEst}{\nu\mathtt{Est}}

\newcommand{\mmin}{m_{\mathtt{min}}}
\newcommand{\mmax}{m_{\mathtt{max}}}

\newcommand{\Cwin}{C_{\mathrm{win}}}
\newcommand{\Csmooth}{C_{\mathrm{smooth}}}
\newcommand{\Csparse}{C_{\mathrm{spar}}}

\newcommand{\epsbias}{\epsilon_{\mathrm{bias}}}
\newcommand{\epserr}{\epsilon_{\mathrm{err}}}
\newcommand{\epssparse}{\epsilon_{\mathrm{spar}}}

\DeclarePairedDelimiter{\abs}{\lvert}{\rvert}
\DeclarePairedDelimiter{\card}{\lvert}{\rvert}

\DeclarePairedDelimiter{\paren}{(}{)}
\DeclarePairedDelimiter{\set}{\{}{\}}

\newcommand{\mfail}{\mathtt{overflow}}
\newcommand{\ssparse}{\vecsigma_{\mathrm{spar}}}
\newcommand{\hmsparse}{\hat{m}_{\mathrm{spar}}}
\newcommand{\msparse}{m_{\mathrm{spar}}}
\newcommand{\Gsparse}{\CG_{\mathrm{spar}}}
\newcommand{\pspar}{p_{\mathrm{spar}}}

\newcommand{\FJ}{\mathrm{FJ}}

\newcommand{\Hash}{\CH}

\makeatletter
\newcommand{\algmargin}{\the\ALG@thistlm}
\algnewcommand{\parState}[1]{\State%
  \parbox[t]{\dimexpr\linewidth-\algmargin}{\strut #1\strut}}

%% file: Files/01-intro.tex

\newpage

\section{Introduction}

We consider approximating the maximum directed cut value of a directed graph by a streaming algorithm presented with a stream of edges in an arbitrary (worst-case) order. Our main result is a single-pass algorithm using $\widetilde{O}(\sqrt{n})$-space that gives a $.483$ approximation algorithm. Along the way we develop the notions of snapshots of graphs and estimates of such snapshots, which introduce new tools for approximating graph theoretic quantities and more generally for approximating Constraint Satisfaction Problems (CSPs). In what follows we explain the background of the directed cut problem, the significance of the result, and the techniques used to achieve this result.

\subsection{Background}

We begin by defining the maximum directed cut ($\mdcut$) problem in a directed graph $\CG$. (These definitions will all be informal; see \cref{sec:prelims} for formal definitions.) Given a  graph $\CG$ on $n$ vertices, labeled $1, \ldots, n$, \emph{cut} of $\CG$ is a binary string $\vecx \in \{0,1\}^n$, assigning a bit to every vertex in $\CG$. We say $\vecx$ \emph{cuts} a directed edge $(u,v)$ if $x_u = 1$ and $x_v = 0$. (Note the asymmetry between $u$ and $v$.) The \emph{value} $\val_\CG(\vecx)$ of a cut $\vecx$ is the total fraction of edges it cuts, and the \emph{value} $\val_\CG$ of $\CG$ is the maximum value of any cut. A uniformly random cut has value $\frac14$ in expectation, so every graph has value at least $\frac14$.

We consider \emph{streaming} algorithms for the problem of estimating the $\mdcut$ value $\val_\CG$ of a directed graph $\CG$, given a stream $\vecsigma = (e_1 = (u_1,v_1),\ldots,e_m = (u_m,v_m))$ of the graph's edges in arbitrary order. We say an algorithm is an \emph{$\alpha$-approximation} for the $\mdcut$ problem if its output $\hat{v}$ satisfies $\alpha \cdot \val_\CG \leq \hat{v} \leq \val_\CG$ (with high probability). We say an algorithm is a \emph{space-$s(n)$ streaming algorithm} (where $n$ is the number of vertices in $\CG$) if it reads the stream of edges $\vecsigma$ in sequential order and uses $s(n)$ space.

The $\mdcut$ problem is one example of a so-called \emph{constraint satisfaction problem (CSP)}. We omit a full definition as we do not require it, but these problems are basically defined by two things: (1) a ``global'' space of allowed ``assignments'' to ``variables'' and (2) a collection of ``local'' constraints, each of which specifies allowed values for a small subset of variables. For $\mdcut$, variables are vertices, assignments are cuts, and constraints are edges; we will use these terms interchangeably. The ``symmetric version'' of $\mdcut$ is another CSP called \emph{maximum cut} ($\mcut$), in which a cut $\vecx$ cuts an edge $(u,v)$ if $x_u \neq x_v$; we mention it here as it serves a useful point of comparison for $\mdcut$.

\subsection{Recent work}

Over the last decade, there has been extensive work on the approximability of various CSPs in various streaming models \cite{KK15,KKS14,GVV17,KKSV17,GT19,KK19,CGV20,CGSV21-finite,SSV21,CGS+22-linear-space,BHP+22,CGS+22-monarchy,SSSV23-random-ordering}; see also the surveys \cite{Sin22,Sud22}.

$\mdcut$ has emerged as the central benchmark for algorithms among CSPs in the streaming setting. It was the first problem shown to admit a non-trivial approximation in sublinear (in $n$) space in the work of Guruswami, Velingker, and Velusamy~\cite{GVV17}. Subsequent work of Chou, Golovnev, and Velusamy~\cite{CGV20} gave an improved algorithm for $\mdcut$ along with a tight bound on the approximability --- pinning the approximability of $\mdcut$ for $o(\sqrt n)$-space streaming at $\frac49$.

\begin{theorem}[\cite{CGV20}]\label{thm:cgv20}
For every $\epsilon > 0$, there is a streaming algorithm (in fact, a linear sketching algorithm) which $(4/9-\epsilon)$-approximates the $\mdcut$ value of a graph in $O_\epsilon(\log n)$ space. Conversely, every $(4/9+\epsilon)$-approximation streaming algorithm for $\mdcut$ uses $\Omega_\epsilon(\sqrt n)$ space.
\end{theorem}

Both the algorithms in \cite{GVV17} and \cite{CGV20} are what previous works have called ``bias-based'' algorithms, or what we will call a ``zeroth-order snapshot'' algorithms. Roughly, the bias of a vertex captures the ratio of its in-degree to its out-degree, and a zeroth-order snapshot computes a histogram of the bias of vertices in the graph and uses this histogram (and no other information) to approximate the $\mdcut$ value of a graph. Strikingly, the work of \cite{CGV20} shows that zeroth-order snapshot based algorithms are optimal among $o(\sqrt{n})$-space streaming algorithms.

Subsequent work of Chou, Golovnev, Sudan, and Velusamy~\cite{CGSV21-finite} showed that this result is part of a broader landscape for $o(\sqrt n)$-space streaming complexity of CSPs. In particular, Chou, Golovnev, Sudan, and Velusamy~\cite{CGSV21-finite} proved a \emph{dichotomy theorem} for all finite CSPs. The understanding of $\mdcut$ plays a central role in their results. In particular, they generalize the zeroth-order snapshot based $\mdcut$ algorithm of \cite{CGV20} to all CSPs. Their lower bounds also generalize the lower bounds from \cite{CGV20} with some notions (``padded one-wise independent problems'') that are direct abstractions of $\mdcut$ and share tight lower bounds. 

One might ask, given a particular CSP, if there are any algorithms that outperform zeroth-order snapshot algorithms studied in \cite{CGSV21-finite}. For a wide class of CSPs, including $\mcut$, the answer is ``NO'' --- there are recent $\Omega(n)$-space lower bounds ruling out all nontrivial approximations \cite{KK19,CGS+22-linear-space}.\footnote{$\Omega(n)$ space is tight up to logarithmic factors because randomly sparsifying down to $O(n/\epsilon^2)$ constraints gives $(1-\epsilon)$-approximations.}\footnote{The condition for inapproximability given in \cite{CGS+22-linear-space} for a predicate $f : \BZ_q^k \to \{0,1\}$ is termed ``width'', and states that $f$'s support contains some translate of the diagonal $\{(a,\ldots,a):a \in \BZ_q^k\}$. More broadly, the strongest known hardness results for CSPs (e.g., also in \cite{CGSV21-finite}) seem to rely on ``niceness'' properties of the support of $f$.} Thus to make advances one has to restrict the problems considered, and in this work we focus on the simplest remaining problem after $\mcut$, namely, $\mdcut$. 

For $\mdcut$, till this work and \blind{a recent related work by the authors}{a recent related work} \cite{SSSV23-random-ordering} it was conceivable that there were no improvements possible in $o(n)$ space. But at the same time the above mentioned lower bound from \cite{CGV20} did not extend to this setting and it was unclear whether this was due to a limitation of the lower bounds techniques or if better algorithms exist.

\blind{In a previous work \cite{SSSV23-random-ordering}, the authors gave}{The prior work \cite{SSSV23-random-ordering} gives} some evidence for the possibility that better algorithms for $\mdcut$ do indeed exist. To be precise, recall that the sketching algorithm of \cite{CGV20} is a $\frac49 \approx 0.444$-approximation, which uses $O(\log n)$ space and is optimal among $o(\sqrt n)$-space streaming algorithms (\cref{thm:cgv20}). \blind{In \cite{SSSV23-random-ordering} we proved}{In \cite{SSSV23-random-ordering} it is proved} that for $\mdcut$, the algorithm of \cite{CGV20} \emph{can} be beaten in certain restricted models such as when the input stream is \emph{randomly} (instead of adversarially) ordered, or the graph has constant maximum-degree. In particular:

\begin{theorem}[\cite{SSSV23-random-ordering}]\label{thm:bounded-degree-alg}
For every $d \in \BN$, there is a streaming algorithm which $0.483$-approximates the $\mdcut$ value of a graph with maximum degree $d$ in $\tilde{O}_d(\sqrt n)$ space.
\end{theorem}

In doing so the work of \cite{SSSV23-random-ordering} introduces the notion that we call a ``first-order snapshot'' --- where information about the input graph is ``compressed'' to a histogram of edges based on the biases of their two endpoints. (See \cref{def:snapshot} below for the definition of a snapshot. We refer here to ``first-order'' snapshots since higher-order snapshots might maintain a histogram of longer length paths or other subgraphs with more than one edge.) For bounded-degree graphs, mapping a graph to its snapshot is clearly a compression (since the number of possible biases is finite), and this compressed information can be estimated, under fairly natural notions of estimation, by an $\tilde{O}(\sqrt n)$-space streaming algorithm. 

However, \cref{thm:bounded-degree-alg} does not answer the question of whether the $\mdcut$ algorithm of \cite{CGV20} can be beaten on \emph{general} graphs in $o(n)$ space. Indeed, their algorithm breaks down in a fundamental way on general graphs, so it could be considered evidence only that more sophisticated lower bound techniques are necessary to rule out such algorithms. We further discuss why we believe that \cref{thm:bounded-degree-alg} was far from a resolution to this question in \cref{sec:significance} below.

\subsection{Main result}

Our main theorem gives an algorithm that uses slightly more than $\sqrt n$ space and outperforms the algorithm of \cite{CGV20}:

\begin{theorem}[Main theorem]\label{thm:main-thm}
There is a streaming algorithm which $0.483$-approximates the $\mdcut$ value of an arbitrary (multi)graph in $\tilde{O}(\sqrt n)$ space.
\end{theorem}

See \cref{thm:main-thm-formal} below for the fully detailed statement.

The fact that we achieve the same approximation factor as \cite{SSSV23-random-ordering} is not a coincidence. Both works obtain their final algorithm by constructing a first-order snapshot of the input graph, and then observing that this information suffices to simulate the performance of ``oblivious'' algorithms on the given input, and finally using a result of Feige and Jozeph~\cite{FJ15} that gives an oblivious algorithm to approximate $\mdcut$ to a factor of $\approx .483$. (Roughly, oblivious algorithms randomly and independently assign vertices to either the $0$-side or the $1$-side where the probability of choosing a side depends on the bias of the vertex, and these probabilities are chosen to optimize the expected number of edges crossing the cut --- a quantity that can be optimized using just the first-order snapshot information.) While this chain of reasoning is similar, every step becomes more complex in the unbounded-degree setting. Indeed as we explain below, designing algorithms for bounded-degree graphs is and has been substantially easier than the general case.

\subsection{Beyond bounded-degree instances}\label{sec:significance}

Before turning to our setting with $\tilde{O}(\sqrt n)$ space, we first remark on the role of degree in the earlier works of \cite{GVV17,CGV20,CGSV21-finite}. The algorithms in all these works work for general degree graphs and use powerful norm estimation algorithms as black boxes. If one were to consider the simpler case of their problems in the bounded-degree setting, these algorithms could have been implemented without reliance on these subroutines. Specifically, their algorithms only need an estimate of the absolute value of ``bias times the degree'' for a random vertex, and this could be estimated by simply picking a random sample of the vertices and computing their bias and degree as the stream passes by. For general CSPs (even on non-Boolean domains) also such a process would suffice, and this would not only simplify the algorithms significantly, it even would achieve a space bound of $O(\log n)$ which is better than the current bounds given in \cite{CGSV21-finite} for general CSPs.

Digging deeper into this analogy one can consider $\ell_p$ norm estimation problems themselves. For this class of problems also one can define a bounded-degree version of the problem --- where one is trying to compute the $\ell_p$ norm of a vector in $\{-C,\ldots,C\}^n$ in the turnstile update model. In this bounded-degree setting, the $\ell_p$ norm can be trivially computed by randomly sampling an $O_C(1)$-sized subset of the coordinates and maintaining their values. Thus $\ell_p$ norms can be estimated in $O(\log n)$ space for every $p$ in this bounded-degree setting, whereas in the general case it is well-known that $\ell_p$ norm estimation requires polynomial in $n$ space for $p > 2$. 

Thus, the bounded-degree setting can be vastly easier to solve and results in this setting may best be viewed as a proof of concept --- though even this ``proof of concept'' may be misleading, as exemplified by the $\ell_p$ norm estimation problem. 

Turning to our specific goal --- that of computing (first-order) snapshots of a graph in $\widetilde{O}(\sqrt{n})$ time --- \blind{our prior work}{the prior work of} \cite{SSSV23-random-ordering} again manages to estimate this snapshot in the bounded-degree setting by sampling $\widetilde{O}(\sqrt{n})$ vertices and maintaining the bias of the sampled vertices as well as the induced subgraph on these vertices. We discuss why this is reasonable in the bounded-degree setting in the following subsection. But such a simple algorithm is definitely not going to work in the general setting! In particular, computing a good estimate of the snapshot is at least as hard as computing the $\ell_1$ norm of a vector in the turnstile model with unit updates. Indeed, ``snapshot'' estimation seems to be a ``higher-level'' challenge than simple norm estimation and roughly requires computing some ``two-wise'' marginals of the graph updates, whereas bias corresponded to ``one-wise'' marginals. Black-box use of norm-estimation algorithms no longer seems to suffice to solve these ``two-wise'' marginal problems, which seem to need new algorithmic ideas. We feel this class of problems and the ideas used here to deal with them may be of even broader interest than the application to $\mdcut$. 

\subsection{Technical overview}\label{sec:overview}

Our goal is to approximate the $\mdcut$ value of a graph $\CG$ by estimating its snapshot $\SnapG$. 

Setting aside the streaming model momentarily, the ``gold standard'' way to estimate the snapshot would be to sample a small set $E$ of edges uniformly and independently at random, measure the biases of the endpoints of every edge in $E$, and use this to estimate the snapshot. Unfortunately, since the stream is adversarially-ordered, there is no obvious way to implement this procedure since by the time a ``random'' edge appears in our stream, many of the edges incident to its endpoints might have already appeared, and thus, we may not know its endpoints' biases.\footnote{As observed in \cite{SSSV23-random-ordering}, when the edges in the stream are \emph{randomly} ordered this simple setup does give an algorithm: One can simply set $E$ to be the first $O(1)$ edges in the stream and then observe the biases of the endpoints over the remainder of the stream.}

To get an algorithm for adversarially-ordered streams, we could hope to somehow sample a set $E$ of edges in a way which maintains the property that for every edge in $E$, we know the bias of its endpoints. While $E$ may not be a uniformly random set of edges, we could still hope for an estimate of the snapshot if $E$ is ``sufficiently'' random. A natural approach proposed in \cite{SSSV23-random-ordering} for doing this is the following. \blind{We}{They} sample a uniform set $S$ of \emph{vertices} upfront, i.e., before the stream, by uniformly including every vertex with some probability $p$ independently. Then, during the stream, \blind{we}{they} measure the bias of every vertex in $S$ and store the \emph{induced subgraph} on $S$ as $E$. Since $S$ is sampled before the stream begins, this approach has the advantage that even though the graph is adversarially-ordered, we end up knowing the biases of the endpoints of every edge in $E$. Here is where the $\sqrt{n}$ space dependency comes from: By a ``birthday paradox'' argument\footnote{It suffices to consider only sparse graphs (see \cref{lem:cut-sparsifier}).}, since $E$ is the induced subgraph on $S$, in order to expect to even see any edges in $E$ we will need $|S| = \Omega(\sqrt n)$. But we are very far from done at this point, because there is a crucial issue as compared to the gold standard case: The edges which are included in $E$ are no longer independent! In particular, if two edges $e$ and $e'$ share a common endpoint (or two, in the case of a multigraph!), then conditioning on $e \in E$ increases the probability that $e' \in E$. Here is where the maximum-degree assumption in \cite{SSSV23-random-ordering} comes in: If $\CG$ has maximum-degree $D$, $e \in E$ is independent of all but $\leq 2D+1$ events $e' \in E$. It turns out that when $D=O(1)$, this lets us get enough control on the variance of which edges show up in $E$ to give a correct estimate. But for larger $D$, this approach completely breaks down, and for good reason: In the extreme example of a \emph{star graph} (i.e., a graph where one vertex is connected to all other vertices), we must store the center of the star in $S$, or otherwise $E$ will be empty! But if we place \emph{every} vertex in $S$ we will use linear space --- we want to store the center with probability $1$, but the other (low-degree) vertices with probability only $O(1/\sqrt{n})$. Thus, in order to extend this simple estimator to general graphs, we will very roughly want to place vertices in $S$ with probability which increases as a function of their degree. Implementing this in the streaming setting creates numerous challenges, and solving these is a main focus of this paper.

\subsubsection{Vertex-sampling in the general case}\label{sec:overview:sampling-vertices}

Our goal now is to extend the vertex sampling approach described above to general graphs. We remark that even in the general case, we can assume WLOG that the number of edges in the graph is $\Omega(\sqrt n)$ and $O(n)$.\footnote{If the graph has $O(\sqrt n)$ edges we can afford to store the entire graph within our space bound. A standard sparsification argument (see \cref{lem:cut-sparsifier}) shows replacing $\CG$ with a random subsample of $O(n/\epsilon^2)$ edges changes the $\mdcut$ value by only $\epsilon$.}

As we mentioned in the previous subsection, we would like to sample a set $S$ of vertices, such that every vertex is included independently with probability which increases with the degree. This is the first step towards estimating the snapshot, which will also require sampling edges between these vertices; we focus on the former task for now, and address the latter in the following subsection.

Instead of sampling one set $S$ of vertices, we will aim for a slightly more detailed goal, which is to sample a set $U_a$ of vertices of degree between $d_{a-1}$ and $d_a$, where $1 = d_1 < \cdots < d_k = O(n)$ is some partition of the possible degrees in the graph. (For concreteness, we use $d_a = 2^{a-1}$.) We envision the graph as consisting of $k$ \emph{layers}; a vertex of degree between $d_{a-1}$ and $d_a$ is in layer $a$ (and has ``degree class'' $a$, in analogy to the bias classes). 

Consider the task of sampling $U_a$, a uniform set of vertices in a fixed layer $a$. Recall that in the previous subsection (the bounded-degree case), we placed all vertices in $S$ with some fixed probability $p$ independently. Now, we would like to place layer-$a$ vertices into $U_a$ with some probability $p_a$ independently.\footnote{There is a technical reason from switching $S$ to $U$ here to denote sets of vertices: It is convenient to think of first sampling a set $S_a$ before the stream to include \emph{all} vertices w.p. $p_a$ (even those not in layer $a$), and then $U_a$ is the intersection of $S_a$ with the set of vertices in layer $a$, which are the vertices we actually want to track.} To fit within our space bound, we only require $|U_a| = O(\sqrt n)$; this turns out to imply that $p_a$ can grow as a function of $a$. For instance, if $d_a = \Omega(\sqrt n)$ then there can only be $O(\sqrt n)$ vertices in layer $a$, so we can even afford $p_a = 1$. (Making $p_a$ this large in high layers is actually necessary for good estimates, as shown by the ``star'' example.) But there is a seeming paradox in this plan: \emph{When a vertex $v$ first appears in the stream, we would like to know its layer $a$, so that we can toss an appropriately biased coin (i.e., Bernoulli-$p_a$) to determine whether it goes into $U_a$; but as this is the \emph{first} appearance of $v$, we know nothing about its degree besides that it is at least $1$!}

One natural way to deal with this problem is to \emph{defer} deciding whether to a vertex has high degree until we see many edges touching it. To do this, we take advantage of subsampling \emph{edges} as well as vertices. To layer $a$ we also associate an ``edge-subsampling probability'' $q_a$ and a ``subsampled graph'' $\CG_a$ which includes every edge in $\CG$ independently with probability $q_a$. We choose $q_a= Cd_a^{-1}$ for a large constant $C$, meaning that vertices with degree $d_a$ in $\CG$ have degree roughly $C$ in $\CG_a$.\footnote{Actually, in order for adequate concentration of the degree in $\CG_a$, we will need $C = \Omega(\log n)$, but we ignore this for simplicity. Also, if $d_a < C$, we set $q_a=1$, i.e., we need no edge-subsampling. This is equivalent to the bounded-degree case we already analyzed.} This allows us to sample $U_a$ in the streaming setting: We sample $\CG_a$ on the fly, and then we add to $U_a$ with probability $p_a$ each new vertex with $\CG_a$-degree roughly between $0.49C$ and $1.01C$.\footnote{We are cheating slightly here: We do not know the degree of such a vertex when it first appears, so we cannot decide whether it has degree falling in this range. Instead, during the stream we can store a set $\hat{U}_a$ containing each vertex with positive $\CG_a$-degree w.p. $p_a$. Then, after the stream, we set $U_a$ to be the set of vertices in $\hat{U}_a$ with $\CG_a$-degree in the appropriate range. This point will come up again when we want to store $\CG_a$-edges associated to these vertices in the following subsection; we will have to store edges for every vertex in $N_a$ and then use ``cutoffs'' to stop storing edges once we know they cannot be in $U_a$. We ignore these details in this overview.} Note that $\CG_a$ is too large to store --- in particular, $\CG_1$ has $m$ edges, and more generally $\CG_a$ has $q_am$ edges in expectation --- so we will need to carefully choose which edges to store when crafting our estimate in the following subsection.

Now $U_a$ will contain a random sample of layer-$a$ vertices, but --- and this is crucial --- it may also contain other randomly sampled vertices, like those in layers $a-1$ or $a+1$. E.g., consider a vertex of degree $d_a+1$, which is technically in layer $a+1$, but is also likely to have $\CG_a$-degree under $q_ad_a$; indeed, one cannot differentiate between this vertex and a layer-$a$ vertex with high probability based on $\CG_a$-degree. To put it another way, by the time we see the first incident edge to a vertex in $\CG_a$, many of its incident edges in $\CG$ may have already passed by in the stream, meaning we \emph{cannot track} its ``global'' bias or degree exactly. This creates a substantial technical issue in even defining the type of estimates we are trying to achieve, which we have to resolve by certain ``smoothing'' arguments which we defer until the subsection after the following (\cref{sec:overview:analysis}).

\subsubsection{Sampling edges for the estimate}\label{sec:overview:sampling-edges}

In the previous section, we described a scheme based on vertex- and edge-subsampling which samples uniform sets $U_a$ of layer-$a$ vertices (perhaps along with ``borderline'' vertices in layer $a+1$ and $a-1$, which we ignore for now). But our ultimate target is estimating the \emph{snapshot}, which counts edges between vertices of different biases; in this sense, these sets of vertices are only means to an end, and we need to also describe how we sample edges \emph{between} these sampled vertices and use them to estimate the snapshot. Recall that in each layer we subsampled a graph $\CG_a$, which was still too large to store; the edges we use to produce the estimate will be a subset of $\CG_a$'s edges which we can actually store.

Since our algorithm now breaks down vertices by their degrees as well as their biases, it turns out that the natural object to aim to estimate is not the snapshot itself, but instead what we will call the \emph{refined snapshot} of $\CG$, denoted $\RSnapG$ (see \cref{def:refined-snapshot} below). This is a four-dimensional array whose $(a,b,i,j)$-th entry contains the fraction of edges in $\CG$ that go from vertices in bias class $i$ and degree class $a$ to vertices in bias class $j$ and degree class $b$. Note that the refined snapshot $\RSnapG$ is only more granular than the snapshot $\SnapG$ which we actually want to estimate --- in particular, $\SnapG$ can be computed from $\RSnapG$ by ``projecting'' the latter to its third and fourth coordinates. Note also that the snapshot's dimensions are constant, while the refined snapshot's dimensions are polylogarithmic (because the number of layers is $k = \Theta(\log n)$). We will abbreviate $A = \RSnapG$ for convenience and focus on estimating $A$ by inspecting edges between these sets $U_a$.

\paragraph{Estimating edges within each layer.} There is a class of entries $(a,b,i,j)$ of $A$ which are relatively simple to estimate: The ``degree diagonal'' $a=b$, or in other words, entries which correspond to edges \emph{within a single layer}. For this, we can just store the induced subgraph of $\CG_a$ on $U_a$. Looking at these induced subgraphs --- modulo the issue of estimating the bias and degree of sampled vertices --- will be roughly equivalent to the bounded-degree case, essentially because vertices in $U_a$ have small degree in $\CG_a$. However, this is only a small subset of the entries of $A$ which we need to estimate. Given $a \neq b \in [k]$, how can we estimate the ``cross edges'' between layers $a$ and $b$?

\paragraph{Estimating cross edges.} The difficulty with estimating cross-edges, in comparison to the in-layer edges discussed before, is that there can be wide discrepancies between the degrees of the edges' endpoints (both in the global graph and in any particular layer). That is, for $a < b \in [k]$, vertices in degree class $b$ are expected to have a high degree in layer $a$ (as they are expected to have $d_b$ edges and we subsample with probability $C d_a^{-1}$) while the vertices in degree class $a$ are expected to have almost no edges in layer $b$.  This makes the concentration analysis more subtle than the bounded-degree case, but we can still get by with two crucial observations:

\snote{Say here: We could do $U_a$ to $U_b$, but it's easier}

\begin{enumerate}
\item \label{item:alledges} When looking at the layer $a$, we can strengthen the algorithm to remember {\em all} edges in $\CG_a$ that are incident on a vertex in $U_a$ (and not just the induced subgraph on $U_a$).

\item \label{item:extravertices} Secondly, as we mentioned in the previous subsection, as the layer $a$ increases, the maximum number of vertices in $\CG$ in layer $a$ decreases. This means that we can afford for the probability $p_a$ that any particular layer-$a$ vertex is stored in $U_a$ to increase, and for instance when $d_a = \Omega(\sqrt n)$ we can even afford $p_a = 1$ as there are only $O(\sqrt n)$ such vertices.
\end{enumerate}

We claim that, with the above modification, one can estimate cross edges between layers $a$ and $b$ by looking at the graph $\CG_a$ and counting the number of edges in this graph that go from vertices in $U_a$ to vertices in $U_b$. To see why this works, we consider $a = 1$ and two cases for $b$:
\begin{itemize}
\item {\bf When $d_b = \Omega(\sqrt{n})$:} In this case, by \cref{item:extravertices}, $U_b$ is sufficiently large as to contain all the vertices in $\CG$ with degree class $b$. Given the fact that $U_b$ has all these vertices and we have \cref{item:alledges}, whether or not a cross edge $e = (u,v)$, where $u$ has degree class $a$ and $v$ has degree class $b$, is counted depends only on whether or not $u \in U_a$ and whether or not $e \in \CG_a$. The latter is independent across all edges while the former has only a small amount of dependence, as the vertices in degree class $a$ have low degree in $\CG_a$, and does not harm the concentration inequalities too much.

\item {\bf When $d_b = o(\sqrt{n})$:} The argument above will not directly work in this case, as now whether or not a cross edge is counted also depends on whether or not $v \in U_b$. As the vertices in degree class $b$ have high degree in $\CG_a$, this creates a lot of dependencies (depending on $d_b/d_a$) and breaks the concentration bounds. 

What saves us here is that in \cref{item:alledges,item:extravertices}, we sample all edges in $\CG_a$ that are incident on $U_a$ and also are relatively likely to remember any particular vertex in $U_b$. Thus, the number of cross edges between $a$ and $b$ that are remembered in $\CG_a$ is much larger than $O(1)$ (which was the number obtained in the bounded-degree case). Having a larger number of cross edges also means we can also afford to deviate by more without affecting the multiplicative guarantee, and this larger deviation will help us deal with the extra dependencies in this case.
\end{itemize}

\subsubsection{The analysis via windowed averaging and smoothing}\label{sec:overview:analysis}

Even with the modifications above, there is a major problem that we still have to overcome. This problem arises because we do not compute the degrees and biases of the vertices in $\CG$ exactly, and instead {\em estimate} them from the sampled graphs $\CG_a$. These estimates will always be slightly off, and this can wreak havoc in the analysis. As we mentioned above, if for instance the degree of a vertex is at the ``boundary'' of degree classes $a$ and $a+1$, it is impossible to determine with high probability, which entries of the estimate for $A$ will the edges which touch this vertex contribute to --- in some subsamples, the vertex could ``appear to be'' in degree class $a$ and in others it could be in $a+1$. But morally, if the partition is sufficiently fine, these mistakes should not matter too much anyhow because it is possible to ``slightly tweak'' the bias of vertices to put them into neighboring classes without significantly changing the $\mdcut$ value.

The approach we come up with to circumvent these issues, which may be of independent interest, is the following: Instead of trying to estimate entries of $A$ individually, we group them into ``windows'' and estimate the average of all the entries in the window instead. For example, when trying to estimate whether a vertex has degree class $a$, we instead take a windowing parameter $w$ and estimate over all degree classes $\set*{ a - w, \dots, a + w }$.\footnote{One technical issue with this approach is we have to handle the ``degenerate'' cases where, e.g. $a < w$ so the set of allowed degree classes is smaller than $2w+1$. We correspondingly have to weight the entries in the matrix to equalize the contributions of different edges, and this introduces some more potential errors in the algorithm as these ``weighting factors'' for sampled edges can also be estimated incorrectly. We ignore these details in this introduction.} Therefore, our algorithm targets a certain kind of ``windowed'' estimate as opposed to a simple $\ell_1$-estimate. In particular, we define a novel (and incomparable) notion of estimating an array which we call \emph{pointwise smooth estimation} (\cref{def:pointwise}). In this notion of estimation, for an estimate $\hat{A}$ of $A$, as $\delta$ gets arbitrarily small, we do not require that each entry of $\hat{A}$ approaches one fixed value; instead, it must approach an \emph{interval} (which gets arbitrarily narrow as $w$ gets arbitrarily large). Sufficiency of this kind of estimate relies on exactly the kind of informal ``tweaking'' (``smoothing'') analysis we mentioned before.

Informally, the intervals are defined as follows: Consider the $(a,b,i,j)$-th entry of $\hat{A}$, and any edge $e$ from bias class $i'$ and degree class $a'$ to bias class $j'$ and degree class $b'$. Then we declare:

\begin{itemize}
    \item ``Inner'' edges: If $\|(a,b,i,j)-(a',b',i',j')\|_\infty \leq w-1$, then $e$ \emph{must} count in $\hat{A}(a,b,i,j)$.
    \item ``Outer'' edges: If $\|(a,b,i,j)-(a',b',i',j')\|_\infty \geq w+1$, then $e$ \emph{must not} count in $\hat{A}(a,b,i,j)$.
    \item ``Borderline'' edges: If $\|(a,b,i,j)-(a',b',i',j')\|_\infty = w$, then $e$ may or may not count in $\hat{A}(a,b,i,j)$.
\end{itemize}

\input{Files/figures/smoothing.tex}

The key lemma in the analysis then states:

\begin{lemma}[Informal version of \cref{lem:pointwise-impl-smooth,lem:graph-smoothing}]
If an approximation ratio $\alpha$ to $\mdcut$ can be achieved by looking at the snapshot of a graph $\CG$, and $\hat{A}$ is a $(w,\delta)$-pointwise estimate of the refined snapshot $A = \RSnapG$, then $\hat{A}$ can be used to achieve an $(\alpha-\epsilon)$-approximation to $\mdcut$ for $\epsilon = O(\delta (k\ell)^2 + w\lambda + 1/w)$, where $k$, $\ell$, and $\lambda$ are the number of degree classes in $\vecd$, the number of bias classes in $\vect$, and the maximum width of any interval in $\vect$, respectively.
\end{lemma}

The three terms in the error $\epsilon$ come from, respectively: Switching from an entrywise ($\ell_\infty$) error guarantee for the refined snapshot $\hat{A}$ to a global ($\ell_1$) error guarantee for the snapshot; a surface area-to-volume ratio bounding errors from the ``borderline''; and the ``smoothing'' operation which tweaks the biases of vertices. Note that there is a significant interplay of parameters here: To achieve any fixed $\epsilon = O(1)$, we'll have to set $w = O(1/\epsilon)$, $\lambda = O(1/w) = O(1/\epsilon^2)$, and since $k = \log n$, we ultimately need $\delta = O(\epsilon^5/\log^2 n)$.\footnote{Such a guarantee is believable because in the ``gold-standard'' case where we sample \emph{independently} the random edges and look at their biases (and degrees), the deviations would be $O(1/\sqrt n)$ by a Chernoff bound.} Complementarily, \cref{lem:bias-arr-correct} is the key correctness statement, stating that such a pointwise smoothed estimate is achieved by our algorithm. For simplicity, in the algorithm, we only handle a fixed partition degree partition $\vecd$ (i.e., where $d_a = 2d_{a-1}$). To make $\lambda$ arbitrarily small, we use the fact that a bias partition can be subdivided arbitrarily while maintaining the approximation ratio of what can be deduced from the snapshot.

\subsection{Future directions}

Via the trivial reduction from $\mcut$, it is known that for all $\epsilon > 0$, streaming algorithms which $(\frac12+\epsilon)$-approximate $\mdcut$ require $\Omega_\epsilon(n)$ space (cf. \cite{CGS+22-linear-space}). There are a number of interesting alternatives for what could happen between $\omega(\sqrt n)$ and $o(n)$ space. Three main scenarios are:

\begin{description}
    \item[Scenario 1.] ``Algorithms beating \cref{thm:main-thm}'': There are $(\frac12-\epsilon)$-approximations in $\tilde{O}(\sqrt n)$ space.
    \item[Scenario 2.] ``First-order snapshots give optimal sublinear-space algorithms'': Beating $0.483$ requires $\Omega(n)$ space.\footnote{Actually, $0.483$ is not \emph{exactly} the best we can do; rather, we are interested in the best ratio achievable by ``continuous snapshot algorithms'' (see \cref{def:snapshot-alg} below).}
    \item[Scenario 3.] ``Approximation vs. space tradeoff'': Beating $0.483$ can be achieved in $o(n)$ space, but $(\frac12-\epsilon)$-approximations for $\epsilon > 0$ require arbitrarily close to $\Omega(n)$ space for arbitrarily small $\epsilon$.
\end{description}

We are particularly interested in the possibility that for $\mdcut$, one can beat the first-order snapshot algorithms we consider here by estimating instead ``higher-order snapshots'', where by a $t$-order snapshot we mean a histogram of bias patterns among subgraphs with $t$ edges. These seem to correspond to $n^{1-1/(t+1)}$-space algorithms in the bounded-degree case; is it possible to build on the techniques in this paper to estimate $t$-order snapshots within this space bound? Conversely, could there be matching ``dichotomy'' lower bounds --- e.g., for $\mdcut$, could first-order snapshot algorithms be optimal in $o(n^{2/3})$ space? Finally, we mention that Singer~\cite{Sin23} gives oblivious algorithms for $\textsf{Max-}k\textsf{AND}$ beating the $o(\sqrt n)$-space approximation ratios calculated in \cite{BHP+22}; could our snapshot estimation techniques be extended to work for these problems, or even for all finite CSPs?

\subsection*{Outline of the paper}

In \cref{sec:prelims} we introduce some notation and review some background material. The main technical content of the paper is from \cref{sec:reduction} onwards, which can be divided into two independent steps. In the first step, we roughly reduce achieving an approximation factor of $0.483$ on a graph $\CG$ to a problem which we call ``pointwise smoothed estimation'' of a graph. The basic definitions and statements here are in \cref{sec:reduction}; we define ``continuous snapshot algorithms'' (\cref{def:snapshot-alg}), one of which achieves a $0.483$-approximation, and show how they can be simulated given ``pointwise smoothed estimates'' (\cref{def:pointwise}); as mentioned above, the reduction itself is divided between \cref{lem:graph-smoothing,lem:pointwise-impl-smooth}. The proofs of these lemmas are postponed until \cref{sec:graph-smoothing-analysis,sec:pointwise-impl-smooth-analysis}. In the second step, we show how such a ``pointwise smoothed estimate'' can be achieved via a streaming algorithm which implements the ``edge-and-vertex-subsampling'' paradigm outlined above. We present the algorithm in \cref{sec:algo}; the key correctness lemma, \cref{lem:bias-arr-correct}, which states that (under certain niceness conditions) we achieve a pointwise smoothed estimate, is proven in \cref{sec:bias-arr-correct}.

%% file: Files/figures/smoothing.tex
\newcommand{\gridsize}{0.25cm}

\begin{figure}[ht]
\centering
\begin{tikzpicture}
\draw [step=\gridsize, fill=gray!30] (-\gridsize*20, \gridsize*0) grid + (\gridsize*3, \gridsize*3) rectangle (-\gridsize*20, \gridsize*0) ;

\draw [step=\gridsize, fill=darkgray] (-\gridsize*19, \gridsize*1) grid + (\gridsize*1, \gridsize*1) rectangle (-\gridsize*19, \gridsize*1) ;

\draw [step=\gridsize, fill=gray!30] (-\gridsize*2, \gridsize*0) grid + (\gridsize*7, \gridsize*7) rectangle (-\gridsize*2, \gridsize*0);

\draw [step=\gridsize, fill=darkgray] (-\gridsize*1, \gridsize*1) grid + (\gridsize*5, \gridsize*5) rectangle (-\gridsize*1, \gridsize*1);

\draw [step=\gridsize, fill=gray!30] (\gridsize*20, \gridsize*0) grid + (\gridsize*15, \gridsize*15) rectangle (\gridsize*20, \gridsize*0);

\draw [step=\gridsize, fill=darkgray] (\gridsize*21, \gridsize*1) grid + (\gridsize*13, \gridsize*13) rectangle (\gridsize*21, \gridsize*1);

\draw[->, ultra thick] (-\gridsize*15, \gridsize*12) -- node[above] {Increasing the window size $w$} (\gridsize*8, \gridsize*12);

\end{tikzpicture}
\caption{A depiction of how larger windows reduce the ``borderline'' effects (in two dimensions). As $w$ becomes larger and larger, a $w \times w$ rectangle (dark gray) dominates its ``boundary'' (light gray) more and more. Geometrically, a rectangle is two-dimensional while its boundary is ``essentially one-dimensional''. However, for estimating the $\mdcut$ value in a graph, smoothing over size-$w$ windows for large $w$ introduces errors from the use of ``continuity'' results (i.e., \cref{lem:graph-smoothing} below). The right choice of $w$ strikes a balance between these two forces.}
\end{figure}
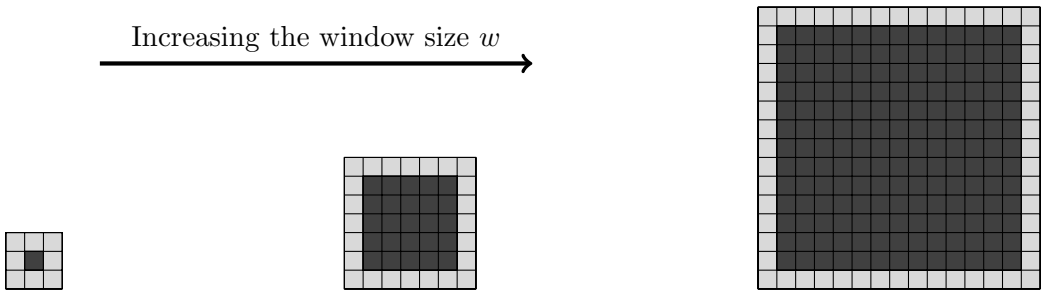

\renewcommand{\gridsize}{0.4cm}
\newcommand{\gridn}{13}
\newcommand{\gridw}{4}

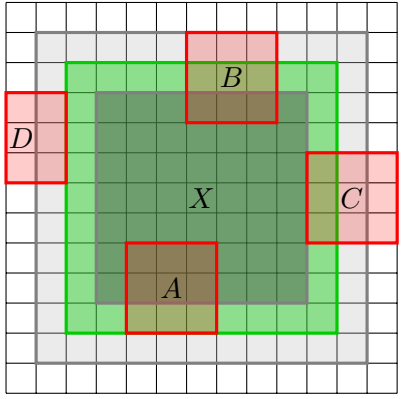
\begin{figure}[ht]

\centering
\begin{tikzpicture}
\draw [black, step=\gridsize] (0,0) grid +(\gridsize*\gridn,\gridsize*\gridn);

\draw [gray, very thick, fill=lightgray, fill opacity=0.3] ({\gridsize*((\gridn-1)/2-(\gridw+1)},{\gridsize*((\gridn-1)/2-(\gridw+1)}) rectangle ({\gridsize*((\gridn+1)/2+(\gridw+1)},{\gridsize*((\gridn+1)/2+(\gridw+1)});
\draw [green!80!black, very thick, fill=green!80!black, fill opacity=0.4] ({\gridsize*((\gridn-1)/2-(\gridw)},{\gridsize*((\gridn-1)/2-(\gridw)}) rectangle ({\gridsize*((\gridn+1)/2+(\gridw)},{\gridsize*((\gridn+1)/2+(\gridw)});
\draw [gray, very thick, fill=darkgray, fill opacity=0.4] ({\gridsize*((\gridn-1)/2-(\gridw-1)},{\gridsize*((\gridn-1)/2-(\gridw-1)}) rectangle ({\gridsize*((\gridn+1)/2+(\gridw-1)},{\gridsize*((\gridn+1)/2+(\gridw-1)});

\node [align=center] at ({\gridsize*(\gridn)/2},{\gridsize*(\gridn)/2}) {$X$};

\draw [red, very thick, fill=red, fill opacity=0.2] ({\gridsize*4},{\gridsize*2}) rectangle ({\gridsize*7},{\gridsize*5});
\node [align=center] at ({\gridsize*5.5},{\gridsize*3.5}) {$A$};

\draw [red, very thick, fill=red, fill opacity=0.2] ({\gridsize*6},{\gridsize*9}) rectangle ({\gridsize*9},{\gridsize*12});
\node [align=center] at ({\gridsize*7.5},{\gridsize*10.5}) {$B$};

\draw [red, very thick, fill=red, fill opacity=0.2] ({\gridsize*10},{\gridsize*5}) rectangle ({\gridsize*13},{\gridsize*8});
\node [align=center] at ({\gridsize*11.5},{\gridsize*6.5}) {$C$};

\draw [red, very thick, fill=red, fill opacity=0.2] ({\gridsize*0},{\gridsize*7}) rectangle ({\gridsize*2},{\gridsize*10});
\node [align=center] at ({\gridsize*0.5},{\gridsize*8.5}) {$D$};

\end{tikzpicture}
\caption{Consider estimating a (two-dimensional) matrix with ``off-by-one'' errors, wherein the mass of each entry may shift to one of $8$ neighboring entries (red boxes). If we estimate an average over a window of size $w=4$ in taxicab distance around an entry $X$ (green rectangle): (i) ``Outer'' entries, such as the one marked $D$, beyond distance $w+1=5$ from the center (light gray) can \emph{never} contribute. (ii) ``Inner'' entries, such as the one marked $A$, within distance at most $w-1=3$ from the center (dark gray) \emph{always} contribute. (iii) ``Borderline'' entries, such as $B$ or $C$, \emph{may or may not contribute}, depending on the specific error pattern.}
\end{figure}

%% file: Files/02-prelims.tex

\section{Preliminaries and notation}\label{sec:prelims}

$[\ell]$ denotes the set of natural numbers $\{1,\ldots,\ell\}$. We use standard asymptotic notation $O(\cdot),o(\cdot)$, etc., with the convention that subscripts (e.g., $f(x,y) = O_y(g(x))$) denote arbitrary dependence in the implicit constant.

\subsection{Matrices and arrays}

For $\ell \in \BN$, we let $\M^\ell \eqdef \BR^{\ell \times \ell}$ denote the space of real $\ell \times \ell$ matrices, $\MN^\ell \subseteq \M^\ell$ the space of matrices with nonnegative entries, and $\MS^\ell \subseteq \M^\ell$ matrices with nonnegative entries summing to $1$. For $i,j \in [\ell]$, $M(i,j)$ denotes the $(i,j)$-th entry of $M$. Given two matrices $M,N \in \M^\ell$, we let $\|M-N\|_1$ and $\|M-N\|_\infty$ denote their entrywise $1$- and $\infty$-norms, respectively, i.e., \[ \|M-N\|_1 \eqdef \sum_{i,j=1}^\ell |M(i,j)-N(i,j)| \text{ and } \|M-N\|_\infty \eqdef \max_{i,j \in [\ell]} |M(i,j) - N(i,j)|. \]

For $k,\ell \in \BN$, we define analogues of this notation for four-dimensional arrays: $\A^{k,\ell} \eqdef \BR^{k \times k \times \ell \times \ell}$ denotes $k \times k \times \ell \times \ell$ arrays, $\AN^{k,\ell}$ nonnegative arrays, and $\AS^{k,\ell}$ nonnegative arrays summing to $1$; we also define $1$- and $\infty$-norms for arrays. We typically use the letters $A$ and $B$ for four-dimensional arrays, and $M$ and $N$ for (two-dimensional) matrices.

\subsection{(Directed) graphs, degrees, biases, and (directed) cuts}

In this paper, we consider directed graphs without self-loops.\footnote{We avoid self-loops because, from the perspective of $\mdcut$ (which we are about to define), a self-loop edge is never satisfied by any assignment and is therefore uninteresting from an algorithmic perspective.} It will be convenient to use two related definitions, ``weighted graphs'' and ``multigraphs'', corresponding to nonnegative real and nonnegative integer edge weights, respectively. In particular, multigraphs will be convenient to encode input to our streaming algorithm, while the more general notion of weighted graphs will be convenient in the analysis.

A \emph{weighted graph} on a vertex-set $V = V(\CG)$ is defined by an \emph{adjacency matrix} $\AdjG \in \MN^{V \times V}$ with zeros on the diagonal. We let $m_\CG = \sum_{u,v \in V} \AdjG(u,v)$ denote the total weight in a weighted graph $\CG$.

Given a vertex $v \in V$ in a weighted graph $\CG$, we define its \emph{out-} and \emph{in-degrees} \[ \dout_\CG(v) \eqdef \sum_{u \in V} \AdjG(v,u) \text{ and }\din_\CG(v) \eqdef \sum_{u \in V} \AdjG(u,v), \] and its \emph{(total) degree} \[ \d_\CG(v) \eqdef \dout_\CG(v) + \din_\CG(v). \] If $\d_\CG(v) = 0$, we say $v$ is \emph{isolated}; otherwise, we define $v$'s \emph{bias} \[ \bias_\CG(v) \eqdef \frac{\dout_\CG(v) - \din_\CG(v)}{\d_\CG(v)} \in [-1,1]. \] Finally, for a ``cut'' $\vecx \in \{0,1\}^V$, we define its \emph{value} in $\CG$ \[ \val_\CG(\vecx) \eqdef \frac1{m_\CG} \sum_{u,v \in V} x_v (1-x_u)\cdot \AdjG(u,v) , \] and the overall $\mdcut$ value of $\CG$ as the maximum value of any cut: \[ \val_\CG(\vecx) \eqdef \max_{\vecx \in \{0,1\}^V} \val_\CG(\vecx). \]

A \emph{multigraph} is a weighted graph where the entries of the adjacency matrix are all integers; equivalently, the graph is specified by a \emph{multiset} of edges $E(\CG) \subseteq \{(u,v) : u \neq v \in V(\CG)\}$, and entries of the matrix equal multiplicies of each edge. Our streaming algorithms will be presented a multigraph with its edges enumerated in arbitrary (adversarial) order, with the goal of achieving an approximation to the $\mdcut$ value of a graph.\footnote{As is standard in the streaming and sketching literature, we will have to assume that the length of the stream $m \leq \poly(n)$. Also, one could consider a more general input model, where we get an arbitrary sequence of edges and (nonnegative real) weights, where the edges are possibly repeated, and the maximum and minimum weights are $w_{\max} \leq \poly(n)$ and $w_{\min} \geq 1/\poly(n)$ respectively. In this model, we can only handle unit weights, but this is essentially without loss of generality because one can multiply by roughly $w_{\max}/(w_{\min} \epsilon)$ and then ``round'' every weight to the nearest integer while preserving the $\mdcut$ value up to $O(\epsilon)$.}

\subsection{Concentration}

We write $\exp(x) = e^{-x}$. We shall need a number of concentration inequalities which operate in different parameter regimes of interest. We list several well-known inequalities as well as some convenient corollaries.

\begin{lemma}[Chernoff upper bound]\label{lem:chernoff-ub}
Let $X_1,\ldots,X_n$ be independent $\{0,1\}$-valued random variables, and let $X = \sum_{i=1}^n X_i$. Then for all $\delta > 0$, \[ \Pr[X \geq (1+\delta) \Exp[X]] \leq \exp(-\delta^2 \Exp[X]/(2+\delta)). \] 
\end{lemma}

\begin{lemma}[Chernoff lower bound]\label{lem:chernoff-lb}
Let $X_1,\ldots,X_n$ be independent $\{0,1\}$-valued random variables, and let $X = \sum_{i=1}^n X_i$. Then for all $0 \leq \delta \leq 1$, \[ \Pr[X \leq (1-\delta) \Exp[X]] \leq \exp(-\delta^2 \Exp[X]/2). \]
\end{lemma}

\begin{corollary}[Two-sided Chernoff bound]\label{lem:chernoff-twosided}
Let $X_1,\ldots,X_n$ be independent $\{0,1\}$-valued random variables, and let $X = \sum_{i=1}^n X_i$. Then for all $0 \leq \delta \leq 1$, \[ \Pr[|X-\Exp[X]| \geq \delta \Exp[X]] \leq 2\exp(-\delta^2 \Exp[X]/3). \]
\end{corollary}

\begin{corollary}[Chernoff upper bound, high deviation form]\label{lem:chernoff-highdev}
Let $X_1,\ldots,X_n$ be independent $\{0,1\}$-valued random variables, and let $X = \sum_{i=1}^n X_i$. Then for all $\eta \geq 3 \Exp[X]$, \[ \Pr[X \geq \eta] \leq \exp(-\eta/8). \]
\end{corollary}

\begin{lemma}[Weighted Chernoff bound {\cite[cf. Theorem 3.3]{CL06}}]
Let $X_1,\ldots,X_n$ be independent $\{0,1\}$-valued random variables. Let $0 < \nu_1,\ldots,\nu_n$ be weights, and let $X = \sum_{i=1}^n \nu_i X_i$. Let $\lambda_0 = \max_i \{\nu_i\}$ and $\lambda_2 = \sum_{i=1}^n \nu_i^2 \Exp[X_i]$. Then for all $\delta > 0$, \[ \Pr[X \geq (1+\delta) \Exp[X]] \leq \exp(-\delta^2\Exp[X]^2/2\lambda_2) \] and \[ \Pr[X \leq (1-\delta) \Exp[X]] \leq \exp(-\delta^2\Exp[X]^2/(2\lambda_2+\lambda_0 \delta \Exp[X]). \]
\end{lemma}

\begin{corollary}[Two-sided weighted Chernoff bound, low weights]\label{lem:weighted-chernoff}
Let $X_1,\ldots,X_n$ be independent $\{0,1\}$-valued random variables. Let $0 < \nu_1,\ldots,\nu_n \leq 1$ be weights, and let $X = \sum_{i=1}^n \nu_i X_i$. Then for all $\delta > 0$, \[ \Pr[|X - \Exp[X]| \geq \delta \Exp[X]] \leq 2 \exp(-\delta^2\Exp[X]/3). \]
\end{corollary}

\begin{proof}
Follows from the previous lemma since $\lambda_2 \leq \sum_{i=1}^n \nu_i \Exp[X_i] = \Exp[X]$ and $\lambda_0 \leq 1$.
\end{proof}

\begin{lemma}[Chebyshev bound]
\label{lem:chebyshev}
Let $X_1,\ldots,X_n$ be random variables, and let $X = \sum_{i=1}^n X_i$. Then for all $\eta > 0$, \[ \Pr[|X - \Exp[X]| \geq \eta] \leq \frac{\Var[X]}{\eta^2}. \]
\end{lemma}

\begin{corollary}[Chebyshev with limited independence]\label{lem:chebyshev-indep}
 Let $X_1,\ldots,X_n$ be random variables such that $0 \leq X_1,\ldots,X_n \leq 1$, and let $X = \sum_{i=1}^n X_i$. Further, suppose that each $X_i$ is independent (pairwise) of all but $D$ variables $\{X_j\}_{j \in [n]}$. Then for all $\eta > 0$, \[ \Pr[|X - \Exp[X]| \geq \eta] \leq \frac{D\cdot \Exp[X]}{\eta^2}. \] In particular, if the variables are pairwise independent, then \[ \Pr[|X - \Exp[X]| \geq \eta] \leq \frac{\Exp[X]}{\eta^2}. \]
\end{corollary}

\begin{proof}
Follows using $\Var[X] = \sum_{i,j=1}^n \Exp[X_iX_j] - \Exp[X_i]\Exp[X_j]$, the limited independence assumption, and the fact that for all $i, j \in [n]$, $\Exp[X_iX_j] \leq \Exp[X_i]$ (using $0 \leq X_i, X_j \leq 1$).
\end{proof}

\subsection{Sparsification for $\mdcut$}

The following lemma is a standard statement about sparsification for the $\mdcut$ problem, which essentially lets us reduce to considering graphs with linearly many edges. We include the proof in \cref{app:prelim-proofs} for completeness.

\begin{lemma}[Linear sparsification preserves $\mdcut$ values]\label{lem:cut-sparsifier}

There exists a universal constant $\Csparse > 0$ such that the following holds. For every $\epssparse \in (0,1)$ and $n,m \in \BN$, suppose $\Csparse n /(\epssparse^2 m) \leq \pspar \leq 1$. Then for every (multi)graph $\CG$ on $n$ vertices with $m$ edges, if we let $\Gsparse$ be the random multigraph resulting from throwing away every edge of $\CG$ independently with probability $1-\pspar$, then with probability $99/100$ over the choice of $\Gsparse$, we have $|\val_\CG - \val_{\Gsparse}| \leq \epssparse$ and $|m_{\Gsparse} - \pspar m| \leq \epssparse \pspar m$.
\end{lemma}

\subsection{$k$-wise independent hash families}

The following definition of a $k$-wise independent hash family will play a role in the algorithm we present in \cref{sec:algo} below.

\begin{definition}
A family of hash functions $\Hash=\{h:[n]\rightarrow [m]\}$ is \emph{$k$-wise independent} if it satisfies the following properties:
\begin{itemize}
 \item For every $x\in [n]$ and $a\in [m]$, and $h \sim \Hash(n,m)$ uniformly, $\Pr[h(x)=a]=\frac{1}{m}$, and
 \item For every distinct $x_1,\dots,x_k\in [n]$, and $h \sim \Hash(n,m)$ uniformly, $h(x_1),\dots,h(x_k)$ are independent random variables.
\end{itemize}
\end{definition}

\begin{lemma}[{\cite{Jof74}, see e.g. \cite[\S2.6]{SSSV23-random-ordering}}]\label{lem:hash-functions}
For every $k,n,m=2^\ell \in \BN$, there exists a family of $k$-wise independent hash functions $\Hash_k = \{h : [n] \rightarrow [m]\}$ such that a uniformly random hash function can be sampled with $O_k(\log n + \log m)$ bits of randomness.
\end{lemma}

%% file: Files/03-snapshots.tex

\section{Reducing $\mdcut$ approximation to snapshot estimation}\label{sec:reduction}

In this section, we develop some machinery to reduce the $\mdcut$ approximation problem for a graph $\CG$ to a problem of estimating a ``pointwise snapshot estimate'' of $\CG$ in the sense of \cref{def:pointwise} below. To begin, we formally define snapshots (\cref{def:snapshot} below). Then, we define various useful notions of smoothing matrices and arrays. Eventually, the statements of the key reduction lemmas are \cref{lem:graph-smoothing,lem:pointwise-impl-smooth}. We also discuss how the measuring the snapshot implies approximation algorithms with factor at least $0.483$ (\cref{lem:fj} below).

\subsection{Snapshots}

Let $\BT^\ell \subseteq \BR^{\ell + 1}$ denote the space of vectors $\vect = (t_0,\ldots,t_\ell)$ such that $t_0 < \cdots < t_\ell$. We call such a vector a \emph{threshold vector of length $\ell$}. Given a threshold vector $\vect \in \BT^\ell$, for any $x \in [t_0,t_\ell]$, we define $x$'s \emph{index} $\Ind^\vect(x)$ (w.r.t. $\vect$) as the unique $i \in [\ell]$ such that $t_{i-1} \leq x < t_i$ (and if $x = t_\ell$ then $\Ind^\vect(x) = \ell$). 

Let $\BT^\ell_{\pm1} \subseteq \BT^\ell$ denote the subset of threshold vectors with $t_0 = -1$ and $t_\ell = 1$. We think of such vectors as defining partitions of biases in graphs. For shorthand, given a weighted graph $\CG$ and a (nonisolated) vertex $v$, we write $\bi^\vect_\CG(v) \eqdef \Ind^\vect(\bias_\CG(v)) \in [\ell]$ for the index representing the ``bias class'' containing $v$, and given a pair of nonisolated vertices $u,v$, we write $\bi^\vect_\CG(u,v) \eqdef (\Ind^\vect(\bias_\CG(u)),\Ind^\vect(\bias_\CG(v))) \in [\ell]^2$ for their pair of bias classes. We say $\vect$ is \emph{$\lambda$-wide} if for every $i \in [\ell]$, $\lambda / 2 \leq t_i - t_{i-1} \leq \lambda$. The width of the partition turns out to factor into the error bound in \cref{lem:graph-smoothing} below. (One should think of $\lambda \approx 1/\ell$.)

\begin{definition}[{Snapshot (``snapshot'' in \cite{SSSV23-random-ordering})}]\label{def:snapshot}
Given a weighted graph $\CG$ and threshold vector $\vect \in \BT^\ell_{\pm1}$, we define the \emph{snapshot} $\SnapG \in \MS^\ell$ by \[ \SnapG(i,j) \eqdef \frac1{m_\CG} \sum_{u,v=1}^n \AdjG(u,v) \1_{\bi^\vect_\CG(u,v) = (i,j)}.\footnote{Note that $\bi_\CG^\vect(u,v)$ is not defined if $v$ or $u$ is isolated. But in either case, $\AdjG(u,v)$ vanishes, so we adopt the convention of discarding these terms.} \]
\end{definition}

In other words, the matrix $\SnapG$ counts the weight fraction of edges in the graph between each pair of bias classes. Note that this is a \emph{normalized} matrix, i.e., its entries sum to $1$, unlike the adjacency matrix $\AdjG$.

Next, we introduce a new version of a snapshot of a graph $\CG$ which also takes into account the degrees of the vertices, which we call the \emph{refined snapshot}.  Suppose we also have a threshold vector $\vecd \in \BT^k$ which partitions vertex degrees in $\CG$, in the following sense: all nonisolated vertices in $\CG$ have degree between $d_0$ and $d_k$. We define similar notations: For nonisolated $v$, we write $\di_\CG^\vecd(v) \eqdef \Ind^\vecd(\deg_\CG(v))$ for the ``degree class'' of $v$. (For notational convenience, if $\deg_\CG(v) = 0$ we will write $\di_\CG^\vecd(v) = -\infty$.) For nonisolated $u,v$, we define:

\begin{equation}
\label{eq:db-ind}
\dbi_\CG^{\vecd,\vect}(u,v) = (\di_\CG^\vecd(u),\di_\CG^\vecd(v),\bi^\vect_\CG(u),\bi^\vect_\CG(v)) . 
\end{equation}

This lets us define:

\begin{definition}[Refined snapshot]\label{def:refined-snapshot}
Given a weighted graph $\CG$ and threshold vectors $\vect \in \BT^\ell_{\pm1}$, $\vecd = (d_0,\ldots,d_k)$, such that every nonisolated vertex $v \in V(\CG)$ has $d_0 \leq \deg_\CG(v) \leq d_k$, we define the \emph{refined snapshot} $\RSnapG \in \AS^{k,\ell}$ by
\begin{equation}
\label{eq:biasdegg}
\RSnapG(a,b,i,j) \eqdef \frac1{m_\CG} \sum_{u,v=1}^n \AdjG(u,v) \1_{\dbi_\CG^{\vecd,\vect}(u,v) = (a,b,i,j)} .
\end{equation}
\end{definition}

This array is only more informative than the snapshot; in particular, the snapshot can be recovered via a ``projection'':

\begin{definition}[Projecting arrays into matrices]\label{def:projection}
Given an array $A \in \A^{k,\ell}$, we define a matrix $\Proj(A) \in \M^\ell$ by projecting onto the third and fourth coordinates, i.e., \[ (\Proj(A))(i,j) = \sum_{a,b=1}^k A(a,b,i,j). \]
\end{definition}

\begin{fact}\label{fact:bias-deg-proj}
Let $\vecd = (d_0,\ldots,d_k) \in \BT^k$ be a degree partition and let $\CG$ be a weighted graph such that all nonisolated vertices have degree between $d_0$ and $d_k$. Then $\Proj(\RSnapG) = \SnapG$.
\end{fact}

\subsection{Defining windows}

To present our formalism for the smoothing analysis, we begin with defining some notations for ``windows'' around entries in (1-dimensional) vectors, (2-dimensional) matrices, and (4-dimensional) arrays. These will correspond to indices to $[\ell]$, $[\ell]^2$, and $[k]^2\times[\ell]^2$, respectively, where $k, \ell \in \BN$. In each case, windows will correspond to a ball of a certain radius in the $\infty$-norm.

More concretely, we make the following definitions:

\begin{definition}[Windows]
\label{def:win}
Suppose $w < \ell \in \BN$. For $i \in [\ell]$, let \[ \Win^{w,\ell}(i) \eqdef \{i' \in [\ell] : \abs*{ i' - i } \leq w\} \] denote the \emph{1-dimensional window} around $i$. For $i,j \in [\ell]$, let \[ \Win^{w,\ell}(i,j) \eqdef \Win^{w,\ell}(i) \times \Win^{w,\ell}(j) = \{(i',j') \in [\ell]^2 : \max\{\abs*{i' - i}, \abs*{j'-j}\} \leq w\} \] denote the \emph{2-dimensional window} around $(i,j)$. Given also $k > w \in \BN$, for $a,b \in [k]$ and $i,j \in [\ell]$, let
\begin{align*}
 \Win^{w,k,\ell}(a,b,i,j) &\eqdef \Win^{w,k}(a,b) \times \Win^{w,\ell}(i,j) \\
 &= \{(a',b',i',j') \in[k]^2 \times [\ell]^2: \max\{\abs*{a-a'},\abs*{b-b'},\abs*{i-i'},\abs*{j-j'}\} \leq w 
 \}
\end{align*}
denote the \emph{$4$-dimensional window} around $(a,b,i,j)$.
\end{definition}

We state various basic but useful facts about windows in \cref{sec:window-facts} below.

\subsection{Smoothed estimates (of matrices)}

We now define what it means to \emph{smooth} a matrix $M \in \M^\ell$ over windows of size $w$.

\begin{definition}[Smoothing matrices]\label{def:smoothing-matrix}
Let $\ell \in \BN$ and $M \in \M^\ell$. For $w < \ell \in \BN$, we define a \emph{smoothed} matrix $M^{\sim w} \in \M^\ell$ by \[ M^{\sim w}(i,j) = \sum_{(i',j') \in \Win^{w,\ell}(i,j)} \nu^{\sim w,\ell}(i',j') \cdot M(i',j'), \] where $\nu^{\sim w,\ell}(i',j') \eqdef 1/|\Win^{w,\ell}(i',j')|$ is a normalization factor.
\end{definition}

Note that the normalization factors $\nu^{\sim w, \ell}(i',j')$ do not depend on the matrix $M$. Informally, their importance is as follows. Suppose $\ell$ is even and $\ell > 2w$. Consider the entries $M(1,1)$ and $M(\ell/2,\ell/2)$. The former will contribute to $\approx w^2$ entries in $M^{\sim w}$ --- in particular, the indices $\{1,\ldots,w+1\} \times \{1,\ldots,w+1\}$ --- while the latter will contribute to $\approx 4w^2$ entries --- in particular, $\{\ell/2-(w+1),\ldots,\ell/2+(w+1)\} \times \{\ell/2-(w+1),\ldots,\ell/2+(w+1)\}$. We will be interested in smoothing snapshots, i.e., $M = \SnapG$, and we would like for the resulting matrices to ``resemble'' snapshots, at least in the sense of still having entries summing to $1$. In particular, the choice of normalization factors ensures that the following holds:

\begin{proposition}[Smoothing preserves entry sum]\label{prop:smoothing-preserves-sum}
For every $M \in \M^\ell$, $\sum_{i,j=1}^\ell M^{\sim w}(i,j) = \sum_{i,j=1}^\ell M(i,j)$. In particular, if $M \in \MS^\ell$ then $M^{\sim w} \in \MS^\ell$.
\end{proposition}

The proof is by a simple double-counting argument; we give it in \cref{sec:double-counting-proofs} below.

\begin{definition}[Smoothed estimates]\label{def:smoothed-est}
    Let $M,\hat{M} \in \BM^\ell$ be matrices, $w \in \BN$, and $\epsilon > 0$. $\hat{M}$ is a \emph{$(w,\epsilon)$-smoothed estimate} of $M$ if $\|\hat{M}-M^{\sim w}\|_1 \leq \epsilon$.
\end{definition}

For us, the sufficiency of this notion of smoothed estimation is given by the following lemma, which roughly states that anything which can be deduced about $\val_\CG$ from $\SnapG$ can also be deduced from a $(w,\epsilon)$-smoothed estimate of $\SnapG$ up to an additive factor $O(\lambda w + \epsilon)$. The lemma is proved in \cref{sec:graph-smoothing-analysis} below.

\begin{lemma}[Smoothed estimate of snapshot suffices]\label{lem:graph-smoothing}
    There exists a universal constant $\Csmooth > 0$ such that the following holds. Let $w < \ell \in \BN$. Let $\vect \in \BT^\ell$ be $\lambda$-wide. Let $\CG$ be any weighted graph with snapshot $M := \SnapG$. Then there exists a weighted graph $\CH$ with snapshot $N := \SnapH$ such that $|\val_\CG - \val_{\CH}| \leq \Csmooth \lambda w$, and $\|N-M^{\sim w}\|_1 \leq \Csmooth \lambda w$.
\end{lemma}

\subsection{Smoothed pointwise estimates (of arrays)}

Unfortunately, we will not be able to directly estimate entries of the smoothed snapshot $\SnapGw$ using our sampling-based algorithm. Roughly, this is because there may be huge discrepancy between degrees of vertices (indeed, from $O(1)$ to $\Omega(n)$), and the subsampling parameters will depend on the degree. So, we reduce to the problem of estimating more refined quantities: smoothed versions of the \emph{array} $\RSnapG$. Our hope is that, if the degree partition is fine enough, these quantities can actually be estimated.

\begin{definition}[Smoothing arrays]
\label{def:smoothing-arrays}
Let $k, \ell \in \BN$ and $A \in \A^{k,\ell}$. For $w < k,\ell \in \BN$, we define a \emph{smoothed} array $A^{\sim w} \in \A^{k,\ell}$ by \[ A^{\sim w}(a,b,i,j) = \sum_{(a',b',i',j') \in \Win^{w,k,\ell}(a,b,i,j)} \nu^{\sim w,k,\ell}(a',b',i',j') \cdot A(a',b',i',j'), \] where $\nu^{\sim w,k,\ell}(a',b',i',j') \eqdef 1/|\Win^{w,k,\ell}(a',b',i',j')|$ is a normalization factor. 
\end{definition}

Let $A \in \AN^{k,\ell}$ be an array with nonnegative entries. For the final step in our reduction, we define a certain notion of a ``pointwise smoothed estimate'' for an array (\cref{def:pointwise} below), and give a statement (\cref{lem:pointwise-impl-smooth} below) which roughly says that such an estimate implies a smoothed estimate for the snapshot in the sense of \cref{def:smoothed-est} above. Such a ``pointwise'' estimate will be what we actually aim to achieve in the algorithm, with $A = \RSnapG$.

The notion of ``pointwise estimate'' relies on the definition of two additional arrays $A^{-w},A^{+w} \in \A_{\geq 0}^{k,\ell}$ which satisfy the inequality $A^{-w} \leq A^{\sim w} \leq A^{+w}$ entrywise and account for ``off-by-one'' errors when estimating $A^{\sim w}$. We first describe these arrays informally as the actual definitions may appear quite technical. Consider an ``estimation'' function $\xi : [k]^2 \times [\ell]^2 \to [k]^2 \times [\ell]^2$ for the graph $\CG$, which takes as input the index $(a,b,i,j)$ of an entry in the array $A$, and outputs a tuple $\xi(a,b,i,j)$ which is promised to differ from $(a,b,i,j)$ by at most $1$ in each entry. (This will correspond to issues with ``borderline'' vertices when $A = \RSnapG$, as described in \cref{sec:overview} above.) Now suppose we tried to estimate the entry $A^{\sim w}(a,b,i,j)$ using the expression \[ \sum_{\xi(a',b',i',j') \in \Win^{w,k,\ell}(a,b,i,j)} \nu^{\sim w,k,\ell}(a',b',i',j') \cdot A(a',b',i',j'). \] The quantities $A^{-w}(a,b,i,j)$ and $A^{+w}(a,b,i,j)$ are lower- and upper-bounds for this expression, respectively, based on the ``worst possible'' function $\xi$.

To be more precise, we define the arrays $A^{-w}, A^{+w} \in \A_{\geq 0}^{k,\ell}$ as follows:

\begin{definition}[Upper- and lower-bound normalization factors]
\label{def:normalize-ul-bounds}
Let $w < k,\ell \in \BN$. Then for $a',b' \in [k], i',j' \in [\ell]$, we define 
\[
\nu^{-w,k,\ell}(a',b',i',j') \eqdef \min_{(a'',b'',i'',j'') \in \Win^{1,k,\ell}(a',b',i',j')} \nu^{\sim w,k,\ell}(a'',b'',i'',j'') ,
\] and 
\[
\nu^{+w,k,\ell}(a',b',i',j') \eqdef \max_{(a'',b'',i'',j'') \in \Win^{1,k,\ell}(a',b',i',j')} \nu^{\sim w,k,\ell}(a'',b'',i'',j'') .
\]
\end{definition}

\begin{definition}[Upper- and lower-bound arrays]
\label{def:arr-ul-bounds}
Let $w < k, \ell \in \BN$ and $A \in \A^{k,\ell}$. We define two arrays $A^{-w},A^{+w} \in \A_{\geq 0}^{k,\ell}$ by defining, for all $a, b \in [k]$ and $i, j \in [\ell]$:
\[ 
A^{-w}(a,b,i,j) \eqdef \sum_{(a',b',i',j') \in \Win^{w-1,k,\ell}(a,b,i,j)} \nu^{-w,k,\ell}(a',b',i',j') A(a',b',i',j') ,
\] 
and 
\[
A^{+w}(a,b,i,j) \eqdef \sum_{(a',b',i',j') \in \Win^{w+1,k,\ell}(a,b,i,j)} \nu^{+w,k,\ell}(a',b',i',j') A(a',b',i',j') . 
\]
\end{definition}

Note that by definition $\nu^{-w,k,\ell}(a',b',i',j') \leq \nu^{\sim w,k,\ell}(a',b',i',j') \leq \nu^{+w,k,\ell}(a',b',i',j')$ (since $(a',b',i',j') \in \Win^{1,k,\ell}(a',b',i',j')$). Hence, (since $A$ has nonnegative entries) we have $A^{-w}(a,b,i,j) \leq A^{\sim w}(a,b,i,j) \leq A^{+w}(a,b,i,j)$, since they sum over windows around $(a,b,i,j)$ of sizes $w-1$, $w$, and $w+1$, respectively, using increasing normalization factors. We also have the following simple lem:
\begin{lemma}
\label{lem:a-bound}
Let $w < k, \ell \in \BN$ and $A \in \AS^{k,\ell}$. For all $a, b \in [k]$ and $i, j \in [\ell]$, we have
\[
A^{-w}(a,b,i,j) \leq A^{+w}(a,b,i,j) \leq 1 .
\] 
\end{lemma}
\begin{proof}
For the first inequality, see above. For the second, note from \cref{def:smoothing-arrays,def:normalize-ul-bounds} that $\nu^{+w,k,\ell}(a',b',i',j') \leq 1$ implying $A^{+w}(a,b,i,j) \leq \sum_{(a',b',i',j') \in \Win^{w+1,k,\ell}(a,b,i,j)} A(a',b',i',j') \leq 1$.
\end{proof}

Now given the definitions of these upper- and lower-bound arrays $A^{+w}$ and $A^{-w}$, we are prepared to define our notion of pointwise estimation:

\begin{definition}[Pointwise smoothed estimates]\label{def:pointwise}
    Let $A, \hat{A} \in \BA^{k,\ell}$ be arrays, $w \in \BN$, and $\delta > 0$. $\hat{A}$ is a \emph{$(w,\delta)$-pointwise smoothed estimate} of $A$ if for every $a,b \in [k],i,j \in [\ell]$, \[  A^{-w}(a,b,i,j) - \delta \leq \hat{A}(a,b,i,j) \leq A^{+w}(a,b,i,j) + \delta. \]
\end{definition}

Such an estimate (where $A$ is the refined snapshot of a graph) is precisely what is achieved by our algorithm; see \cref{lem:bias-arr-correct} below. Correspondingly, we can state our key lemma reducing smoothed estimation of a matrix (in our application, the snapshot) to pointwise smoothed estimation of an array (in our application, the refined snapshot) which projects to that matrix in the sense of \cref{def:projection}:

\begin{lemma}[Pointwise smoothed estimate of array $\implies$ smoothed estimate of matrix]\label{lem:pointwise-impl-smooth}
    There exists a universal constant $\Cwin > 0$ such that the following holds. Let $A, \hat{A} \in \BA^{k,\ell}$ be arrays, $w \in \BN$, and $\delta > 0$. Suppose $\hat{A}$ is a $(w,\delta)$-pointwise smoothed estimate of $A$. Then for $M := \Proj(A)$ and $\hat{M} := \Proj(\hat{A})$, $\hat{M}$ is a $(w,\epsilon)$-smoothed estimate of $M$, for $\epsilon := \delta (k\ell)^2 + \Cwin/w$.
\end{lemma}

This lemma is proved in \cref{sec:pointwise-impl-smooth-analysis} below. We remark that there is a new factor of $(k\ell)^2$ in the error term corresponding to the product of dimensions of the array; this is because we are switching from an $\ell_\infty$-type guarantee to an $\ell_1$-type guarantee. Recall also that in our application (where $A$ is the refined snapshot), we will have $k=\Theta(\log n)$; this means we need $\delta$ to be shrinking as a function of $n$ to apply the lemma, but this turns out to be achievable.

\subsection{Snapshot algorithms}

The final piece of the puzzle is actually identifying how snapshots let us reason about the $\mdcut$ value of a graph. We make the following definition for algorithms which do exactly this:

\begin{definition}[Continuous snapshot algorithms]\label{def:snapshot-alg}
    A \emph{continuous snapshot algorithm} is defined by a threshold vector $\vect \in \BT^\ell$ and a function $\Obl : \M^\ell \to [0,1]$ such that:
    \begin{enumerate}
        \item \emph{Correctness:} For every weighted graph $\CG$, $\Obl(\SnapG) \leq \val_\CG$.
        \item \emph{Continuity:} $|\Obl(M)-\Obl(N)| \leq \|M-N\|_1$.
    \end{enumerate}

    We say $\Obl$ is an \emph{$\alpha$-approximation} if for every weighted graph $\CG$, $\alpha\, \val_\CG \leq \Obl(\SnapG)$.
\end{definition}

Our aim will be to implement any such algorithm as a streaming algorithm via the snapshot estimation machinery we described above. In particular, Feige and Jozeph~\cite{FJ15} showed that one such algorithm achieves a $0.483$-approximation:

\begin{lemma}[{Implied by Feige and Jozeph \cite{FJ15}}]\label{lem:fj}
There exists a constant $\alpha_\FJ \in (0.4835,0.4899)$, $\ell_\FJ \in \BN$, a vector of bias thresholds $\vect_\FJ \in \BT_{\pm 1}^{\ell_\FJ}$, and a vector of probabilities $\vecr_\FJ = (r_1,\ldots,r_\ell) \in [0,1]^{\ell_\FJ}$ such that the function $\Obl(M) := \sum_{i,j=1}^{\ell_\FJ} r_i (1-r_j) M(i,j)$ is a continuous snapshot algorithm achieving a $\alpha_\FJ$-approximation.
\end{lemma}

Though we will not need this fact, we remark that the estimate for $\val_\CG$ given in the lemma corresponds to a simple randomized assignment for $\CG$, namely the so-called ``oblivious'' assignment which assigns each nonisolated vertex $v \in V(\CG)$ to $1$ with probability $r_i$ and $0$ otherwise, where $i = \bi^\vect_\CG(v)$ is its bias class.

One final notion we will need regarding snapshot algorithms is the following. Suppose we have a (continuous) snapshot algorithm $\Obl_\orig$ with a threshold vector $\vect_\orig$, and we want to reduce the \emph{width} of $\vect_\orig$, i.e., the size of the largest interval, yielding a new (continuous) snapshot algorithm with a smaller-width threshold vector. (We will want to do this in order to apply \cref{lem:graph-smoothing}.) Consider greedily packing intervals of width $\leq \lambda$ within each interval in $\vect$ in some canonical way; the resulting partition, which we denote by $\RefinePart(\vect_\orig)$, has $\ell' \leq \ell + 1/\lambda$ intervals. There is a natural corresponding snapshot algorithm, which given an $\ell' \times \ell'$ snapshot collapses it to the corresponding $\ell \times \ell$ snapshot and runs $\Obl$, and remains an $\alpha$-approximation. We denote this algorithm by $\RefineAlg_\lambda(\vect_\orig,\Obl_\orig)$.

%% file: Files/04-algorithm.tex

\section{The algorithm}
\label{sec:algo}

In this section, we present an algorithm that approximates the $\mdcut$ value of an arbitrary multigraph, thereby proving \cref{thm:main-thm} (in its full version \cref{thm:main-thm-formal} below). The algorithm itself has to deal with a number of technical issues, so we begin by outlining the algorithm and giving a number of pointers which hopefully make it easier to digest.

\subsection{Overview}\label{sec:algo-overview}

\paragraph{Informal recap.} At the highest level, we want to be able to sample random edges in a multigraph $\CG$ and estimate the biases of their endpoints. In particular, we fix some global partition $\vect \in \BT^\ell_{\pm 1}$ of the interval $[-1,+1]$ of possible biases into $\ell$ ``bias classes''. We want to estimate the snapshot $\SnapG$ of the input graph $\CG$, which is a matrix whose entries count the number of edges between each pair of bias classes. This can be used to simulate a ``continuous snapshot algorithm'' applied to $\CG$ (\cref{def:snapshot-alg}), and one such algorithm achieves a $0.483$-approximation (\cref{lem:fj}).

To do this, we fix some global partition $d_0 \leq \cdots \leq d_k$ of degrees in $\CG$ dividing the graph's vertices into ``layers'', where layer $a$ is the vertices of degree between $d_{a-1}$ and $d_a$, and aim to estimate the ``refined snapshot'' $\RSnapG$ (\cref{def:refined-snapshot}), which counts the number of edges whose endpoints are in particular bias classes and layers. Towards this, we hope to obtain representative samples $U_1,\ldots,U_k$ of vertices, $U_a$ containing vertices in layer (roughly) $a$, such that (1) we (roughly) know the bias classes of every sampled vertex and (2) we also have representative samples of edges between the layers.

More precisely, when we say ``roughly'' we mean ``up to $w\pm 1$'' allowing us to apply the ``smoothing'' techniques developed in \cref{sec:reduction}. Concretely, this means that we target producing a ``pointwise smoothed estimate'' of the refined snapshot in the sense of \cref{def:pointwise}; the claim that we can achieve this is \cref{lem:bias-arr-correct} below.

\paragraph{Outline of the algorithm.} The outline of our scheme is as follows. For each ``layer'' $a \in \{1,\ldots,k\}$, we fix some probabilities $q_a,p_a \in [0,1]$. We subsample the edges of $\CG$ with probability $q_a$ and then subsample positive-degree vertices in $\CG_a$ with probability $p_a$ to get a subset of vertices, which we denoted $\vStored_a$ in the algorithm. We hope to store the $\CG_a$-neighborhoods of these vertices to get a subset of edges denoted $\eStored_a$. Ideally, we can use $\eStored_a$ to estimate the degrees and biases of vertices in $\vStored_a$, and further, if we see an edge stored in $\eStored_a$ between vertices in $\vStored_a$ and $\vStored_b$, we can count it in the appropriate entry of the refined snapshot. However, note that $\eStored_a$ cannot necessarily store all neighbors for every vertex in $\vStored_a$; in particular, while our edge subsampling ensures that $\vStored_a$ is unlikely to contain vertices of $\CG$-degree much less than $d_a$, it may still contain vertices of $\CG$-degree much \emph{greater} than $d_a$, in which case we can not hope to store all its neighbors.

We write the algorithm in a particularly ``well-factored'' form called a ``sketching algorithm''; while we omit the full definition of such an algorithm, we remark that it essentially means that the algorithm behaves ``independently'' on each edge arriving in the stream. In particular, we can construct a ``sketch'' corresponding to each individual edge $e$, which consists of a pair $(\vStored_a,\eStored_a)$ for each layer $a$, where $\vStored_a$ contains either endpoint with probability $p_a$ independently and $\eStored_a$ with probability $q_a$. Then, we can ``compose'' the sketches for the entire stream together by taking the unions of the sets in these sketches, up to appropriate cutoffs.

It is useful to recall that the subsampling probabilities $q_a$ (for edges) and $p_a$ (for vertices) decrease and increase, respectively, as a function of the layer number $a$.

\paragraph{Some pointers.} It is very helpful to remember that there are two distinct sources of randomness in the algorithm: We ``sparsify'' the graph by subsampling edges, and then subsample which vertices we actually store. In the analysis we will first analyze this edge-subsampling, and then conditioned on certain ``good outcomes'' for the edge-subsampling we will analyze the vertex-subsampling. Correspondingly, it is useful to keep in mind the graph $\CG_a$ consisting of all the edges sampled in layer $a$ (each is sampled w.p. $q_a$) and the set $\CN_a$ of positive-degree vertices in $\CG_a$. (However, the actual streaming algorithm is not necessarily be able to store these sets as they can grow too large; it only stores subsets $\eStored_a$ and $\vStored_a$, respectively.)

When we actually see an edge $e = (u,v)$ in $\CG_a$, how do we know what to do with it (i.e., which entry of the matrix should it contribute to)? Note that $\Exp[\deg_{\CG_a}(v)] = q_a \deg_\CG(v)$. Thus, we can hope to use $q_a^{-1} \deg_{\CG_a}(v)$ (which we call $v$'s ``apparent degree'') as an estimate for $\deg_G(v)$. Roughly, this should work out if $\deg_\CG(v)$ is decently large; for instance, we can use the Chernoff bound to show that w.h.p. $\deg_\CG(v)/2 \leq q_a^{-1} \deg_{\CG_a}(v) \leq 2 \deg_\CG(v)$, so the apparent degree moves by at most $1$ interval relative to the actual degree. The same sort of analysis is necessary to analyze the bias of a vertex and ultimately prove that we get a so-called ``pointwise smoothed estimate''.

Finally, one remaining technical issue stems from the fact that when we want to subsample a set of nonisolated vertices in the subsampled graph $\CG_a$, but we do not know the nonisolated vertices ahead of time. In particular, each time we see a \emph{new} nonisolated vertex we want to toss a $p$-biased coin --- but if we decide \emph{not} to store a vertex, we need to ``remember'' this decision if we happen to see it again. This would be manageable if the algorithm had random access to the results of $n$ biased coin flips, but this model would be somewhat nonstandard. Instead, as in \cite{SSSV23-random-ordering}, we observe that when proving concentration it is sufficient to have four-wise independence in the vertex-subsampling procedure, and thus, we decide whether to store a vertex by plugging it into a previously sampled four-wise independent hash function (see \cref{lem:hash-functions}).

\subsection{Describing the algorithm}

The goal of this section is to prove \cref{thm:main-thm}. We begin by presenting an algorithm (\cref{alg:wrappedsketch} below) for estimating the $\mdcut$ value of a stream of edges corresponding to a graph $\CG$ given an estimate for the number of edges in $\CG$. This algorithm first produces the sketch for the stream containing the sampled vertices and edges (via a ``sketch'' subroutine, \cref{alg:sketch}, and a ``compose'' subroutine, \cref{alg:combine}), and then feeds this sketch into another subroutine, \cref{alg:output}, which estimates degrees and biases among the sampled vertices and counts sampled edges, creates an estimate for the refined snapshot of the graph, and uses this to approximate the graph's $\mdcut$ value. The key correctness lemma, \cref{lem:bias-arr-correct} below, states that this estimate for the refined snapshot is a ``pointwise smoothed estimate'', allowing us to apply the machinery from \cref{sec:reduction}.

We begin with several tables containing definitions of parameters to be used in the algorithms.

\begin{table}[ht]
\begin{center}
\begin{tabularx}{\linewidth}{| l | X | X |} 
 \hline
 \textbf{Notation} & \textbf{Value} & \textbf{Description} \\ \hline\hline
 $w \in \BN$ & $1/\epsilon$ & Size of windows for smoothing \\ \hline
 $\lambda > 0$ & $\epsilon/w$ & Maximum width of intervals in the refinement $\vect$ of $\vect_\orig$ (see next line) \\ \hline
 $\vect \in \BT^\ell$ & $\RefinePart_\lambda(\vect_\orig)$ & Refinement of $\vect_\orig$ into $\ell$ intervals of width at most $\lambda$ \\ \hline
 $\ell \in \BN$ & $\leq \ell_\orig +1/\lambda$ & Number of intervals in $\vect$ \\ \hline
 $\Obl$ & $\RefineAlg_\lambda(\vect_\orig,\Obl_\orig)$ & Corresponding refinement of $\Obl_\orig$ \\ \hline
\end{tabularx}
\caption{Global parameters determined by $\epsilon$ alone.}
\label{table:eps}
\end{center}
\end{table}

\begin{table}[ht]
\begin{center}
\begin{tabularx}{\linewidth}{| l | X | X |}
 \hline
 \textbf{Notation} & \textbf{Value} & \textbf{Description} \\ \hline\hline
 $\mmin \in \BN$ & $\sqrt{n}$ & Minimum number of edges handled \\ \hline
 $\mmax \in \BN$ & $\Csparse n/(\epsilon)^2$ where $\Csparse$ is as in \cref{lem:cut-sparsifier} & Maximum number of edges handled  (see \cref{lem:cut-sparsifier}) \\ \hline
 $k^* \in \BN$ & $6 \log \log n$ & Number of degree intervals before we begin subsampling edges \\ \hline
 $D \in \BN$ & $2^{k^*+w+2}$ & W.h.p. bound on max-degree of ``counted vertices'' in subsampled graphs (parameter for space bound), note that $D = O_\epsilon(\log^6 n)$ \\ \hline
 $\eCutoff$ & $\log^7 n$ & Maximum number of stored neighbors per vertex \\ \hline
\end{tabularx}
\caption{Global parameters determined by $\epsilon$ and $n$.}
\label{table:epsn}
\end{center}
\end{table}

\begin{table}[ht]
\begin{center}
\begin{tabularx}{\linewidth}{| l | X | X |}
 \hline
 \textbf{Notation} & \textbf{Value} & \textbf{Description} \\ \hline\hline
 $k \in \BN$ & $\log (2\hat{m})$ & Number of degree intervals (we will have $\hat{m} \leq \mmax$ and thus, $k = O_\epsilon(\log n))$ \\ \hline
 $\rho > 0$ & $1000 \sqrt{D} \cdot (k\ell)^3 /\epsilon$ & Factor controlling space usage, note that $\rho = O_\epsilon(\log^6 n)$ (assuming $\hat{m} \leq \mmax$) \\ \hline
 $p_0 > 0$ & $\rho/\sqrt{\hat{m}}$ & Factor in vertex-subsampling probability \\ \hline
 $\vCutoff$ & $10 \rho \sqrt{2\hat{m}}$ & Maximum number of stored vertices per layer, note that $\vCutoff = O_\epsilon(\sqrt n \log^6 n)$ (assuming $\hat{m} \leq \mmax$) \\ \hline
\end{tabularx}
\caption{Global parameters determined by $\epsilon$, $n$, and $\hat{m}$.}
\label{table:epsnm}
\end{center}
\end{table}

\begin{table}
\begin{center}
\begin{tabularx}{\linewidth}{| l | l | X |} 
 \hline
 \textbf{Notation} & \textbf{Value} & \textbf{Description} \\ \hline\hline
 $d_a \in \BN$ & $2^a$ & Degree partition. We also define $\vecd = (d_0, \dots, d_k)$ where $d_0 = 1$. \\ \hline
 $q_a \in [0,1]$ & $\min\{2^{k^*-a},1\}$ & Edge-sampling probability \\ \hline
 $p_a \in [0,1]$ & $\min\{p_0 q_a^{-1}, 1\}$ & Vertex-sampling probability \\ \hline
\end{tabularx}
\caption{``Per-layer'' parameters defined for all $a \in [k]$ and determined by $\epsilon$, $n$, and $\hat{m}$.}
\label{table:layer}
\end{center}
\end{table}

\begin{algorithm}[H]
\caption{Our algorithm.}
 
\label{alg:wrappedsketch}
\begin{algorithmic}[1]

\Input A multigraph $\CG$ and an estimate $\hat{m}$ for the number of edges.

\State For all $a \in [k]$, sample a hash function $\pi_a : [n] \to [1/p_a]$ from $\Hash_4(n,1/p_a)$ (see \cref{lem:hash-functions}).\label{line:hash}

\State Initialize an empty sketch $(m,(\vStored_1,\eStored_1),\ldots,(\vStored_k,\eStored_k))$.

\For{each edge $(u,v)$ in the stream}

\State Use \cref{alg:combine} to combine the current sketch with the sketch of $(u,v)$ according to \cref{alg:sketch}, and store the result in the current sketch.

\EndFor

\State Run \cref{alg:output} on the final sketch to obtain the output.
 
\alglinenoNew{algcommon}
\alglinenoPush{algcommon}
\end{algorithmic}
\end{algorithm}

\begin{algorithm}[H]
\caption{Sketch for an input edge.}
\label{alg:sketch}
\begin{algorithmic}[1]

\alglinenoPop{algcommon}

\Input An edge $(u,v)$, an estimate $\hat{m}$ of the total number of edges, and hash functions $\{\pi_a : [n] \to [1/p_a]\}$.

\State Set $m \gets 1$.

\For{$a = 1, \dots, k$}
 
\State Toss a biased coin which is $1$ with probability $q_a$, and let $z$ denote its output. \label{line:forget-edge}
 
\If{$z=1$}
 
\State $\vStored_a = \set*{ v' \mid v' \in \set*{ u, v } \wedge \pi_a(v') = 1 }$. \label{line:vstored}

\State If $\vStored_a \neq \emptyset$, set $\eStored_a \gets \{ (u,v) \}$. Otherwise, set $\eStored_a \gets \emptyset$. \label{line:ecutoff}

\Else

\State $\vStored_a, \eStored_a \gets \emptyset$.

\EndIf
\EndFor

\State Output the sketch $\paren*{ m, \paren*{ \eStored_1,\vStored_1 } , \dots, \paren*{ \eStored_k,\vStored_k } }$.

\alglinenoPush{algcommon}
\end{algorithmic}
\end{algorithm}

\begin{algorithm}[H]
\caption{Combining two sketches.}

\label{alg:combine}
\begin{algorithmic}[1]

\alglinenoPop{algcommon}

\Input Two sketches $\paren{ m^{(1)}, \paren{ \eStored^{(1)}_1,\vStored^{(1)}_1 } , \dots, \paren{ \eStored^{(1)}_k,\vStored^{(1)}_k } }$ and $\paren{ m^{(2)}, \paren{ \eStored^{(2)}_1,\vStored^{(2)}_1 } , \dots, \paren{ \eStored^{(2)}_k,\vStored^{(2)}_k } }$.

\State Set $m \gets m^{(1)} + m^{(2)}$.

\For{$a = 1, \dots, k$}

\parState{%
Set $\vStored_a \gets \bot$ if either $\vStored^{(1)}_a$ or $\vStored^{(2)}_a$ is $\bot$, or if $\card{ \vStored^{(1)}_a \cup \vStored^{(2)}_a } > \vCutoff$. Otherwise, set $\vStored_a \gets \vStored^{(1)}_a \cup \vStored^{(2)}_a$. \label{line:vcombine}
}
 
\parState{%
If $\vStored_a = \bot$, then set $\eStored_a \gets \bot$. Else, set $\eStored_a \gets \eStored^{(1)}_a \cup \eStored^{(2)}_a$ (multiset union). Iteratively remove all edges $(u,v)$ from $\eStored_a$ for which we have for all $v' \in \set*{ u, v }$ that either $v' \notin \vStored_a$ or $\deg_{\eStored_a}(v') > \eCutoff$. \label{line:ecombine}
}
 
\EndFor

\State Output the sketch $\paren*{ m, \paren*{ \eStored_1,\vStored_1 } , \dots, \paren*{ \eStored_k,\vStored_k } }$.

\alglinenoPush{algcommon}
\end{algorithmic}
\end{algorithm}

\begin{algorithm}[H]
\caption{Computing the output from a sketch.}

\label{alg:output}
\begin{algorithmic}[1]

 \alglinenoPop{algcommon}

 \Input A sketch $\paren*{ m, \paren*{ \eStored_1,\vStored_1 } , \dots, \paren*{ \eStored_k,\vStored_k } }$.

 \State If $\exists a \in [k]: \eStored_a = \bot$, then return $\bot$. 

 \State For all $a \in [k]$ and $v \in \vStored_a$, define: \label{line:destbest}
 \[
 \dEst_a(v) = \min\{q_a^{-1} \cdot \deg_{\eStored_a}(v), d_k\} \hspace{0.5cm} \text{and} \hspace{0.5cm} \bEst_a(v) = \bias_{\eStored_a}(v) .
 \]
 \State For all $a \in [k]$ and $i \in [\ell]$, define: \label{line:vest}
 \begin{multline*}
 \vEst_{a,i} = \left\{ v \in \vStored_a \mid \deg_{\eStored_a}(v) < \eCutoff \right. \\
 \left. {} \wedge \Ind^\vecd(\dEst_a(v)) \in \Win^{w,k}(a) \wedge \Ind^\vect(\bEst_a(v)) \in \Win^{w,\ell}(i) \right\} .
 \end{multline*}
 \State For all $a, b \in [k]$ and $i, j \in [\ell]$, define:
 \[
 \AEst_{a,b,i,j} = \sum_{ (u,v) \in \eEst_{a,b,i,j} } \nEst^{w,k,\ell}_{a,b}(u,v) ,
 \]
 where:
 \[
 \eEst_{a,b,i,j} = \eStored_{\min\{a,b\}} \cap (\vEst_{a,i} \times \vEst_{b,j}) ,
 \] 
 and
 \[
 \nEst^{w,k,l}_{a,b}(u,v) = \nu^{\sim w,k,\ell}\paren*{ \Ind^\vecd(\dEst_a(u)),\Ind^\vecd(\dEst_b(v)),\Ind^\vect(\bEst_a(u)),\Ind^\vect(\bEst_b(v)) } .
 \] \label{line:aest}
 
\State Define the array $\hat{A} \in \AN^{k,\ell}$ as: 
\begin{equation}
\label{eq:hata} 
\hat{A}(a,b,i,j) = \frac{\AEst_{a,b,i,j}}{m q_{\min\{a,b\}} p_ap_b}. 
\end{equation}

\State Define $\hat{M} \gets \Proj(\hat{A}) \in \M^\ell$ and output $\hat{v} \gets \Obl(\hat{M}) - \frac{\epsilon}4$.

\alglinenoPush{algcommon}
\end{algorithmic}
\end{algorithm}

\subsection{Analyzing the algorithm}

In this subsection, we prove our main result (\cref{thm:main-thm}) which asserts that we can achieve a $\tilde{O}(\sqrt n)$-space $0.483$-approximation algorithm. The essence of this algorithm is \cref{alg:wrappedsketch}; unfortunately we cannot quite use the latter directly,  because it requires an estimate for the number of edges in $\CG$ which we do not have \emph{a priori}, but this can be fixed with some standard tricks.

In fact, we prove the following more detailed theorem:

\begin{theorem}
\label{thm:main-thm-formal}
For every fixed continuous snapshot algorithm $(\vect_\orig,\Obl_\orig)$ achieving a ratio of $\alpha$, the following holds. Let $\epsilon > 0$. There is a sketching algorithm which, given the edges (in adversarial order) of a multigraph $\CG$ on $n$ vertices and $m$ edges, uses $2^{O(1/\epsilon)} \cdot \sqrt{n} \cdot \log^{O(1)} (n+m)$ space and, with probability at least $9/10$, outputs a value $\hat{v}$ satisfying $(\alpha- \epsilon) \val_\CG \leq \hat{v} \leq \val_\CG$.
\end{theorem}

For correctness of the algorithm, we will need the following lemma, whose proof we defer to later:

\begin{lemma}
\label{lem:bias-arr-correct}
Let $\epsilon > 0$, $n,m,\hat{m} \in \BN$, and $\CG$ be a multigraph on $n$ vertices and $\mmin \leq m \leq \mmax$ edges and such that $m/2 \leq \hat{m} \leq 2m$. Let $A= \RSnapG$ and $\hat{A}$ be as computed by \cref{eq:hata} in the call to \cref{alg:output} in \cref{alg:wrappedsketch}. Letting $w,k,\ell$ be defined as in \cref{table:eps,table:epsnm}, we have with probability at least $99/100$ that $\hat{A}$ is a $(w,\epsilon/(k\ell))^2$-pointwise smoothed estimate of $A$.
\end{lemma}

Returning to the main claim:

\begin{proof}
We assume for simplicity that we know an \emph{a priori} bound $m < n^C$. The algorithm we use will run the following sketches in parallel:

\begin{itemize}
    \item A sketch to count the number of edges $m$ in the input graph.
    \item A ``buffer'' sketch to store up to $\mmin$ edges from the input graph (and $\bot$ if there are more).
    \item For $t \in \set*{ \log_{1.9} \mmin, \dots, \log_{1.9} (n^C) }$, set $\hat{m}(t) \gets 1.9^t$, $\hmsparse(t) \gets \min\paren*{ \mmax, \hat{m}(t) }$, and $\pspar(t) \gets \hmsparse(t) / \hat{m}(t)$. Obtain a graph stream $\hat{\CG}(t)$ by including each edge of $\hat{\CG}$ with probability $\pspar(t)$ independently and run \cref{alg:wrappedsketch} with $\hat{\CG}(t)$ and $\hmsparse(t)$ to get an output.
\end{itemize}

After the stream, if $m \leq \mmin$ we solve the instance (which we have stored) exactly. Otherwise, we choose the unique $t$ such that $\hat{m}(t) \leq m < 1.9 \cdot \hat{m}(t)$ and return the output corresponding to this $t$.

The proof of the theorem now proceeds in several steps.

\paragraph{Space bound.} In the buffer we store at most $\mmin$ edges, which takes $\tilde{O}(\mmin) = \tilde{O}(\sqrt n)$ space. Then, we invoke \cref{alg:wrappedsketch} on $O(C \log n)$ values of $t$, so it suffices to show the bound separately for each value of $t$. For any such $t$, note that each sketch created by \cref{alg:sketch} for a single edge is at most $k=O(\log n)$ copies of the edge and its endpoints, and therefore this sketch has size $O(\log^2 n)$. We now bound the size of the sketch obtained by combining the sketches using \cref{alg:combine}.  which means that it suffices to bound the size of the pairs $\paren*{ \eStored_1,\vStored_1 } , \dots, \paren*{ \eStored_k,\vStored_k }$. For this note that the number $a$ of such pairs is $O(\log n)$ and thus, it suffices to bound the size of each pair. For any such pair $a$, the number of vertices in $\vStored_a$ is at most $\vCutoff$ (due to \cref{line:vcombine}) and the number of edges is at most $\vCutoff \cdot \eCutoff$ (due to \cref{line:ecombine}). Finally, we observe that $\eCutoff = \log^{O(1)}n$ and, since $\hmsparse(t) \leq \mmax$, we have $\vCutoff = 2^{O(1/\epsilon)} \cdot \sqrt{n} \cdot \log^{O(1)} n$.

\paragraph{Reducing to ``correct'' $\hat{m}$.} If $m \leq \mmin$ we solve the instance exactly. Otherwise, we return the output for $t$ satisfying $\hat{m}(t) \leq m < 1.9 \cdot \hat{m}(t)$. For this $t$, we have:
\[
\hmsparse(t) \leq m \cdot \pspar(t) \leq 1.9 \cdot \hmsparse(t) .
\]
Using this and \cref{lem:cut-sparsifier} with $\epssparse = \epsilon$ (note that the conditions of \cref{lem:cut-sparsifier} are satisfied as either $\pspar(t) = 1$ and the lemma is trivial or $\pspar(t) = \mmax/\hat{m}(t) \geq \frac{\Csparse n}{(\epsilon)^2 m}$ ), we have with probability at least $99/100$ that $\abs*{ \val_\CG - \val_{\hat{\CG}(t)} } \leq \epsilon$ and the number of edges in $\hat{\CG}(t)$ is between $\hmsparse(t)/2$ and $2\hmsparse(t)$.

\paragraph{Applying the reduction.} Define $\hat{\CG} := \hat{\CG}(t)$. By \cref{lem:bias-arr-correct}, $\hat{A}$ is a $(w,\epsilon/(k\ell)^2)$-pointwise smoothed estimate of $A = \RSnap_{\hat{\CG},\vecd,\vect}$ with probability $99/100$. Conditioning on this event, we can then apply \cref{lem:pointwise-impl-smooth} to conclude that $\hat{M}$ is a $(w,\epsilon+\Cwin/w)$-smoothed estimate of $M := \RSnap_{\hat{\CG},\vect}$, i.e., $\|\hat{M}-M^{\sim w}\|_1 \leq \epsilon + \Cwin/w$.

By \cref{lem:graph-smoothing}, there exists a weighted graph $\CH$ with snapshot $N := \SnapH$ such that $\|N-M^{\sim w}\|_1 \leq \Csmooth\lambda w$ and $|\val_{\hat{\CG}}-\val_{\CH}| \leq \Csmooth \lambda w$. By the triangle inequality, $\|N-\hat{M}\|_1 \leq \Csmooth \lambda w+\Cwin/w + \epsilon$ and $|\val_\CG-\val_\CH| \leq \Csmooth \lambda w + \epsilon$.

Finally, we assumed $\alpha \val_{\CH} \leq \Obl(N) \leq \val_{\CH}$. By continuity of the snapshot algorithm $\CO$, we have $|\Obl(\hat{M})-\Obl(N)| \leq \Csmooth \lambda w+\Cwin/w + \epsilon$. Re-parametrizing $\epsilon$ finishes the proof.
\end{proof}

\cref{thm:main-thm} follows from \cref{thm:main-thm-formal} by instantiating the latter with the specific continuous snapshot algorithm in \cref{lem:fj}.

%% file: Files/05-alg-proofs.tex

\newcommand{\E}{J}
\newcommand{\e}{e}
\newcommand{\Gk}{\CJ_{[k]}}
\newcommand{\Edges}{\CJ}
\newcommand{\mult}{\mathsf{mult}}

\section{Achieving a pointwise smoothed estimate of the refined snapshot: Proving \cref{lem:bias-arr-correct}}\label{sec:bias-arr-correct}

\newcommand{\Egooddeg}{\CE_{\mathtt{deg}}}
\newcommand{\Egoodbias}{\CE_{\mathtt{bias}}}
\newcommand{\Egoodcounte}{\CE_{\mathtt{count1}}}
\newcommand{\Egoodcountv}{\CE_{\mathtt{count2}}}
\newcommand{\Eofg}{\CE_{\mathtt{of1}}}
\newcommand{\Eofh}{\CE_{\mathtt{of2}}}
\newcommand{\Egood}{\CE_{\mathtt{good}}}

The goal of this section is to prove \cref{lem:bias-arr-correct}, the key correctness lemma of our algorithm, which states that the estimated array $\hat{A}$ is a pointwise smoothed estimate for the refined snapshot $\RSnapG$. More precisely, this means that it lies between the upper- and lower-bound arrays $\RSnapG^{+w},\RSnapG^{-w}$ up to some additive error (see \cref{def:pointwise,def:arr-ul-bounds}). Throughout this section, we work with a fixed $\epsilon, n, m, \hat{m}$ and $\CG$. Besides the variables in \cref{table:eps,table:epsn,table:epsnm,table:layer}, we also define the following notation for each $a \in [k]$:
\begin{itemize}
    \item We associate to each edge in the stream the \emph{index} of its position. For  $\e \in [m]$, we let $(u_\e,v_\e)$ denote the $\e$-th edge in the stream (of edges of $\CG$). We let $\Edges_a$ denote the set of edge indices where $z = 1$ in \cref{line:forget-edge} in the $a$-th execution of the loop.

    \item $\CG_a$ will denote the random variable denoting the subgraph of $\CG$ containing the multiset of edges $\{(u_\e,v_\e) : \e \in \Edges_a\}$.
    
    \item $m_a$ will denote the random variable equal to the number of edges in $\CG_a$.
    \item $\CN_a$ will denote the set-valued random variable that is equal to the set of all non-isolated vertices in $\CG_a$.
\end{itemize}

Note that these random variables are all determined by the randomness in \cref{line:forget-edge}.

For $a \in [k]$, we also define the set-valued random variable $S_a = \{v \in [n] : \pi_a(v) = 1\}$, that is, $S_a$ is the set of all vertices that hash to $1$ under $\pi_a$. Note that $S_a$ is determined by the randomness in \cref{line:hash}. Also, note that the randomness' in \cref{line:forget-edge,line:hash} are independent. 

The plan of attack is to build up to the full proof by defining several ``good'' events, that each happen with high probability and together imply that we get the desired output. We start by analyzing the condition in which \cref{alg:combine} sets either $\vStored_a$ or $\eStored_a$ to $\bot$, for some $a \in [k]$, in the combined sketch (we call this the overflow conditions in \cref{sec:overflow}. In \cref{sec:deg-Ga,sec:bias-Ga}, we analyze the degrees and the biases (respectively) of the vertices the $\CG_a$ and show that they are close to the corresponding values in $\CG$. By \cref{lem:biasdeggood}, this implies that our estimates $\nEst$ in \cref{line:aest} are in the desired range. Finally, in \cref{sec:counting-edges-e,sec:counting-edges-v,sec:final-proof-of-bias-arr-correct}, we build on these lemmas and show that our final output is as desired.

\subsection{Overflow condition}
\label{sec:overflow}

The goal of this section is to define and analyze two events $\Eofg$ and $\Eofh$ such that \cref{alg:combine} never sets either $\vStored_a$ or $\eStored_a$ to $\bot$, for any $a \in [k]$, when they occur. These are defined in \cref{lem:ofg} and \cref{lem:ofh} respectively.

\begin{lemma}
\label{lem:ofg}
We have $\Pr(\Eofg) \geq 999/1000$, where the probability is over the randomness in \cref{line:forget-edge} and $\Eofg$ is defined as follows: For every $a \in [k]$, we have:
\[
m_a \leq \begin{cases}
2 q_a m , &\text{~if~} p_0 \leq q_a \\
5 \rho \sqrt{m} , &\text{~if~} p_0 > q_a 
\end{cases} .
\]
\end{lemma}
\begin{proof}
Note that for all $a \in [k]$, the random variable $m_a$ is the sum of $m$ independent and identically distributed indicator random variables that take the value $1$ with probability $q_a$. For all $a \in [k]$ such that $p_0 \leq q_a$, this fact together with \cref{lem:chernoff-ub} gives:
\[
\Pr\paren*{ m_a \geq 2 q_am } \leq \exp\paren*{ - q_a m / 3 } \leq \exp( - \rho \sqrt{m} / 3) \leq o(1/k) .
\]
Similarly, for all $a \in [k]$ such that $p_0 > q_a$ which implies $5 \rho \sqrt{m} \geq 3 p_0 m \geq 3 q_a m$, \cref{lem:chernoff-highdev} gives:
\[
\Pr\paren*{ m_a \geq 5 \rho \sqrt{m} } \leq \exp\paren*{ - 5 \rho \sqrt{m} / 8 } \leq o(1/k) .
\]
The lemma now follows by combining the two inequalities above and applying a union bound over all $a \in [k]$.
\end{proof}

Next, let $\Edges_1, \ldots, \Edges_k$ be a set of subsets of $[m]$ that may be sampled in \cref{line:forget-edge}. For convenience, we shall often abbreviate $\Gk = \Edges_1,\ldots,\Edges_k$. Observe that $\Gk$ determines whether or not $\Eofg$ occurs. We have:

\begin{lemma}
\label{lem:ofh}
Fix any $\Gk$ such that $\Eofg$ occurs. We have $\Pr\paren*{ \Eofh^{\Gk} } \geq 999/1000$, where the probability is over the randomness in \cref{line:hash} and $\Eofh^{\Gk}$ is defined as follows: For every $a \in [k]$, we have:
\[
\card*{ \CN_a \cap S_a } < \vCutoff .
\]
\end{lemma}
\begin{proof}
Fix $\Gk$ as in the lemma. We show the lemma by showing that for all $a \in [k]$, we have $\Pr\paren*{ \card*{ \CN_a \cap S_a } \geq \vCutoff } \leq o(1/k)$. Indeed, the lemma then follows by a union bound. Fix $a \in [k]$. If $p_0 > q_a$, we have by \cref{lem:ofg} that $\card*{ \CN_a } \leq 2m_a \leq 10 \rho \sqrt{m} < \vCutoff$ and the result follows. Thus, assume that $p_0 \leq q_a$. As $\Eofg$ occurs, we have $\card*{ \CN_a } \leq 2m_a \leq 4q_am$ by \cref{lem:ofg}. For all $v \in \CN_a$, define the indicator random variable $I_v$ to be $1$ if and only if $v \in S_a$. Note that these random variables are pairwise independent as \cref{line:hash} samples a $4$-wise independent hash function. Define $I = \sum_{v \in \CN_a} I_v$ and note that (as $I$ is a sum of pairwise independent indicator random variables) $\Var[I] \leq \Exp[I] \leq 4p_aq_am = 4p_0m$ as $p_0 \leq q_a$. We get:
\begin{align*}
\Pr\paren*{ \card*{ \CN_a \cap S_a } \geq \vCutoff } &= \Pr\paren*{ I \geq \vCutoff } \\
&\leq \Pr\paren*{ I \geq 5p_0m } \tag{As $\vCutoff \geq 10 \rho \sqrt{\hat{m}} = 10 p_0 \hat{m} \geq 5 p_0 m$} \\
&\leq \frac{\Var[I]}{(p_0m)^2} \tag{\cref{lem:chebyshev} and $\Exp[I] \leq 4p_0m$} \\
&\leq \frac{4}{p_0m} \tag{As $\Var[I] \leq 4p_0m$} \\
&\leq \frac{8}{ \rho \sqrt{m} }= o(1/k) .
\end{align*}
\end{proof}

We now formalize our claim above that the overflow never occurs $\mfail$ if $\Eofg$ and $\Eofh$ occur.
\begin{lemma}
\label{lem:overflow}
Fix any $\Gk$ such that $\Eofg$ occurs. If $\Gk$ is sampled in \cref{line:forget-edge} and $\Eofh^{\Gk}$ occurs, then the following hold (with probability $1$):
\begin{enumerate}
 \item \label{item:lem:overflow:1} For all $a \in [k]$, we have $\vStored_a, \eStored_a \neq \bot$.
 \item \label{item:lem:overflow:2} For all $a \in [k]$, we have $\vStored_a = \CN_a \cap S_a$.
 \item \label{item:lem:overflow:3} For all $a \in [k]$ and $u \in \vStored_a$ such that $\deg_{\eStored_a}(u) < \eCutoff$, we have for all $v \in [n]$ that $\mult_{\eStored_a}(u,v) = \mult_{\CG_a}(u,v)$ and $\mult_{\eStored_a}(v,u) =\mult_{\CG_a}(v,u)$.
 \item \label{item:lem:overflow:4} For all $a \in [k]$, $u \in \vStored_a$, we have $\deg_{\eStored_a}(u) < \eCutoff \iff \deg_{\CG_a}(u) < \eCutoff$.
 \item \label{item:lem:overflow:5} For all $a,b \in [k]$ and $i,j\in[\ell]$, we have $\sum_{\e \in \Edges_{\min\{a,b\}}} \1_{(u_\e,v_\e) \in \vEst_{a,i} \times \vEst_{b,j}} = \sum_{(u,v) \in \eStored_{\min\{a,b\}}} \1_{(u,v) \in \vEst_{a,i} \times \vEst_{b,j}}$.
\end{enumerate} 
\end{lemma}
\begin{proof}
We fix the randomness in \cref{line:forget-edge,line:hash} arbitrarily such that both $\Eofg$ and $\Eofh^{\Gk}$ occur and show that \cref{item:lem:overflow:1,item:lem:overflow:2,item:lem:overflow:3,item:lem:overflow:4} hold. Note that fixing the randomness fixes the value of all variables in \cref{alg:sketch,alg:combine}. We show each part in turn.
\begin{enumerate}
\item Fix $a \in [k]$ and note that $\vStored_a = \bot \iff \eStored_A = \bot$. This means that it suffices to show that $\vStored_a \neq \bot$. This is because by \cref{lem:ofh}, we have $\card*{ \CN_a \cap S_a } < \vCutoff$.

\item Follows directly from \cref{item:lem:overflow:1}.
\item We claim that for each $\e \in [m]$, the $\e$-th edge in the stream is placed in $\eStored_a$ iff $\e \in \Edges_a$. As the $\implies$ direction is trivial, we only show the $\impliedby$ direction. We only show the first equation as the proof for the second one is analogous. 
Let $v \in [n]$ be arbitrary such that $(u,v) \in \CG_a$. Note that, as $u \in \vStored_a$, the edge $(u,v)$ was added to $\eStored$ in its own sketch in \cref{alg:sketch}. Thus, the only way $(u,v) \notin \eStored_a$ is if it was removed in \cref{line:ecombine}. However, this would contradict $\deg_{\eStored_a}(u) < \eCutoff$, and we are done.

\item Follows directly from \cref{item:lem:overflow:3}.
\item Follows from \cref{item:lem:overflow:3} and the definition of $\vEst$ on \cref{line:vest}.
\end{enumerate}

\end{proof}

\subsection{Degrees after edge-subsampling}\label{sec:deg-Ga}

For $a \in [k]$ and $v \in [n]$, define a random variable for the {\em apparent degree index} as follows:
\begin{equation}
\label{eq:adi}
\adi^\vecd_a(v) = \Ind^\vecd\paren*{ \min\{q_a^{-1} \deg_{\CG_a}(v), d_k\} } .
\end{equation}
In the above definition, we adopt the convention that the right-hand side is $-\infty$ if $\deg_{\CG_a}(v) = 0$. We remark that this definition is similar to but different from the quantity $\dEst_a(v)$ we defined in \cref{alg:output}, since that was only defined for vertices in $\vStored_a$ and counted only the edges in $\eStored_a$. We begin with the following lemma, which gives a general characterization of degrees of vertices in the graphs $\CG_a$.
\begin{lemma}
\label{lem:deg-Ga}
We have $\Pr\paren*{ \Egooddeg } \geq 999/1000$, where the probability is over the randomness in \cref{line:forget-edge} and $\Egooddeg$ is defined as follows: For all $a \in [k]$ and $v \in [n]$, we have:
\begin{enumerate}
\item \label{item:deg-Ga:1} If $\deg_\CG(v) \geq 2^{a-w-3}$, then we have $\deg_\CG(v)/2 < q_a^{-1} \deg_{\CG_a}(v) < 2\deg_\CG(v)$.
\item \label{item:deg-Ga:2} If $\deg_\CG(v) \leq 2^{a-w-3}$, then we have $q_a^{-1} \deg_{\CG_a}(v) < 2^{a-w-1}$.
\end{enumerate}
\end{lemma}
\begin{proof}
We will upper bound the probability that $\Egooddeg$ does not happen. For this, we fix $a \in [k]$ and $v \in [n]$ and upper bound the probability that one of \cref{item:b-lb,item:b-ub} does not hold for this $a$ and $v$ by $o(1/n^2)$. The proof then follows by a union bound. If $a \leq k^*$, then $\deg_{\CG_a}(v) = \deg_\CG(v)$ and $q_a = 1$ and there is nothing to show. We therefore assume that $a > k^* \implies q_a = 2^{k^* - a}$. Note that $\deg_{\CG_a}(v)$ is a sum of $\deg_\CG(v)$-many independent and identically distributed indicator random variables, each of which is $1$ with probability $q_a$. 

Consider first the case $\deg_\CG(v) \geq 2^{a-w-3}$. This means that $q_a \deg_\CG(v) \geq 2^{k^*-w-3} \geq 5 0\log n$. In this case, we use \cref{lem:chernoff-ub,lem:chernoff-lb} to get:
\begin{align*}
\Pr\paren*{ \deg_{\CG_a}(v) \geq 2 q_a \deg_\CG(v) } &\leq \exp\paren*{ - q_a \deg_\CG(v)/3 } = o(1/n^2) . \\
\Pr\paren*{ \deg_{\CG_a}(v) \leq q_a \deg_\CG(v) / 2 } &\leq \exp\paren*{ - q_a \deg_\CG(v)/8 } = o(1/n^2) .
\end{align*}

Now, consider the case $\deg_\CG(v) \leq 2^{a-w-3}$. In this case, we use \cref{lem:chernoff-highdev} to get:
\[
\Pr\paren*{ \deg_{\CG_a}(v) \geq q_a 2^{a-w-1} } \leq \exp\paren*{ - q_a 2^{a-w-4} } = o(1/n^2) .
\]
Adding the three bounds gives the result.
\end{proof}
We can deduce several nice properties about the random variable $\adi^\vecd_a(v)$ defined in \cref{eq:adi} when $\Egooddeg$ happens.

\begin{lemma}
\label{lem:deg-Ga-prop}
If $\Egooddeg$ occurs, the following hold for all $a \in [k]$ and $v \in [n]$ (with probability $1$):
\begin{enumerate}
\item \label{item:b-cutoff} If $\deg_\CG(v) \leq 2^{a+w+1}$, then we have $\deg_{\CG_a}(v) < \eCutoff$.
\item \label{item:b-lb} If $\di^\vecd_\CG(v) \geq a-w-1$, then we have $\adi^\vecd_a(v) \geq \di^\vecd_\CG(v)-1$.
\item \label{item:b-ub} We have $\adi^\vecd_a(v) \leq \max(\di^\vecd_\CG(v)+1, a-w-1)$.
\end{enumerate}
\end{lemma}
\begin{proof}
Fix any $\Gk$ such that $\Egooddeg$ occurs. Also, fix $a$ and $v$. Fixing $\Gk$ also fixes the value of $\deg_{\CG_a}(v)$ and $\adi^\vecd_a(v)$. We prove each part in turn:
\begin{enumerate}
\item If not, we have from \cref{item:deg-Ga:2} of \cref{lem:deg-Ga} that $\deg_\CG(v) > 2^{a-w-3}$. \cref{item:deg-Ga:1} of \cref{lem:deg-Ga} then implies the following contradiction:
\[
\eCutoff \leq \deg_{\CG_a}(v) \leq 2 q_a \deg_\CG(v) \leq q_a 2^{a+w+2} \leq 2^{k^*+w+2} .
\]
\item Note that $\di^\vecd_\CG(v) \geq a-w-1$ implies that $\deg_\CG(v) \geq 2^{a-w-2}$. \cref{item:deg-Ga:1} of \cref{lem:deg-Ga} implies $q_a^{-1} \deg_{\CG_a}(v) > \deg_\CG(v)/2$. As $d_k \geq \deg_\CG(v)$ as well, it follows from \cref{eq:adi} that $\adi^\vecd_a(v) \geq \di^\vecd_\CG(v) - 1$. 
\item If $q_a^{-1} \deg_{\CG_a}(v) \leq 2^{a-w-1}$, there is nothing to show. Otherwise, we have from \cref{item:deg-Ga:2} of \cref{lem:deg-Ga} that $\deg_\CG(v) > 2^{a-w-3}$. Now, using \cref{item:deg-Ga:1} of \cref{lem:deg-Ga}, we get $q_a^{-1} \deg_{\CG_a}(v) < 2\deg_\CG(v)$. It follows that $\adi^\vecd_a(v) \leq \di^\vecd_\CG(v)+1$.
\end{enumerate}
\end{proof}

Next, for $a \in [k]$, we define the following set-valued random variable that is determined by the randomness in \cref{line:forget-edge}.
\begin{equation}
\label{eq:vhat}
\hat{V}_a = \set*{v \in [n] \mid \deg_{\CG_a}(v) < \eCutoff \wedge \adi^\vecd_a(v) \in \Win^{w,k}(a) } .
\end{equation}
We also define the following sets:
\begin{equation}
\label{eq:V}
\begin{split}
V^-_a &= \set*{v \in [n] \mid \di^\vecd_\CG(v) \in \Win^{w-1,k}(a) } . \\
V^+_a &= \set*{v \in [n] \mid \di^\vecd_\CG(v) \in \Win^{w+1,k}(a) } .
\end{split}
\end{equation}
Recall that we defined the parameter $D = 2^{k^* + w + 2}$. These definitions satisfy:
\begin{lemma}
\label{lem:egooddeg}
If $\Egooddeg$ occurs, the following hold for all $a \in [k]$ (with probability $1$):
\begin{enumerate}
\item \label{item:egooddeg:1} We have $V^-_a \subseteq \hat{V}_a \subseteq V^+_a$.
\item \label{item:egooddeg:2} For all $a' \geq a \in [k]$ and all $v \in V^+_{a'}$, we have $\deg_{\CG_a}(v) \leq 2D q_a / q_{a'}$.
\item \label{item:egooddeg:3} For all $v \in V^+_a$, we have $\abs*{ \adi^\vecd_a(v) - \di^\vecd_\CG(v) } \leq 1$.
\end{enumerate}
\end{lemma}
\begin{proof}
Fix any $\Gk$ such that $\Egooddeg$ occurs. Fixing $\Gk$ also fixes the value of $\deg_{\CG_a}(v)$ and $\adi^\vecd_a(v)$ for all $a \in [k]$ and $v \in [n]$. Fix $a \in [k]$ as in the lemma. We prove each part in turn.
\begin{enumerate}
\item We first show that $V^-_a \subseteq \hat{V}_a$. Let $v \in V^-_a$ be arbitrary. By \cref{eq:V}, we have $\di^\vecd_\CG(v) \in \Win^{w-1,k}(a)$. Using \cref{item:b-cutoff,item:b-lb,item:b-ub} of \cref{lem:deg-Ga-prop}, we get $\deg_{\CG_a}(v) < \eCutoff$ and $\adi^\vecd_a(v) \geq \di^\vecd_\CG(v)-1 \geq a - w$ and $\adi^\vecd_a(v) \leq \di^\vecd_\CG(v)+1 \leq a + w$. It follows from \cref{eq:vhat} that $v \in \hat{V}_a$. As $v \in V^-_a$ was arbitrary, we have $V^-_a \subseteq \hat{V}_a$, as desired.

We now show that $\hat{V}_a \subseteq V^+_a$ in the contrapositive. Let $v \in [n] \setminus V^+_a$ be arbitrary. By \cref{eq:V}, we have $\di^\vecd_\CG(v) \notin \Win^{w+1,k}(a)$ Thus, either $\di^\vecd_\CG(v) < a-w-1$ or $\di^\vecd_\CG(v) > a+w+1$. Using \cref{item:b-ub} of \cref{lem:deg-Ga-prop} in the former case and \cref{item:b-lb} in the latter, we get that either $\adi^\vecd_a(v) \leq a-w-1 < a-w$ or $\adi^\vecd_a(v) \geq \di^\vecd_\CG(v)-1 > a + w$. In either case, we have from \cref{eq:vhat} that $v \in [n] \setminus \hat{V}_a$. As $v \in [n] \setminus V^+_a$ was arbitrary, the result follows.

\item As $v \in V^+_{a'}$, we have $\di^\vecd_\CG(v) \in \Win^{w+1,k}(a')$. It follows that $\di^\vecd_\CG(v) \geq a'-w-1 \geq a-w-1$. From \cref{item:b-ub} of \cref{lem:deg-Ga-prop}, we get $\adi^\vecd_a(v) \leq \di^\vecd_\CG(v)+1 \leq a' + w + 2$. From \cref{eq:adi}, we get that $\min\{q_a^{-1} \deg_{\CG_a}(v), d_k\} \leq 2^{a' + w + 2}$. As \cref{item:deg-Ga:1} of \cref{lem:deg-Ga} implies that $q_a^{-1} \deg_{\CG_a}(v) < 2\deg_\CG(v) \leq 2d_k$, we can continue as $q_a^{-1} \deg_{\CG_a}(v) \leq 2^{a' + w + 3}$. This means:
\[
\deg_{\CG_a}(v) \leq 2^{a' + w + 3} q_a \leq 2D q_a / q_{a'} .
\]

\item As $v \in V^+_a$, we have $\di^\vecd_\CG(v) \in \Win^{w+1,k}(a)$. The result follows from \cref{item:b-ub,item:b-lb} of \cref{lem:deg-Ga-prop}.
\end{enumerate}
\end{proof}

\subsection{Biases after edge-subsampling}\label{sec:bias-Ga}

For $a \in [k]$ and $v \in [n]$, we define a random variable for the {\em apparent bias index} as follows:
\begin{equation}
\label{eq:abi}
\abi^\vect_a(v) = \bi^\vect_{\CG_a}(v) .
\end{equation}
(The right hand side is undefined if $\deg_{\CG_a}(v) = 0$.) Note that this random variable is determined by $\CG_a$ and therefore, by the randomness in \cref{line:forget-edge}.

\begin{lemma}
\label{lem:bias-Ga}
We have $\Pr\paren*{ \Egoodbias } \geq 999/1000$, where the probability is over the randomness in \cref{line:forget-edge} and $\Egoodbias$ is defined as follows: For all $a \in [k]$ and $v \in V^+_a$, we have:
\begin{enumerate}
\item \label{item:bias-Ga:1} We have:
\[
\abs*{ \deg_{\CG_a}(v) - q_a \deg_\CG(v) } \leq \frac{\epsbias}{12} \cdot q_a \deg_\CG(v) .
\]
\item \label{item:bias-Ga:2} If $\dout_\CG(v) \geq \deg_\CG(v)/2$, then we have:
\[
\abs*{ \dout_{\CG_a}(v) - q_a \dout_\CG(v) } \leq \frac{\epsbias}{12} \cdot q_a \dout_\CG(v) .
\]
\item \label{item:bias-Ga:3} If $\din_\CG(v) \geq \deg_\CG(v)/2$, then we have:
\[
\abs*{ \din_{\CG_a}(v) - q_a \din_\CG(v) } \leq \frac{\epsbias}{12} \cdot q_a \din_\CG(v) .
\]
\end{enumerate}
\end{lemma}
\begin{proof}
We will upper bound the probability that $\Egoodbias$ does not happen. For this, we fix $a \in [k]$ and $v \in V^+_a$ and upper bound the probability that one of \cref{item:bias-Ga:1,item:bias-Ga:2,item:bias-Ga:3} does not hold for this $a$ and $v$ by $o(1/n^2)$. The proof then follows by a union bound. If $a \leq k^*$, then $\CG_a = \CG$ and $q_a = 1$ and there is nothing to show. We therefore assume that $a > k^* \implies q_a = 2^{k^* - a}$ and show \cref{item:bias-Ga:2} as the proofs of \cref{item:bias-Ga:1,item:bias-Ga:3} is analogous. Note that $\dout_{\CG_a}(v)$ is a sum of $\dout_\CG(v)$-many independent and identically distributed indicator random variables, each of which is $1$ with probability $q_a$. Thus, we have from \cref{lem:chernoff-twosided} that:
\begin{align*}
\Pr\paren*{ \abs*{ \dout_{\CG_a}(v) - q_a \dout_\CG(v) } \geq \frac{\epsbias}{12} \cdot q_a \dout_\CG(v) } &\leq 2\exp\paren*{ -\frac{\epsbias^2}{500} \cdot q_a \dout_\CG(v) } \\
&\leq 2\exp\paren*{ -\frac{\epsbias^2}{1000} \cdot q_a \deg_\CG(v) } \\
&\leq 2\exp\paren*{ -\frac{\epsbias^2}{1000} \cdot 2^{k^* - w - 2} } \tag{\cref{eq:V} implies $\deg_\CG(v) \geq 2^{a-w-2}$} \\
&= o(1/n^2) .
\end{align*}
\end{proof}

\begin{lemma}
\label{lem:bias-Ga-prop}
If $\Egoodbias$ occurs, for all $a \in [k]$ and $v \in V^+_a$, we have $|\abi^\vect_a(v)-\bi^\vect_\CG(v)|\leq 1$ (with probability $1$).
\end{lemma}
\begin{proof}
Fix any $\Gk$ such that $\Egoodbias$ occurs. Also, fix $a$ and $v$. Fixing $\Gk$ also fixes the value of $\bi^\vect_\CG(v)$ and $\abi^\vect_a(v)$. We assume that $\dout_\CG(v) \geq \deg_\CG(v)/2$ as the proof for the case $\din_\CG(v) \geq \deg_\CG(v)/2$ is analogous. By \cref{item:bias-Ga:1,item:bias-Ga:2} of \cref{lem:bias-Ga}, we have:
\begin{align*}
1 - \frac{\epsbias}{12} \leq \frac{ \dout_{\CG_a}(v) }{ q_a \dout_\CG(v) } \leq 1 + \frac{\epsbias}{12} . \\
1 - \frac{\epsbias}{12} \leq \frac{ \deg_{\CG_a}(v) }{ q_a \deg_\CG(v) } \leq 1 + \frac{\epsbias}{12} .
\end{align*}
It follows that
\[
1 - \frac{\epsbias}{4} \leq \frac{ \dout_{\CG_a}(v) / \dout_\CG(v) }{ \deg_{\CG_a}(v) / \deg_\CG(v) } \leq 1 + \frac{\epsbias}{4} ,
\]
which implies
\[
\frac{ \dout_\CG(v) }{ \deg_\CG(v) } - \frac{\epsbias}{4} \leq \frac{ \dout_{\CG_a}(v) }{ \deg_{\CG_a}(v) } \leq \frac{ \dout_\CG(v) }{ \deg_\CG(v) } + \frac{\epsbias}{4} .
\]
As $\bias = 2\dout/\deg - 1$, we get:
\[
\bias_\CG(v) - \frac{\epsbias}{2} \leq \bias_{\CG_a}(v) \leq \bias_\CG(v) + \frac{\epsbias}{2} .
\]
As consecutive entries of $\vect$ are at least $\epsbias/2$ apart, this means that $|\abi^\vect_a(v)-\bi^\vect_\CG(v)| \leq 1$ as desired.
\end{proof}

Next, for $a \in [k]$ and $i \in [\ell]$, we define the following set-valued random variable that is determined by the randomness in \cref{line:forget-edge}.
\begin{equation}
\label{eq:vhati}
\hat{V}_{a,i} = \set*{ v \in \hat{V}_a \mid \abi^\vect_a(v) \in \Win^{w,\ell}(i) } .
\end{equation}
We also define the following sets for all $a \in [k]$ and $i \in [\ell]$:
\begin{equation}
\label{eq:Vi}
\begin{split}
V^-_{a,i} &= \set*{ v \in V^-_a \mid \bi^\vect_\CG(v) \in \Win^{w-1,\ell}(i) } . \\
V^+_{a,i} &= \set*{ v \in V^+_a \mid \bi^\vect_\CG(v) \in \Win^{w+1,\ell}(i) } .
\end{split}
\end{equation}
These definitions satisfy:
\begin{lemma}
\label{lem:egoodbias}
If $\Egooddeg$ and $\Egoodbias$ occur, for all $a \in [k]$ and $i \in [\ell]$, we have $V^-_{a,i} \subseteq \hat{V}_{a,i} \subseteq V^+_{a,i}$ (with probability $1$).
\end{lemma}
\begin{proof}
Follows immediately from \cref{item:egooddeg:1} of \cref{lem:egooddeg}, \cref{lem:bias-Ga-prop}, and the definition of $\Win$.
\end{proof}

Recall the set $\vEst_{a,i}$ defined in \cref{line:vest} and the sets $\CN_a$ and $S_a$ defined at the beginning of \cref{sec:bias-arr-correct}. We have:

\begin{lemma}
\label{lem:biasdeggood}
Fix any $\Gk$ such that all of $\Eofg$, $\Egooddeg$, and $\Egoodbias$ occur. If $\Gk$ is sampled in \cref{line:forget-edge} and $\Eofh^{\Gk}$ occurs, then the following hold (with probability $1$):
\begin{enumerate}
\item \label{item:biasdeggood:1} For all $a \in [k]$ and $i \in [\ell]$, we have $\hat{V}_{a,i} \cap S_a = \vEst_{a,i}$ .
\item \label{item:biasdeggood:2} Let $a,b \in [k]$ and $i,j \in [\ell]$. For all $u \in \vEst_{a,i}$ and $v \in \vEst_{b,j}$, it holds that:
\[
\nu^{-w,k,\ell}\paren*{ \dbi^{\vecd,\vect}_\CG(u,v) } \leq \nEst^{w,k,\ell}_{a,b}(u,v) \leq \nu^{+w,k,\ell}\paren*{ \dbi^{\vecd,\vect}_\CG(u,v) } .
\]
\end{enumerate}
\end{lemma}
\begin{proof}
We fix the randomness in \cref{line:forget-edge,line:hash} arbitrarily such that all the events in the lemma occur. Note that fixing the randomness fixes the value of all variables in \cref{alg:sketch,alg:combine} and also fixes the random variables defined above. We prove each part in turn.

\begin{enumerate}
\item Fix $a \in [k]$ and $i \in [\ell]$. Observe that, for any $v \in [n]$, we have the following equivalences:
\begin{align*}
&v \in \hat{V}_{a,i} \cap S_a \\
&\hspace{0.5cm}\iff v \in \hat{V}_{a,i} \wedge v \in S_a \\
&\hspace{0.5cm}\iff v \in \hat{V}_a \wedge \abi^\vect_a(v) \in \Win^{w,\ell}(i) \wedge v \in S_a \tag{\cref{eq:vhati}} \\
&\hspace{0.5cm}\iff \deg_{\CG_a}(v) < \eCutoff \wedge \adi^\vecd_a(v) \in \Win^{w,k}(a) \wedge \abi^\vect_a(v) \in \Win^{w,\ell}(i) \wedge v \in S_a \tag{\cref{eq:vhat}} \\
&\hspace{0.5cm}\iff \deg_{\CG_a}(v) < \eCutoff \wedge \adi^\vecd_a(v) \in \Win^{w,k}(a) \wedge \abi^\vect_a(v) \in \Win^{w,\ell}(i) \wedge v \in \CN_a \cap S_a \tag{As $\adi^\vecd_a(v) \in \Win^{w,k}(a)$ implies by \cref{eq:adi} that $\deg_{\CG_a}(v) > 0 \implies v \in \CN_a$} \\
&\hspace{0.5cm}\iff \deg_{\CG_a}(v) < \eCutoff \wedge \adi^\vecd_a(v) \in \Win^{w,k}(a) \wedge \abi^\vect_a(v) \in \Win^{w,\ell}(i) \wedge v \in \vStored_a \tag{\cref{item:lem:overflow:2} of \cref{lem:overflow}} \\
&\hspace{0.5cm}\iff \deg_{\eStored_a}(v) < \eCutoff \wedge \adi^\vecd_a(v) \in \Win^{w,k}(a) \wedge \abi^\vect_a(v) \in \Win^{w,\ell}(i) \\
&\hspace{12.5cm} {} \wedge v \in \vStored_a \tag{\cref{item:lem:overflow:4} of \cref{lem:overflow}} .
\end{align*}
To continue, note that for any $v \in \vStored_a$ such that $\deg_{\eStored_a}(v) < \eCutoff$, we have from \cref{item:lem:overflow:3} of \cref{lem:overflow} that $\din_{\eStored_a}(v) = \din_{\CG_a}(v)$ and $\dout_{\eStored_a}(v) = \dout_{\CG_a}(v)$. Conclude from \cref{line:destbest,eq:adi,eq:abi} that for any such $v$, we have $\adi^\vecd_a(v) = \Ind^\vecd(\dEst_a(v))$ and $\abi^\vect_a(v) = \Ind^\vect(\bEst_a(v))$. We can now continue as:
\begin{align*}
v \in \hat{V}_{a,i} \cap S_a &\iff \deg_{\eStored_a}(v) < \eCutoff \wedge \Ind^\vecd(\dEst_a(v)) \in \Win^{w,k}(a) \\
&\hspace{4cm} \wedge \Ind^\vect(\bEst_a(v)) \in \Win^{w,\ell}(i) \wedge v \in \vStored_a \\
&\iff v \in \vEst_{a,i} \tag{\cref{line:vest}} .
\end{align*}
\item As $u \in \vEst_{a,i}$, we have by \cref{line:vest} that $u \in \vStored_a$ and $\deg_{\eStored_a}(u) < \eCutoff$. From \cref{item:lem:overflow:3} of \cref{lem:overflow}, this means that $\din_{\eStored_a}(u) = \din_{\CG_a}(u)$ and $\dout_{\eStored_a}(u) = \dout_{\CG_a}(u)$. Using \cref{line:destbest,eq:adi,eq:abi}, this means that $\adi^\vecd_a(u) = \Ind^\vecd(\dEst_a(u))$ and $\abi^\vect_a(u) = \Ind^\vect(\bEst_a(u))$. As similar arguments apply for $v$, we get from \cref{line:aest} that:
\[
\nEst^{w,k,l}_{a,b}(u,v) = \nu^{\sim w,k,\ell}\paren*{ \adi^\vecd_a(u),\adi^\vecd_b(v),\abi^\vect_a(u),\abi^\vect_b(v) }
\]
Now, as we have \cref{def:normalize-ul-bounds}, it suffices to show that:
\[
\paren*{ \adi^\vecd_a(u),\adi^\vecd_b(v),\abi^\vect_a(u),\abi^\vect_b(v) } \in \Win^{1,k,\ell}\paren*{ \dbi^{\vecd,\vect}_\CG(u,v) } .
\]
From \cref{def:win,eq:db-ind}, this follows if we show that:
\begin{align*}
\abs*{ \adi^\vecd_a(u) - \di_\CG^\vecd(u) }, \abs*{ \adi^\vecd_b(v) - \di_\CG^\vecd(v) } &\leq 1 , \\
\abs*{ \abi^\vect_a(u) - \bi^\vect_\CG(u) }, \abs*{ \abi^\vect_b(v) - \bi^\vect_\CG(v) } &\leq 1 .
\end{align*}
This first inequality is due to \cref{item:b-lb,item:b-ub} of \cref{lem:deg-Ga-prop} while the second inequality is due to \cref{lem:bias-Ga-prop} (note that \cref{item:biasdeggood:1} implies that $u \in \hat{V}_{a,i}$ and $v \in \hat{V}_{b,j}$ and we have \cref{lem:egoodbias}).

\end{enumerate}
\end{proof}

\subsection{Counting target edges after edge-subsampling}\label{sec:counting-edges-e}

For $a, b \in [k]$ and $i, j \in [\ell]$, define:
\begin{equation}
\label{eq:E}
E^-_{a,b,i,j} = \{ e \in [m] : (u_e,v_e) \in V^-_{a,i}\times V^-_{b,j}\} \hspace{.5cm} \text{and} \hspace{.5cm} E^+_{a,b,i,j} = \{e \in [m]: (u_e,v_e) \in V^+_{a,i} \times V^+_{b,j}\} .
\end{equation}

Recall from the statement of \cref{lem:bias-arr-correct} the notation $A = \RSnapG$. Next, for $a, b \in [k]$ and $i, j \in [\ell]$, we define the random variables:

\begin{equation}
\label{eq:Y}
\begin{split}
Y^-_{a,b,i,j} &= \sum_{\e \in E^-_{a,b,i,j}} \1_{\e \in \Edges_{\min\{a,b\}}} \cdot \nu^{-w,k,\ell}\paren*{ \dbi^{\vecd,\vect}_\CG(u_\e,v_\e) } . \\
Y^+_{a,b,i,j} &= \sum_{\e \in E^+_{a,b,i,j}} \1_{\e \in \Edges_{\min\{a,b\}}} \cdot \nu^{+w,k,\ell}\paren*{ \dbi^{\vecd,\vect}_\CG(u_\e,v_\e) } . 
\end{split}
\end{equation}

Recall the notation in \cref{def:arr-ul-bounds}. The following lemma calculates the expectation of these random variables.
\begin{lemma}
\label{lem:Y-exp}
For all $a, b \in [k]$ and $i, j \in [\ell]$, it holds that:
\[
\Exp[ Y^-_{a,b,i,j} ] = q_{\min\{a,b\}} m \cdot A^{-w}(a,b,i,j) \hspace{1cm} \text{and} \hspace{1cm} \Exp[ Y^+_{a,b,i,j} ] = q_{\min\{a,b\}} m \cdot A^{+w}(a,b,i,j) .
\]
\end{lemma}
\begin{proof}
We only show the first equation as the proof of the second one is analogous. Examining \cref{eq:Y}, we recall that each variable $\1_{\e \in \Edges_{\min\{a,b\}}}$ is $1$ w.p. $q_{\min\{a,b\}}$. Thus, by linearity of expectation that it suffices to show that:
\begin{equation}
\label{eq:Y-exp-to-show}
m \cdot A^{-w}(a,b,i,j) = \sum_{\e \in E^-_{a,b,i,j}} \nu^{-w,k,\ell}\paren*{ \dbi^{\vecd,\vect}_\CG(u_\e,v_\e) } .
\end{equation}
This equation follows because:
\begin{align*}
m \cdot A^{-w}(a,b,i,j) &= m \cdot \sum_{(a',b',i',j') \in \Win^{w-1,k,\ell}(a,b,i,j)} \nu^{-w,k,\ell}(a',b',i',j') A(a',b',i',j') \tag{\cref{def:arr-ul-bounds}} \\
&= \sum_{(a',b',i',j') \in \Win^{w-1,k,\ell}(a,b,i,j)} \nu^{-w,k,\ell}(a',b',i',j') \sum_{\e=1}^m \1_{\dbi_\CG^{\vecd,\vect}(u_\e,v_\e) = (a',b',i',j')}\tag{\cref{eq:biasdegg}} \\
&= \sum_{\e=1}^m \1_{\dbi_\CG^{\vecd,\vect}(u_\e,v_\e) \in \Win^{w-1,k,\ell}(a,b,i,j)} \cdot \nu^{-w,k,\ell}\paren*{ \dbi_\CG^{\vecd,\vect}(u_\e,v_\e) } .
\end{align*}
Now, by \cref{eq:E}, it suffices to show that every pair of nonisolated vertices $(u,v)$, we have:
\[ 
\dbi_\CG^{\vecd,\vect}(u,v) \in \Win^{w-1,k,\ell}(a,b,i,j) \iff (u,v) \in V^-_{a,i} \times V^-_{b,j} .
\]
This follows because:
\begin{align*}
&\dbi_\CG^{\vecd,\vect}(u,v) \in \Win^{w-1,k,\ell}(a,b,i,j) \\
&\hspace{1cm}\iff \paren*{ \di_\CG^\vecd(u),\di_\CG^\vecd(v),\bi^\vect_\CG(u),\bi^\vect_\CG(v) } \in \Win^{w-1,k,\ell}(a,b,i,j) \tag{\cref{eq:db-ind}} \\
&\hspace{1cm}\iff \di_\CG^\vecd(u) \in \Win^{w-1,k}(a) \wedge \di_\CG^\vecd(v) \in \Win^{w-1,k}(b) \\
&\hspace{4cm} {} \wedge \bi^\vect_\CG(u) \in \Win^{w-1,\ell}(i) \wedge \bi^\vect_\CG(v) \in \Win^{w-1,\ell}(j)  \tag{\cref{def:win}} \\
&\hspace{1cm}\iff u \in V^-_{a,i} \wedge v \in V^-_{b,j} . \tag{\cref{eq:V,eq:Vi}}
\end{align*}
\end{proof}

Using standard concentration bounds, we get:
\begin{lemma}
\label{lem:Y-sandwich}
We have $\Pr(\Egoodcounte) \geq 999/1000$, where the probability is over the randomness in \cref{line:forget-edge} and $\Egoodcounte$ is defined as follows: For all $a,b \in [k]$ and $i,j \in [\ell]$, we have:
\begin{align*}
\abs*{ Y^-_{a,b,i,j} - \Exp[ Y^-_{a,b,i,j} ] } &\leq \frac{ \epsilon }{ 2(k\ell)^2 } \cdot q_{\min\{a,b\}} m . \\
\abs*{ Y^+_{a,b,i,j} - \Exp[ Y^+_{a,b,i,j} ] } &\leq \frac{ \epsilon }{ 2(k\ell)^2 } \cdot q_{\min\{a,b\}} m . 
\end{align*}
\end{lemma}
\begin{proof}
We will upper bound the probability that $\Egoodcounte$ does not happen. For this, we fix $a,b \in [k]$ and $i,j \in [\ell]$ and upper bound the probability that the first inequality does not hold for this $a,b$ and $i,j$ by $o(1/n)$. The proof for the second inequality is similar and the lemma, thus follows by a union bound. Fix $a,b \in [k]$ and $i,j \in [\ell]$. We have from \cref{lem:Y-exp,lem:weighted-chernoff} that:
\begin{align*}
\Pr\paren*{ \abs*{ Y^-_{a,b,i,j} - \Exp[ Y^-_{a,b,i,j} ] } \geq \frac{ \epsilon }{ 2(k\ell)^2 } \cdot q_{\min\{a,b\}} m } &\leq 2 \exp\paren*{ - \frac{ \epsilon^2 \cdot q_{\min\{a,b\}} m }{ 12 (k\ell)^4 A^{-w}(a,b,i,j) } } \\
&\leq 2 \exp\paren*{ - \frac{ \epsilon^2 \cdot q_{\min\{a,b\}} m }{ 12 (k\ell)^4 } } \tag{\cref{lem:a-bound}} \\
&\leq 2 \exp\paren*{ - \frac{ \epsilon^2 q_k m }{ 12 (k\ell)^4 } } \tag{\cref{table:layer}} \\
&= o(1/n) ,
\end{align*}
where, for the last step, recall that $k = \log (2\hat{m}) = O_\epsilon(\log n)$, while $q_k m = 2^{k^*-k} m = m \cdot \frac{ \log^6 n}{ 2 \hat{m} } \geq 0.5 \log^6 n$. So the negated expression in the exponent is $\Theta_\epsilon(\log^2 n)$ and we get the required bound of $o(1/n)$.
\end{proof}

\subsection{Counting target edges after vertex-subsampling}\label{sec:counting-edges-v}

Recall that $S_a = \{v \in [n] : \pi_a(v) = 1\}$ for $a \in [k]$ and that $S_a$ is determined by the randomness in \cref{line:hash}. Observe from \cref{line:vstored} that only vertices in $S_a$ can possibly be in $\vStored_a$. Next, let $\CJ_1,\ldots,\CJ_k$ be a set of subsets of $[m]$, corresponding to edge-indices, that may be sampled in \cref{line:forget-edge} and recall our notation $\Gk = \CJ_1,\ldots,\CJ_k$. For such a $\Gk$ and any $a, b \in [k]$ and $i, j \in [\ell]$, define the random variable:
\begin{equation}
\label{eq:X}
\begin{split}
X^{\Gk,-}_{a,b,i,j} &= \sum_{\e \in E^-_{a,b,i,j} \cap \Edges_{\min\{a,b\}}} \1_{u_\e \in S_a} \1_{v_\e \in S_b} \nu^{-w,k,\ell}\paren*{ \dbi^{\vecd,\vect}_\CG(u,v) } . \\
X^{\Gk,+}_{a,b,i,j} &= \sum_{\e \in E^+_{a,b,i,j} \cap \Edges_{\min\{a,b\}}} \1_{u_\e \in S_a} \1_{v_\e \in S_b}  \nu^{+w,k,\ell}\paren*{ \dbi^{\vecd,\vect}_\CG(u,v) } . \\
\end{split}
\end{equation}

Note that the random variables $X^{\Gk,-}_{a,b,i,j}$ and $X^{\Gk,+}_{a,b,i,j}$ are determined solely by the randomness in \cref{line:hash} (as they only depend on the randomness in the sets $S_1, \dots, S_k$). Our definitions satisfy the concentration lemma in \cref{lem:X-sandwich} below but first we calculate the expectation of the random variables defined above (recall that the value of $Y^-_{a,b,i,j}$ and $Y^+_{a,b,i,j}$, for all $a, b \in [k]$ and $i, j \in [\ell]$, and whether or not $\Egooddeg$ and $\Egoodcounte$ occur is determined by the graphs $\Gk$):

\begin{lemma}
\label{lem:X-exp}
Fix any $\Gk$. For all $a,b \in [k]$ and $i,j \in [\ell]$, we have:
\[
\Exp[ X^{\Gk,-}_{a,b,i,j} ] = p_ap_b Y^-_{a,b,i,j} \hspace{1cm}\text{and}\hspace{1cm} \Exp[ X^{\Gk,+}_{a,b,i,j} ] = p_ap_b Y^+_{a,b,i,j} .
\]
\end{lemma}
\begin{proof}
We only prove the former as the latter is analogous. By linearity of expectation, \cref{eq:Y} and the fact that the hash-functions sampled in \cref{line:hash} are $4$-wise independent, we have:
\[
\Exp[X^{\Gk,-}_{a,b,i,j}] = \sum_{\e \in E^-_{a,b,i,j} \cap \CJ_{\min\{a,b\}}} p_ap_b \cdot \nu^{-w,k,\ell}(\dbi^{\vecd,\vect}_\CG(u_\e,v_\e)) = p_ap_b \cdot Y^-_{a,b,i,j} .
\]
\end{proof}

\begin{lemma}
\label{lem:X-sandwich}
Fix any $\Gk$ such that both $\Egooddeg$ and $\Egoodcounte$ occur. It holds that $\Pr\paren*{ \Egoodcountv^{\Gk} } \geq 999/1000$, where the probability is over the randomness in \cref{line:hash} and $\Egoodcountv^{\Gk}$ is defined as follows: For all $a,b \in [k]$ and $i,j \in [\ell]$, we have:
\begin{align*}
\abs*{ X^{\Gk,-}_{a,b,i,j} - \Exp[ X^{\Gk,-}_{a,b,i,j} ] } &\leq \frac{ \epsilon }{ 2(k\ell)^2 } \cdot p_a p_b q_{\min\{a,b\}} m . \\
\abs*{ X^{\Gk,+}_{a,b,i,j} - \Exp[ X^{\Gk,+}_{a,b,i,j} ] } &\leq \frac{ \epsilon }{ 2(k\ell)^2 } \cdot p_a p_b q_{\min\{a,b\}} m . 
\end{align*}
\end{lemma}
\begin{proof}
Fix $\Gk$ and note that this also fixes $Y^-_{a,b,i,j}$ and $Y^+_{a,b,i,j}$. We will upper bound the probability that $\Egoodcountv^{\Gk}$ does not happen. For this, we fix $a,b \in [k]$ and $i,j \in [\ell]$ and upper bound the probability that the first inequality does not hold for this $a,b$ and $i,j$ by $\frac{ 10^{-4} }{ (k\ell)^2 }$. The proof for the second inequality is similar and the lemma, thus follows by a union bound. Fix an arbitrary $a,b \in [k]$ and $i,j \in [\ell]$ and assume $a \leq b$ without loss of generality.

Note from \cref{table:layer} that $p_0 \leq p_1 \dots \leq p_k \leq 1$. This means that if $p_a = 1$, then $p_a = p_b = 1$ and $X^{\Gk,-}_{a,b,i,j}$ is a constant value independent of the randomness and the inequality follows. Thus, we can assume $p_a < 1$ which means that $p_a = p_0q_a^{-1}$. For convenience, define $F = E^-_{a,b,i,j} \cap \Edges_a$ and $\nabla = \frac{ \epsilon }{ 2(k\ell)^2 } \cdot p_a p_b q_a m$.  For $\e \in F$, define the random variable:
\[
I_\e := \1_{u_\e \in S_a} \1_{v_\e \in S_b} \nu^{-w,k,\ell}\paren*{ \dbi^{\vecd,\vect}_\CG(u_\e,v_\e) } .
\]

Plugging into \cref{eq:X}, we get that $X^{\Gk,-}_{a,b,i,j} = \sum_{\e \in F} I_\e$. Next, we claim that for all $\e \in F$, there are at most $5D \cdot \frac{p_b}{p_0}$ many $\e' \in F$ such that $I_\e$ and $I_{\e'}$ are not independent, where $D$ and $p_0$ are as in \cref{table:epsn,table:epsnm}. Assuming this claim for now, \cref{lem:chebyshev-indep} says:
\[
\Pr\paren*{ \abs*{ X^{\Gk,-}_{a,b,i,j} - \Exp[ X^{\Gk,-}_{a,b,i,j} ] } \geq \nabla } \leq \frac{5D p_b \cdot \Exp[ X^{\Gk,-}_{a,b,i,j} ]}{ p_0 \nabla^2}. 
\] 
Now, recall from \cref{lem:X-exp} that $\Exp[ X^{\Gk,-}_{a,b,i,j} ] = p_ap_b Y^-_{a,b,i,j}$. As $\Egoodcounte$ occurs, we have from \cref{lem:Y-exp,lem:Y-sandwich} that $Y^-_{a,b,i,j} \leq q_a m \cdot \paren*{ A^{-w}(a,b,i,j) + \frac{ \epsilon }{ 2(k\ell)^2 } } \leq 2 q_a m$ by \cref{lem:a-bound}. Plugging in we get:
\[
\Pr\paren*{ \abs*{ X^{\Gk,-}_{a,b,i,j} - \Exp[ X^{\Gk,-}_{a,b,i,j} ] } \geq \nabla } \leq \frac{10D p_b \cdot p_a p_b q_a m}{ p_0 \nabla^2} = \frac{ 40 (k\ell)^4 D }{ p_0 \epsilon^2 p_a q_a m } , 
\] 
by definition of $\nabla$. Recalling that $p_a = p_0q_a^{-1}$, we get:
\[
\Pr\paren*{ \abs*{ X^{\Gk,-}_{a,b,i,j} - \Exp[ X^{\Gk,-}_{a,b,i,j} ] } \geq \nabla } \leq \frac{ 80 (k\ell)^4 D }{ \rho^2 \epsilon^2 } \leq \frac{ 10^{-4} }{ (k\ell)^2 } . 
\] 
It remains to show the claim. For this, we consider two cases, based on whether or not $p_b = 1$. In both cases, we use the fact that the hash-functions sampled in \cref{line:hash} are $4$-wise independent. If $p_b = 1$, then $S_b = [n]$ (with probability $1$) and $I_\e$ and $I_{\e'}$ are not independent only if $u_\e = u_{\e'}$. As $F \subseteq \Edges_a$, this means that the number of $\e' \in F$ such that the random variables $I_\e$ and $I_{\e'}$ are not independent is at most $\deg_{\CG_a}(u_\e)$, implying that it suffices to show the bound $\deg_{\CG_a}(u_\e) \leq 2D$. For this, note that $F \subseteq E^-_{a,b,i,j}$ implies by \cref{eq:E} that $u_\e \in V^-_{a,i} \subseteq V^+_a$ by \cref{eq:Vi} and \cref{item:egooddeg:1} of \cref{lem:egooddeg}. Now, using \cref{item:egooddeg:2} of \cref{lem:egooddeg}, we get that $\deg_{\CG_a}(u_\e) \leq 2D$, as desired.

Now, consider the case $p_b < 1$ which implies $p_b/p_0 = 1/q_b$. In this case, we have that $I_\e$ and $I_{\e'}$ are not independent only if the pairs $(u_\e,v_\e)$ and $(u_{\e'},v_{\e'})$ have a common element. Using the fact that $F \subseteq \Edges_a$, we get that the number of $\e' \in F$ such that the random variables $I_\e$ and $I_{\e'}$ are not independent is at most $2 + \deg_{\CG_a}(u_\e) + \deg_{\CG_a}(v_\e)$, implying that it suffices to show the bounds $\deg_{\CG_a}(u_\e) \leq 2D$ and $\deg_{\CG_a}(v_\e) \leq 2D/q_b$. For this, note that $F \subseteq E^-_{a,b,i,j}$ implies by \cref{eq:E} that $u_\e \in V^-_{a,i}$ and $v_\e \in V^-_{b,j}$. As in the previous part, we get that $u_\e \in V^+_a$ and $v_\e \in V^+_b$. Now, using \cref{item:egooddeg:2} of \cref{lem:egooddeg}, we get that $\deg_{\CG_a}(u_\e) \leq 2D$ and $\deg_{\CG_a}(v_\e) \leq 2D/q_b$, as desired.
\end{proof}

\subsection{Putting it all together: Proving \cref{lem:bias-arr-correct}}\label{sec:final-proof-of-bias-arr-correct}

We are now ready to prove \cref{lem:bias-arr-correct}.

\begin{proof}[Proof of \cref{lem:bias-arr-correct}]
Define the event $\Egood$ (over the randomness in \cref{line:hash,line:forget-edge}) that for all $a, b \in [k]$ and $i, j \in [\ell]$, it holds that:
\[
A^{-w}(a,b,i,j) - \frac{\epsilon}{(k\ell)^2} \leq \hat{A}(a,b,i,j) \leq A^{+w}(a,b,i,j) + \frac{\epsilon}{(k\ell)^2} .
\]
Observe that \cref{lem:bias-arr-correct} is the same as showing $\Pr\paren*{ \Egood } \geq 99/100$. We show this by upper bounding the probability of the complement event. Using a union bound, we have:
\begin{align*}
\Pr\paren*{ \overline{ \Egood } } &\leq \Pr\paren*{ \overline{\Eofg} \vee \overline{\Egooddeg} \vee \overline{\Egoodbias} \vee \overline{\Egoodcounte} } + \Pr\paren*{ \overline{ \Egood } \mid \Eofg,\Egooddeg,\Egoodbias,\Egoodcounte } \\
&\leq \frac{4}{1000} + \Pr\paren*{ \overline{ \Egood } \mid \Eofg,\Egooddeg,\Egoodbias,\Egoodcounte } \tag{\cref{lem:ofg,lem:deg-Ga,lem:bias-Ga,lem:Y-sandwich}} .
\end{align*}
Thus, it suffices to bound the last term by $\frac{2}{1000}$. For this, note that the events $\Eofg$, $\Egooddeg$, $\Egoodbias$, and $\Egoodcounte$ are all determined by the randomness in \cref{line:forget-edge}, or equivalently by the subsets $\Gk = \CJ_1, \dots, \CJ_k$. Thus, it suffices to show that for any $\Gk$ such that all the events $\Eofg$, $\Egooddeg$, $\Egoodbias$, and $\Egoodcounte$ occur, we have $\Pr\paren*{ \overline{ \Egood } \mid \Gk } \leq \frac{1}{1000}$. Fix such a $\Gk$. By another union bound, we have:
\begin{align*}
\Pr\paren*{ \overline{ \Egood } \mid \Gk } &\leq \Pr\paren*{ \overline{ \Eofh^{\Gk} } \vee \overline{ \Egoodcountv^{\Gk} } \mid \Gk } + \Pr\paren*{ \overline{ \Egood } \mid \Gk, \Eofh^{\Gk}, \Egoodcountv^{\Gk} } \\
&\leq \frac{2}{1000} + \Pr\paren*{ \overline{ \Egood } \mid \Gk, \Eofh^{\Gk}, \Egoodcountv^{\Gk} } \tag{\cref{lem:ofh,lem:X-sandwich}} ,
\end{align*}
where the last inequality also uses the fact that $\Egoodcountv^{\Gk}$ and $\Eofh^{\Gk}$ are determined by the randomness in \cref{line:hash} and therefore, independent of $\Gk$. To finish the proof, we show that for any choice of the randomness in \cref{line:hash,line:forget-edge} such that the events $\Gk$, $\Eofh^{\Gk}$, and $\Egoodcountv^{\Gk}$ occur, the event $\Egood$ also occurs. Fix such a randomness for the rest of the proof. Note that fixing this randomness fixes the value of all the random variables defined above and also fixes the value of all variables in \cref{alg:sketch,alg:combine}. Henceforth, we shall often abuse notation and use the name of the random variable to denote the value it is fixed to, e.g., $X^{\Gk,-}_{a,b,i,j}$ will denote the value $X^{\Gk,-}_{a,b,i,j}$ is fixed to. We use a similar notation for variables. We first claim that, for all $a,b \in [k]$ and $i,j \in [\ell]$, we have:
\begin{equation}
\label{eq:bacorr:1}
X^{\Gk,-}_{a,b,i,j} \leq \hat{A}(a,b,i,j) \cdot p_a p_b q_{\min\{a,b\}} m \leq X^{\Gk,+}_{a,b,i,j} .
\end{equation}
We prove \cref{eq:bacorr:1} later but assuming it for now, we get from \cref{lem:X-exp,lem:X-sandwich} that (for all $a,b \in [k]$ and $i,j \in [\ell]$):
\[
\frac{ Y^-_{a,b,i,j} }{ q_{\min\{a,b\}} m } - \frac{\epsilon}{2(k\ell)^2} \leq \hat{A}(a,b,i,j) \leq \frac{ Y^+_{a,b,i,j} }{ q_{\min\{a,b\}} m } + \frac{\epsilon}{2(k\ell)^2} .
\]
Continuing using \cref{lem:Y-exp,lem:Y-sandwich}, we get (for all $a,b \in [k]$ and $i,j \in [\ell]$):
\[
A^{-w}(a,b,i,j) - \frac{\epsilon}{(k\ell)^2} \leq \hat{A}(a,b,i,j) \leq A^{+w}(a,b,i,j) + \frac{\epsilon}{(k\ell)^2} ,
\]
as desired.

It remains to prove \cref{eq:bacorr:1}. We only show the second inequality as the proof for the first one is analogous. For this, fix $a,b \in [k]$ and $i,j \in [\ell]$ and note from \cref{eq:hata} that $\hat{A}(a,b,i,j) \cdot p_a p_b q_{\min\{a,b\}} m = \AEst_{a,b,i,j}$. Thus, it suffices to show that $\AEst_{a,b,i,j} \leq X^{\Gk,+}_{a,b,i,j}$. For this, we first combine \cref{line:aest} with \cref{item:biasdeggood:2} of \cref{lem:biasdeggood} to get:
\[
\AEst_{a,b,i,j} \leq \sum_{ (u,v) \in \eStored_{\min\{a,b\}} }\1_{u \in \vEst_{a,i}} \1_{v \in \vEst_{b,j}} \nu^{+w,k,\ell}\paren*{ \dbi^{\vecd,\vect}_\CG(u,v) } .
\]
By \cref{item:lem:overflow:5} of \cref{lem:overflow}, we deduce:
\[
\AEst_{a,b,i,j} \leq \sum_{\e \in \CJ_{\min\{a,b\}}} \1_{u_\e \in \vEst_{a,i}} \1_{v_\e \in \vEst_{b,j}} \nu^{+w,k,\ell}\paren*{ \dbi^{\vecd,\vect}_\CG(u_\e,v_\e) } .
\]

Next, we apply \cref{item:biasdeggood:1} of \cref{lem:biasdeggood} to get $\vEst_{a,i} \times \vEst_{b,j} = \paren*{ \hat{V}_{a,i} \cap S_a } \times \paren*{ \hat{V}_{b,j} \cap S_b }$. The right-hand side equals $\paren*{ \hat{V}_{a,i} \times \hat{V}_{b,j} } \cap \paren*{ S_a \times S_b }$ and we get:
\[
\AEst_{a,b,i,j} \leq \sum_{ \e \in \Edges_{\min\{a,b\}} } \1_{u_\e \in \hat{V}_{a,i})} \1_{v_\e \in \hat{V}_{b,j}}\1_{u \in S_a} \1_{v \in S_b} \nu^{+w,k,\ell}\paren*{ \dbi^{\vecd,\vect}_\CG(u_\e,v_\e) } .
\]
To finish the proof, use \cref{eq:E,eq:X,lem:egoodbias} to get
\begin{align*}
    \AEst_{a,b,i,j} &\leq \sum_{\e \in E^+_{a,b,i,j} \cap \CJ_{\min\{a,b\}}} \1_{u_\e \in S_a} \1_{v_\e \in S_b} \nu^{+w,k,\ell}\paren*{ \dbi^{\vecd,\vect}_\CG(u_\e,v_\e) } = X^{\Gk,+}_{a,b,i,j}.
\end{align*}

\end{proof}

%% file: Files/06-snapshot-proofs.tex
\section{From pointwise smoothed array estimates to smoothed matrix estimates: Proving \cref{lem:pointwise-impl-smooth}}\label{sec:pointwise-impl-smooth-analysis}

In this section, we prove \cref{lem:pointwise-impl-smooth} in several steps. This lemma states that a ``pointwise smoothed estimate'' (\cref{def:pointwise}) of an array $A$ implies a ``smoothed estimate'' (\cref{def:smoothed-est}) of its projection $M = \Proj(A)$ (\cref{def:projection}). We will aim to apply this in the algorithm where $A = \RSnapG$ is the refined snapshot of a graph $\CG$ (\cref{def:refined-snapshot}) and $M=\SnapG$ is its snapshot (\cref{def:snapshot}).

To prove \cref{lem:pointwise-impl-smooth}, we first use the following simple proposition, which states that to estimate the matrix $M^{\sim w}$ with $1$-norm error, it suffices to instead estimate the array $A^{\sim w}$ with $1$-norm error.

\begin{proposition}\label{lem:4-to-2}
For every $w < k,\ell \in \BN$, and $A \in \A^{k,\ell}$, let $M = \Proj(A)$. Then $M^{\sim w} = \Proj(A^{\sim w})$. Moreover, for any array $\hat{A} \in \A^{k,\ell}$, define $\hat{M} = \Proj(\hat{A})$. Then $\|\hat{M}-M^{\sim w}\|_1 \leq \|\hat{A}-A^{\sim w}\|_1$.
\end{proposition}

We prove this proposition (again by double-counting) in \cref{sec:double-counting-proofs} below. The second key ingredient involves the definition of pointwise smoothed estimates, and in particular, the upper- and lower-bound arrays $A^{-w}$ and $A^{+w}$ (\cref{def:arr-ul-bounds}).:

\begin{lemma}[``Sandwiching $A^{\sim w}$'']\label{lem:sandwich}
There exists a universal constant $\Cwin > 0$ such that for every $w < k,\ell \in \BN$ and every $A \in \AS^{k,\ell}$, $\|A^{+w} - A^{-w}\|_1 \leq \Cwin/w$.
\end{lemma}

We prove this lemma in \cref{sec:sandwiching} below. But first, we show how these two statements together will let us finally prove \cref{lem:pointwise-impl-smooth}:

\begin{proof}[Proof of \cref{lem:pointwise-impl-smooth} modulo \cref{lem:4-to-2,lem:sandwich}]
Let $\hat{A}$ be a $(w,\delta)$-pointwise smoothed estimate for $A$, for $\delta = \epsilon/(k\ell)^2$. This means, by definition, that $A^{-w} - \delta \leq \hat{A} \leq A^{+w}+\delta$ entrywise. Thus, we can pick some $B \in \A^{k,\ell}$ such that $\|B - \hat{A}\|_\infty \leq \delta$ and $A^{-w} \leq B \leq A^{+w}$ entrywise. The former implies $\|B - \hat{A}\|_1 \leq \delta (k\ell)^2 = \epsilon$. To use the latter, we recall that also $A^{-w} \leq A^{\sim w} \leq A^{+w}$ entrywise; therefore, $|B-A^{\sim w}| \leq A^{+w}-A^{-w}$ entrywise, so $\|B - A^{\sim w}\|_1 \leq \|A^{+w}-A^{-w}\|_1$. By \cref{lem:sandwich} we deduce $\|B - A^{\sim w}\|_1 \leq \Cwin/w$. Hence, by the triangle inequality, $\|\hat{A} - A^{\sim w}\|_1 \leq \epsilon + \Cwin/w$. Finally, we apply \cref{lem:4-to-2} to conclude that $\|\hat{M}-M^{\sim w}\|_1 \leq \epsilon + \Cwin/w$ as well, for $\hat{M} = \Proj(\hat{A})$ and $M = \Proj(A)$.
\end{proof}

\subsection{Basic properties of windows}\label{sec:window-facts}

To prove \cref{lem:4-to-2,lem:sandwich}, we begin by stating a number of basic facts about windows (presented only for the necessary dimensions, for brevity).

\begin{fact}[Size of 1D and 4D windows]\label{fact:window-size}
The size of $d$-dimensional windows is bounded within a factor of $2^d$. In particular,
\begin{itemize}
    \item In one dimension, for every $w < \ell \in \BN$ and $i \in [\ell]$, \[ w+1 \leq |\Win^{w,\ell}(i)| \leq 2w+1. \]
    \item In four dimensions, for every $w < k,\ell \in \BN$ and $a,b \in [k]$, $i,j \in [\ell]$, \[ (w+1)^4 \leq |\Win^{w,k,\ell}(a,b,i,j)| \leq (2w+1)^4. \]
\end{itemize}
\end{fact}

\begin{fact}[Size difference of 1D windows]\label{fact:window-size-diff}
For every $w \leq w' < \ell \in \BN$ and $i \in [\ell]$, \[ |\Win^{w',\ell}(i) \setminus \Win^{w,\ell}(i)| \leq 2(w'-w). \]
\end{fact}

\begin{fact}[Symmetry of containment]\label{fact:window-symm-contain}
Containment in windows is a symmetric property, i.e.,
\begin{itemize}
    \item In two dimensions, for every $w < \ell \in \BN$ and $i,j,i',j' \in [\ell]$, \[ (i',j') \in \Win^{w,\ell}(i,j) \iff (i,j) \in \Win^{w,\ell}(i',j'). \]
    \item In four dimensions, for every $w < k,\ell \in \BN$ and $a,b,a',b' \in [k]$, $i,j,i',j' \in [\ell]$, \[ (a',b',i',j') \in \Win^{w,k,\ell}(a,b,i,j) \iff (a,b,i,j) \in \Win^{w,k,\ell}(a',b',i',j'). \]
\end{itemize}
\end{fact}

As a corollary, we have the following fact:

\begin{fact}[Counting containments]\label{fact:window-count-contain}
The number of windows a particular index is contained equals the size of its window, i.e.,
\begin{itemize}
    \item In two dimensions, for every $w < \ell \in \BN$ and $i',j' \in [\ell]$, the number of distinct windows containing $(i',j')$ is \[ |\{(i,j) \in [\ell]^2 : \Win^{w,\ell}(i,j) \ni (i',j')\}| = |\Win^{w,\ell}(i',j')|. \]
    \item In four dimensions, for every $w < k, \ell \in \BN$ and $a',b' \in [k]$, $i',j' \in [\ell]$, the number of distinct windows containing $(a',b',i',j')$ is \[ |\{(a,b,i,j) \in [k]^2 \times [\ell]^2 : \Win^{w,k,\ell}(a,b,i,j) \ni (a',b',i',j')\}| = |\Win^{w,k,\ell}(a',b',i',j')|. \]
\end{itemize}
\end{fact}

The following is essentially the triangle inequality for the $\infty$-norm, and roughly states that ``nearby windows look alike'':

\begin{fact}[Triangle inequality for windows]\label{fact:window-triangle-ineq}
For every $w+w' < k, \ell \in \BN$, and every $a,b,a',b' \in [k]$, $i,j,i',j'\in[\ell]$, \[ (a',b',i',j') \in \Win^{w',k,\ell}(a,b,i,j) \implies \Win^{w,k,\ell}(a',b',i',j') \subseteq \Win^{w+w',k,\ell}(a,b,i,j). \]
\end{fact}

\subsection{Double-counting arguments: Proving \cref{prop:smoothing-preserves-sum,lem:4-to-2}}\label{sec:double-counting-proofs}

In this subsection, we prove \cref{lem:4-to-2}, but we begin with a simpler version of the argument which proves \cref{prop:smoothing-preserves-sum}, which states that smoothing preserves the sum of entries in a matrix.

\begin{proof}[Proof of \cref{prop:smoothing-preserves-sum}]
This holds by \cref{fact:window-count-contain} and ``double-counting''. In particular,
\begin{align*}
\sum_{i,j=1}^\ell M^{\sim w}(i',j') &= \sum_{i,j=1}^\ell \sum_{i',j' \in \Win^{w,\ell}(i,j)} \nu^{\sim w,\ell}(i',j') M(i',j') \tag{def. of $M^{\sim w}$} \\
&= \sum_{i',j'=1}^\ell |\{(i,j) \in [\ell]^2 : \Win^{w,\ell}(i,j) \ni (i',j')\}| \cdot \nu^{\sim w,\ell}(i',j')  M(i',j') \tag{exchanging sums} \\
&= \sum_{i',j'=1}^\ell |\Win^{w,\ell}(i',j')| \cdot \nu^{\sim w,\ell}(i',j')  M(i',j') \tag{\cref{fact:window-count-contain}} \\
&= \sum_{i',j'=1}^\ell M(i',j'). \tag{def. of $\nu^{\sim w,\ell}$} \\
\end{align*}
\end{proof}

\begin{proof}[Proof of \cref{lem:4-to-2}]
Recall $M = \Proj(A)$ and $N = \Proj(A^{\sim w})$. We first want to show $M^{\sim w} = N$. Again, we can do this by ``double-counting''. Recall that $\Win^{w,k,\ell}(a,b,i,j) = \Win^{w,k}(a,b) \times \Win^{w,\ell}(i,j)$ (and thus $\nu^{\sim w,k,\ell}(a,b,i,j) = \nu^{\sim w,k}(a,b) \cdot \nu^{\sim w,\ell}(i,j)$). We have:

\begin{align*}
    N(i,j) &= \sum_{a,b=1}^k A^{\sim w}(a,b,i,j) \tag{def. of $N$} \\
    &= \sum_{a,b=1}^k \sum_{(a',b',i',j') \in \Win^{w,k,\ell}(a,b,i,j)} \nu^{\sim w,k,\ell}(a',b',i',j') A(a',b',i',j') \tag{def. of $A^{\sim w}$} \\
    &= \sum_{(i',j') \in \Win^{w,\ell}(i,j)} \nu^{\sim w,\ell}(i',j') \left(\sum_{a,b=1}^k \sum_{(a',b')\in\Win^{w,k}(a,b)} \nu^{\sim w,k}(a',b') A(a',b',i',j')\right) \tag{exchanging sums} \\
    &= \sum_{(i',j') \in \Win^{w,\ell}(i,j)} \nu^{\sim w,\ell}(i',j') \left(\sum_{a',b'=1}^k |\{(a,b) \in [k]^2 : \Win^{w,k}(a,b) \ni (a',b')\}| \cdot \nu^{\sim w,k}(a',b') A(a',b',i',j')\right) \tag{exchanging sums} \\
    &= \sum_{(i',j') \in W^{w,\ell}(i,j)} \nu^{\sim w,\ell}(i',j') \left(\sum_{a',b'=1}^k A(a',b',i',j')\right) \tag{\cref{fact:window-count-contain}} \\
    &= \sum_{(i',j') \in W^{w,\ell}(i,j)} \nu^{\sim w,\ell}(i',j') M(i',j') \tag{def. of $M$} \\
    &= M^{\sim w}(i,j) \tag{def. of $M^{\sim w}$}.
\end{align*}

Now, consider an array $\hat{A} \in \A^{k,\ell}$; recall, $\hat{M} = \Proj(\hat{A})$, and we want to show that $\|\hat{M}-M^{\sim w}\|_1 \leq \|\hat{A}-A^{\sim w}\|_1$. This follows essentially from the triangle inequality. Indeed:

\begin{align*}
    \|\hat{M}-M^{\sim w}\|_1 &= \sum_{i,j=1}^\ell |\hat{M}(i,j) - M^{\sim w}(i,j)| \tag{def. of $\|\cdot\|_1$} \\
    &= \sum_{i,j=1}^\ell |\hat{M}(i,j) - N(i,j)| \tag{first part of lemma} \\
    &= \sum_{i,j=1}^\ell \abs*{\sum_{a,b=1}^k \hat{A}(a,b,i,j) - \sum_{a,b=1}^k A^{\sim w}(a,b,i,j)} \tag{defs. of $\hat{M}, N$} \\
    &\leq \sum_{a,b=1}^k \sum_{i,j=1}^\ell \abs*{\hat{A}(a,b,i,j) - A^{\sim w}(a,b,i,j)} \tag{triangle ineq.} \\
    &= \|\hat{A}-A^{\sim w}\|_1 \tag{def. of $\|\cdot\|_1$}.
\end{align*}
\end{proof}

\subsection{``Sandwiching'': Proving \cref{lem:sandwich}}\label{sec:sandwiching}

Now we prove \cref{lem:sandwich}, bounding the $1$-norm between $A^{-w}$ and $A^{+w}$ for an array $A \in \AS^{k,\ell}$. We begin with the following lemma:

\begin{lemma}[``Borders of 4D windows are effectively 3D'']\label{lem:small-borders}
There exists a universal constant $\Cwin' > 0$ such that for every $w < k,\ell \in \BN$, and $a,b \in [k], i,j \in [\ell]$, $|\Win^{w+1,k,\ell}(a,b,i,j) \setminus \Win^{w-1,k,\ell}(a,b,i,j)| \leq \Cwin' w^3$.
\end{lemma}

\begin{proof}
Let $w^-_1 = |\Win^{w-1,k}(a)|$ and $w^+_1 = |\Win^{w+1,k}(a)|$ denote the sizes of the 1-dimensional windows around $a$ of sizes $w-1$ and $w+1$, respectively. By \cref{fact:window-size}, we have $w^-_1 \leq 2(w-1)+1 = 2w-1$. By \cref{fact:window-size-diff}, we have $w^+_1 \leq w^-_1 + 4$.

We can similarly define $w^-_2,w^+_2,w^-_3,w^+_3,w^-_4,w^+_4$, and then we have
\begin{align*}
|\Win^{w+1,k,\ell}(a,b,i,j) - \Win^{w-1,k,\ell}(a,b,i,j)|  &= w^+_1w^+_2w^+_3w^+_4-w^-_1w^-_2w^-_3w^-_4 \\
&\leq (w^-_1+4)(w^-_2+4)(w^-_3+4)(w^-_4+4)-w^-_1w^-_2w^-_3w^-_4 \\
&= 4 C_3 + 16C_2 + 64 C_1 + 256
\end{align*}

where $C_t$ denotes the degree-$t$ symmetric polynomial evaluated on $w^-_1,w^-_2,w^-_3,w^-_4$ (i.e., it sums the products of sets of $t$ values). Combining with the upper-bounds on $w_i^-$, we conclude \[ 4 \binom{4}3 (2w-1)^3 + 16 \binom{4}2 (2w-1)^2 + 64 \binom{4}1 (2w-1) + 256 = 16(8w^3+12w^2+14+5), \] so setting $\Cwin' = 16(8+12+14+5) = 624$ is sufficient.
\end{proof}

Now, we prove \cref{lem:sandwich}:

\begin{proof}[Proof of \cref{lem:sandwich}]
We begin by examining the difference $\Delta(a,b,i,j) = A^{+w}(a,b,i,j) - A^{-w}(a,b,i,j)$ in a single entry. Note that $\Delta(a,b,i,j) \geq 0$ for all $a,b,i,j$ since $A$'s entries are nonnegative. We have
\begin{align*}
    \Delta(a,b,i,j) &= \sum_{(a',b',i',j') \in \Win^{w-1,k,\ell}(a,b,i,j)} (\nu^{+w,k,\ell}(a',b',i',j')-\nu^{-w,k,\ell}(a',b',i',j')) A(a',b',i',j') \\
    &+ \sum_{(a',b',i',j') \in \Win^{w+1,k,\ell}(a,b,i,j) \setminus \Win^{w-1,k,\ell}(a,b,i,j)} \nu^{+w,k,\ell}(a',b',i',j') A(a',b',i',j').
\end{align*}
In other words, the error comes from two places: Different normalizations inside $\Win^{w-1,k,\ell}(a,b,i,j)$, and then extra entries in $\Win^{w+1,k,\ell}(a,b,i,j)$ which are not counted at all in $\Win^{w-1,k,\ell}(a,b,i,j)$.

Recall that \[ \nu^{+w,k,\ell}(a',b',i',j') = \max_{(a'',b'',i'',j'') \in \Win^{1,k,\ell}(a',b',i',j')} 1/|\Win^{w,k,\ell}(a'',b'',i'',j'')|. \] Now by the triangle inequality for windows (\cref{fact:window-triangle-ineq}), we have $\Win^{w-1,k,\ell}(a',b',i',j') \subseteq \Win^{w,k,\ell}(a'',b'',i'',j'')$ where $(a'',b'',i'',j'')$ is the maximizing index in $\nu^{+w,k,\ell}$; thus,  $\nu^{+w,k,\ell}(a',b',i',j') \leq 1/|\Win^{w-1,k,\ell}(a',b',i',j')|$. Similarly, $\nu^{-w,k,\ell}(a',b',i',j') \geq 1/|\Win^{w+1,k,\ell}(a',b',i',j')|$. So \[ \nu^{+w,k,\ell}(a',b',i',j')-\nu^{-w,k,\ell}(a',b',i',j') \leq \frac{|\Win^{w+1,k,\ell}(a',b',i',j')| - |\Win^{w-1,k,\ell}(a',b',i',j')|}{|\Win^{w+1,k,\ell}(a',b',i',j')| \cdot |\Win^{w-1,k,\ell}(a',b',i',j')|} \leq \Cwin' w^{-5}, \] where we upper-bound the numerator by $\Cwin' w^3$ with \cref{lem:small-borders} and lower-bound the denominator by $w^8$ with \cref{fact:window-size}. Hence we can write
\begin{align*}
    \Delta(a,b,i,j) &\leq \Cwin' w^{-5} \sum_{(a',b',i',j') \in \Win^{w-1,k,\ell}(a,b,i,j)} A(a',b',i',j') \\
    &+ w^{-4} \sum_{(a',b',i',j') \in \Win^{w+1,k,\ell}(a,b,i,j) \setminus \Win^{w-1,k,\ell}(a,b,i,j)} A(a',b',i',j'),
\end{align*}
where for the second term we use that $\nu^{+w,k,\ell}(a',b',i',j') \geq 1/|\Win^{w+1,k,\ell}(a',b',i',j')| \geq w^{-4}$ (by \cref{fact:window-triangle-ineq,fact:window-size}).

Finally, we bound $\|A^{+w}-A^{-w}\|_1 = \sum_{a,b=1}^k \sum_{i,j=1}^\ell \Delta(a,b,i,j)$ by double-counting and symmetry of inclusion in windows (\cref{fact:window-symm-contain,fact:window-count-contain}). In particular, an entry $A(a',b',i',j')$ will be counted $|\Win^{w-1,k,\ell}(a',b',i',j')| \leq 16w^4$ times in the first sum (by \cref{fact:window-size}), and $|\Win^{w+1,k,\ell}(a',b',i',j')\setminus\Win^{w-1,k,\ell}(a',b',i',j')| \leq \Cwin' w^3$ times in the second sum, giving an overall bound of $\|A^{+w}-A^{-w}\|_1 \leq 17\|A\|_1 \Cwin'/w$, which is the desired bound when $\Cwin = 17\Cwin'$ (since $\|A\|_1 = 1$ by assumption).
\end{proof}

%% file: Files/07-smoothing-proofs.tex
\section{``Smoothing graphs'': Proving \cref{lem:graph-smoothing}}\label{sec:graph-smoothing-analysis}

In this section, we prove \cref{lem:graph-smoothing}, which roughly states that ``smoothing a graph's snapshot does not affect its $\mdcut$ value too much''. (Smoothing a matrix was defined in \cref{def:smoothing-matrix} above.) More precisely, \cref{lem:graph-smoothing} states that given any graph $\CG$ one can construct a graph $\CH$ whose snapshot $\SnapH$ resembles a smoothing of $\CG$'s snapshot $\SnapG$, and whose $\mdcut$ value $\val_\CH$ is similar to $\CG$'s value $\val_\CG$, and that these two error terms are $O(\lambda w)$ where $w$ is the size of the windows being smoothed over and $\lambda$ is the width of the bias partition.

We construct $\CH$ based on $\CG$ in two stages: First, we construct an intermediate (weighted) graph $\CK$ by ``blowing up'' each vertex of $\CG$ (which preserves biases of vertices and the $\mdcut$ value). Then, we slightly modify $\CK$ to form the (weighted) graph $\CH$ by perturbing the bias at each vertex, and argue that this does not change the $\mdcut$ value too much.

We begin with some notation. We label the vertices of $\CG$ by $[n] = \{1,\ldots,n\}$. We assume WLOG that $\CG$ has no isolated vertices (if not, we simply remove all isolated vertices and analyze the new graph; this doesn't affect $\SnapG$ or $\val_\CG$).

\subsection{Defining $\CK$ via (weighted) blowups}

First, given $\CG$, we construct a graph $\CK$ as follows. For each vertex $v \in [n]$, in $\CK$ we create $ |\Win^{w,\ell}(\bi^\vect_\CG(v))|$ distinct vertices labeled $\{v\} \times \Win^{w,\ell}(\bi^\vect_\CG(v))$. Thus, $\CK$ has vertex-set $\CV \eqdef \bigcup_{v =1}^n \{v\} \times \Win^{w,\ell}(\bi^\vect_\CG(v))$.

Now if $\AdjG(u,v)$ is the (positive) weight of an edge $u \to v$ in $\CG$, we divide its weight evenly among all the copies of $u$ and $v$ in $\CK$, i.e., we place weight
\begin{equation}\label{eq:weight-K}
    \AdjK(u,i',v,j') \eqdef \nu^{\sim w,\ell}(\bi^\vect_\CG(u,v)) \cdot \AdjH(u,v)
\end{equation}
on all edges $(u,i') \to (v,j')$ for $(i',j') \in \Win^{w,\ell}(\bi^\vect_\CG(u,v))$.\footnote{Recall $\nu^{\sim w,\ell}(i,j) = |\Win^{w,\ell}(i,j)|$.}

We claim that $\CK$ has the following properties:

\begin{claim}\label{claim:K}
For every vertex $v \in [n]$ and $i' \in \Win^{w,\ell}(\bi^\vect_\CG(v))$, we have:
\begin{enumerate}
    \item $\dout_\CK(v,i') = \dout_\CG(v)\cdot|\Win^{w,\ell}(\bi^\vect_\CG(v))|^{-1}$.
    \item $\din_\CK(v,i') = \din_\CG(v)\cdot|\Win^{w,\ell}(\bi^\vect_\CG(v))|^{-1}$.
    \item $\d_\CK(v,i') = \d_\CG(v)\cdot|\Win^{w,\ell}(\bi^\vect_\CG(v))|^{-1}$.
    \item $\bias_\CK(v,i') = \bias_\CG(v)$.
\end{enumerate}
Further, $m_\CG = m_\CK$ and $\val_\CG = \val_\CK$.
\end{claim}

The proof is basically a series of double-counting arguments using the definition of $\CK$, and we defer it until the end of the section.

\subsection{Constructing $\CH$ via perturbations}

Now, we construct $\CH$ by modifying $\CK$. We begin with a new graph $\CH_0$ consisting of a copy of $\CK$ and a new isolated vertex $\sstar$, so that $\CH_0$ has vertex-set $\CV \cup \{\sstar\}$, and let $\CV^* = \bigcup_{v\in[n]} (\{v\} \times (\Win^{w,\ell}(\bi^\vect_\CG(v))) \setminus \{\bi^\vect_\CG(v)\})) \subseteq \CV$ denote the set of vertices whose biases we'd like to modify. We let $m_0 \eqdef m_\CG = m_\CK = m_{\CH_0}$ denote the total weight in $\CG$, $\CK$, and $\CH_0$ (which coincide by \cref{claim:K}).

Next, we iteratively modify the biases of vertices using the following claim, which roughly lets us decrease the bias of a vertex $(v,i') \in \CV^*$ by increasing its in-degree and decreasing its out-degree, or increase the bias of a vertex by increasing its out-degree and decreasing its in-degree, ``without modifying the graph too much''.

\begin{claim}\label{claim:change-bias}
Let $\CL$ be any weighted graph on vertex-set $\CV\cup\{\sstar\}$. For every $(v,i') \in \CV^*$, $0 \leq \alpha$, and $0 \leq \beta \leq \dout_\CL(v,i')$, there exists an ``updated'' graph $\CL'$ on vertex-set $\CV\cup\{\sstar\}$, such that:

\begin{enumerate}
    \item ``All vertices except $(v,i')$ and $\sstar$ stay the same'': For all $(u,j') \neq (v,i') \in \CV$, $\dout_{\CL'}(u,j') = \dout_\CL(u,j')$ and $\din_{\CL'}(u,j') = \din_\CL(u,j')$.
    \item ``$(v,i')$'s change is controlled'': $\dout_{\CL'}(v,i') = \dout_\CL(v,i')-\beta$ and $\din_{\CL'}(v,i') = \din_\CL(v,i')+\alpha$.
    \item ``The new weight added is bounded'': $m_{\CL'} = m_\CL + \alpha$.\label{item:changed-total-weight}
    \item ``Changes in weight are bounded'': $\|\Adj_{\CL'} - \AdjL\|_1 = \alpha+2\beta$. \label{item:changed-weight}
\end{enumerate}

The same holds when we swap $\dout$'s with $\din$'s.
\end{claim}

A few quick remarks: Note that $\AdjL$ and $\Adj_{\CL'}$ are \emph{not} normalized, so in \cref{item:changed-weight} we are simply summing the magnitude of the change in weight from $\CL$ to $\CL'$ over every edge. Also, in general the out- or in-degree of $\sstar$ will have to change in order to ensure the first two items hold; the second two items state that this doesn't cause too many problems for us.

\begin{proof}[Proof of \cref{claim:change-bias}]
Suppose we are decreasing the bias of $(v,i')$ --- the other case is analogous. Roughly, here is the process: We start by setting $\CL'$ to be the same as $\CL$; we will only adjust the weights of a few edges. We  increase the weight of the edge $\sstar \to (v,i')$ by $\alpha$. Then, we arbitrarily distribute weight $\beta$ among out-edges $(v,i') \to (u,j')$ and then transfer this distribution to edges $\sstar \to (u,j')$.

More formally, we do the following. For each $(u,j') \neq (v,i') \in \CV$, we pick a coefficient $0 \leq \beta_{(u,j')} \leq \AdjL(v,i',u,j')$ such that $\sum_{(u,j') \neq (v,i') \in \CV} \beta_{(u,j')} = \beta$; this is possible as $\beta \leq \dout_\CL(v,i') = \sum_{(u,j') \neq (v,i') \in \CV} \AdjL(v,i',u,j')$ by assumption, and so $\beta_{(u,j')}$'s can be picked greedily. Now we define the adjacency matrix of the new graph $\CL'$ as follows: We set $\AdjLP(\star,\star,v,i') = \AdjL(\star,\star,v,i') + \alpha$. For each $(u,j') \neq (v,i') \in \CV$, we set $\AdjLP(v,i',u,j') = \AdjL(v,i',u,j') - \beta_{(u,j')}$ and $\AdjLP(\star,\star,u,j') = m_\CL(\star,\star,u,j') + \beta_{(u,j')}$. Finally, for all remaining edges, i.e., $(w,k'),(u,j') \neq (v,i') \in \CV$, we set $\AdjLP(w,k',u,j') = \AdjL(w,k',u,j')$. The four desiderata follow immediately from this construction.
\end{proof}

Now consider an arbitrary vertex $(v,i') \in \CV^*$, and suppose WLOG $i' < \bi^\vect_\CG(v)$, so we are aiming to decrease $\bias_{\CH_0}(v,i')$. Arbitrarily pick some ``target bias'' $b^*$ such that $\Ind^\vect(b^*) = i'$. Consider setting \[ \alpha, \beta := \frac12(\bias_{\CH_0}(v,i')-b^*)\; \d_{\CH_0}(v,i'). \] Since $\bias_{\CH_0}(v,i')=\bias_\CG(v)$ and $i' \in \Win^{w,\ell}(\bi^\vect_{\CH_0}(v))$, we know $(\bias_{\CH_0}(v,i')-b^*)\le \lambda(w+1)$. Then since $\vect$ is $\lambda$-wide, we have
\begin{equation}
    \alpha, \beta \leq \frac12 \lambda(w+1) \; \d_{\CH_0}(v,i').\label{eq:changed-weight}
\end{equation}
Moreover, if we apply \cref{claim:change-bias} with $\CL = \CH_0$, then we get a new graph $\CH_1 = \CL'$ in which $(v,i')$ has bias
\begin{align*}
    \bias_{\CH_1}(v,i') &= \frac{(\dout_{\CH_0}(v,i') - \beta) - (\din_{\CH_0}(v,i') + \alpha)}{(\dout_{\CH_0}(v,i') - \beta) + (\din_{\CH_0}(v,i') + \alpha)} \\
    &= \bias_{\CH_0}(v,i') - \frac{2\beta}{\d_{\CH_0}(v,i')} \\
    &= b^*
\end{align*}
and the in- and out-degrees of all other vertices except $\sstar$ are fixed. We can apply this procedure to all vertices in $\CV^*$ whose biases need to be decreased, along with the analogous procedure for vertices in $\CV^*$ whose biases need to be \emph{increased}. After applying this operation to every vertex in $\CV^*$, we end up with a graph which we'll denote $\CH$, which has the following properties:

\begin{claim}\label{claim:H}
For \emph{every} vertex $(v,i') \in \CV$, $\bi^\vect_\CH(v,i') = i'$. Also, $m_0 \leq m_\CH \leq (1+\lambda(w+1))m_0$ and $\|\AdjH-\AdjHZ\|_1 \leq 3\lambda(w+1) m_0$.
\end{claim}

\begin{proof}
When we constructed $\CH$ from $\CH_0$, when modifying a vertex $(v,i') \in \CV^*$, \cref{eq:changed-weight} says we set $\alpha, \beta \leq \frac12 \lambda(w+1) \deg_{\CH_0}(v,i')$; hence, by \cref{claim:change-bias}, \[
    m_0 \leq m_\CH \leq m_0 + \frac12 \sum_{(v,i') \in \CV^*} \lambda(w+1) \deg_{\CH_0}(v,i') = (1+\lambda(w+1))m_0
\]
since $m_0 = m_\CH = \frac12 \sum_{(v,i') \in \CV \cup \{\sstar\}} \d_\CH(v,i'))$ is the total weight in $\CH$ and $\CV^* \subseteq \CV \cup \{\sstar\}$. Similarly, by \cref{claim:change-bias} and the triangle inequality, we can bound the distance between the \emph{unnormalized} adjacency matrices of $\CH$ and $\CH_0$:
\[
    \|\AdjH - \AdjHZ\|_1 \leq \frac12 \sum_{(v,i') \in \CV^*} 3\lambda(w+1) \deg_{\CH_0}(v,i') \leq 3\lambda(w+1) m_0.
\]
\end{proof}

\subsection{Proving that $\CH$ fulfills the desiderata}

Now, we claim that $\CH$ fulfills the two desiderata (i.e., the upper bounds on $|\val_\CG - \val_\CH|$ and $\|\SnapGw-\SnapH\|_1$), proving the lemma.

\begin{proof}[Proof of \cref{lem:graph-smoothing}]
According to our definition, both $\val_\CG$ and $\SnapG$ are normalized by $m_\CG$ (the total weight in $\CG$), and the same goes for $\CH$. Towards both proofs, it will be convenient for us to first bound the distance between the \emph{unnormalized} variants of these objects, and then argue that $m_\CH \approx m_\CG$ and therefore the normalized versions are close as well. We set $\Csmooth = 7$.

\paragraph{Upper-bounding $|\val_\CG-\val_\CH|$.} Recall that $\val_\CK = \val_\CG$. Further, $\val_{\CH_0} = \val_\CK$ as $\CH_0$ simply adds an isolated vertex (i.e., $\sstar$).

Now let \[ h = m_\CH \val_\CH \text{ and } g = m_0 \val_{\CH_0} \] denote the \emph{unnormalized} $\mdcut$ values of $\CH$ and $\CH_0$, respectively. We claim that $|g-h| \leq \|\AdjH-\AdjHZ\|_1$ (and thus, $|g-h|\leq 3\lambda(w+1) m_0$ by \cref{claim:H}).
Observe that
\[
|g-h| = \left|\max_{\vecy \in \{0,1\}^{\CV^*}}m_\CH \val_\CH(\vecy) -\max_{\vecy \in \{0,1\}^{\CV^*}}m_0 \val_{\CH_0}(\vecy) \right| \le \left|\max_{\vecy \in \{0,1\}^{\CV^*}}(m_\CH \val_\CH(\vecy) - m_0 \val_{\CH_0}(\vecy)) \right| \, .
\]
Indeed, if $\vecy \in \{0,1\}^{\CV^*}$ is any assignment, we have
\begin{align*}
    \abs*{m_\CH \val_\CH(\vecy) - m_0 \val_{\CH_0}(\vecy)} &= \abs*{\sum_{(u,i'),(v,j')\in \CV\cup\{\sstar\}} (\AdjH(u,i',v,j') - \AdjHZ(u,i',v,j')) \cdot y_{(u,i')} (1-y_{(v,j')})} \\
    &\leq \sum_{(u,i'),(v,j') \in \CV\cup\{\sstar\}} \abs*{\AdjH(u,i',v,j') - \AdjHZ(u,i',v,j')} \tag{tri. ineq. and $y_{(u,i')} \in \{0,1\}$} \\
    &= \|\AdjH-\AdjHZ\|_1. \tag{def. of $\|\cdot\|_1$}
\end{align*}

Finally, we can write \[ |\val_\CG - \val_\CH| = \abs*{\frac1{m_0} g - \frac1{m_\CH} h} = \abs*{\frac1{m_0}(g-h) + \eta} \leq 3\lambda(w+1) + \eta \] by the triangle inequality, where $\eta = h(1/m_0 - 1/m_\CH) = h(m_\CH-m_0)/(m_\CH m_0)$ is a ``normalization error'' term. We can upper bound this error term: \[ \eta = \frac{m_\CH (m_\CH - m_0) \val_\CH }{m_\CH m_0} \leq \frac{m_\CH - m_0}{m_\CH} \leq \lambda(w+1) \] using $\val_\CH \leq 1$ and \cref{claim:H}.

\paragraph{Upper-bounding $\|\SnapGw-\SnapH\|_1$.} Now, we analyze the distance between the unnormalized matrices $G = m_0 \SnapGw$ and $H = m_\CH \SnapH$. This suffices, similarly to the case of the $\mdcut$ values, as by the triangle inequality \[ \left\|\frac1{m_0}G - \frac1{m_\CH} H \right\|_1 \leq \frac1{m_0}\|G-H\|_1 + \eta \] where $\eta = \|H\|_1(1/m_0 - 1/m_\CH) \leq \lambda(w+1)$ as $\|H\|_1 = m_\CH$.

The basic idea here is that if we consider an edge $(u,v)$ in $\CG$, its weight $\AdjG(u,v)$ will contribute (up to normalization) to the entries corresponding to $\Win^{w,\ell}(\bi^\vect_\CG(u,v))$ in $\SnapGw$. In particular, for each position $(i',j') \in \Win^{w,\ell}(\bi^\vect_\CG(u,v))$ which $(u,v)$ contributes to in $\SnapG$, we have created a copy $(u,i',v,j')$ in $\CH$ which will contribute to the corresponding entry in $\SnapH$! This correspondence between contributions is not exact: our construction of $\CH$ slightly modified the weights inherited from $\CG$, and introduced a new vertex $\sstar$ with uncontrolled bias, but these issues can be taken care of with our bound on $\|\AdjH-\AdjHZ\|_1$.

By definition, we have
\begin{align*}
    H(i,j) &= \sum_{(u,i'),(v,j') \in \CV} \AdjH(u,i',v,j') \1_{\bi^\vect_\CH(u,i',v,j')=(i,j)} \\
    &+ \sum_{(u,i') \in \CV} \AdjH(u,i',\star,\star) \1_{\bi^\vect_\CH(u,i',\star,\star)=(i,j)} \\
    &+ \sum_{(v,j') \in \CV} \AdjH(\star,\star,v,j') \1_{\bi^\vect_\CH(\star,\star,v,j')=(i,j)} \\
    &+ \AdjH(\star,\star,\star,\star) \1_{\bi^\vect_\CH(\star,\star,\star,\star)=(i,j)}.
\end{align*}

Now define the matrix $H'$ by ``forgetting about $\sstar$'', i.e., we only take the first term in the expansion of $H$: \[ H'(i,j) = \sum_{(u,i'), (v,j') \in \CV} \AdjH(u,i',v,j') \1_{\bi^\vect_\CH(u,i',v,j') = (i,j)}. \] Note that $\|H-H'\|_1 = \deg_\CH\sstar$ since every edge in $\CH$ contributes to precisely one entry. Further, since $\sstar$ was isolated in $\CH_0$, \cref{claim:H,claim:change-bias} imply $\deg_\CH\sstar \leq 3\lambda(w+1) m_0$. Thus, by the triangle inequality, it suffices to prove that $\|G-H'\|_1 \leq 3 \lambda(w+1) m_0$.

Now $H'$'s entries sum only over vertices in $\CV$ (i.e., not $\sstar$), and we constructed $\CH$ so that for each vertex $(v,i')$ we have $\bi^\vect_\CH(v,i') = i'$ (\cref{claim:H}). Thus, the $(i,j)$-th entry in $H'$ counts only edges $(u,i',v,j')$ such that $i = i'$ and $j = j'$, so we can write \[ H'(i,j) = \sum_{(u,i) \in \CV} \sum_{(v,j) \in \CV} \AdjH(u,i,v,j). \]

On the other hand, we have \[ m_0 \SnapG(i,j) = \sum_{u,v=1}^n \AdjG(u,v) \1_{\bi^\vect_\CG(u,v)=(i,j)} \] and so we have
\begin{align*}
    G(i,j) &= m_0 \sum_{(i',j') \in \Win^{w,\ell}(i,j)} \nu^{\sim w,\ell}(i',j') \SnapG(i',j') \tag{def. of smoothing} \\
    &= m_0\sum_{(i',j') \in \Win^{w,\ell}(i,j)} \nu^{\sim w,\ell}(i',j') \cdot \frac{1}{m_0}\sum_{u,v=1}^n \AdjG(u,v) \1_{\bi^\vect_\CG(u,v) = (i',j')} \tag{def. of $\SnapG$}\\
    &= \sum_{(i',j') \in \Win^{w,\ell}(i,j)} \nu^{\sim w,\ell}(i',j')\sum_{u,v=1}^n \AdjG(u,v) \1_{\bi^\vect_\CG(u,v) = (i',j')} \, .
\end{align*}

Now an edge $(u,v)$ has a nonzero contribution iff $(i,j) \in \Win^{w,\ell}(\bi^\vect_\CG(u,v))$, i.e., if $(u,i),(v,j) \in \CV$; in this case, its contribution is precisely $\nu^{\sim w,\ell}((\bi^\vect_\CG(u,v))) \cdot \AdjG(u,v) = \AdjK(u,i,v,j) = \AdjHZ(u,i,v,j)$, so we can write \[ G(i,j) = \sum_{(v,i) \in \CV} \sum_{(u,j) \in \CV} \AdjHZ(u,i,v,j). \]

This matches our expression for $H'(i,j)$, except that the weights come from $\CH_0$ instead of $\CH$. To complete the analysis:
\begin{align*}
    \|G-H'\|_1 &= \sum_{i,j=1}^\ell \left|\sum_{(u,i) \in \CV} \sum_{(v,j) \in \CV} \AdjHZ(u,i,v,j)-\sum_{(u,i) \in \CV} \sum_{(v,j) \in \CV} \AdjH(u,i,v,j)\right| \\
    &\leq \sum_{i,j=1}^\ell \sum_{(u,i) \in \CV} \sum_{(v,j) \in \CV}\left| \AdjHZ(u,i,v,j)-\AdjH(u,i,v,j)\right| \tag{triangle ineq.} \\
    &\leq \sum_{(u,i') \in \CV} \sum_{(v,j') \in \CV} \left| \AdjHZ(u,i',v,j')-\AdjH(u,i',v,j')\right| \tag{each term occurs once} \\
    &\leq \|\AdjH - \AdjHZ\|_1 \\
    &\leq 3\lambda(w+1) m_0 \tag{\cref{claim:H}},
\end{align*}
as desired.
\end{proof}

\subsection{Proof of \cref{claim:K}}

\begin{proof}[Proof of \cref{claim:K}]
For any $v \in [n]$, any particular copy $(v,i')$ of a vertex $v$ has out-degree
\begin{align*}
    \dout_\CK(v,i') &= \sum_{u=1}^n \sum_{j' \in \Win^{w,\ell}(\bi^\vect_\CG(u))} \AdjK(v,i',u,j') \tag{def. of $\dout$} \\
    &= \sum_{u =1}^n \sum_{j' \in \Win^{w,\ell}(\bi^\vect_\CG(u))} \nu^{\sim w,\ell}(\bi^\vect_\CG(u,v)) \cdot \AdjG(u,v) \tag{def. of $\CK$} \\
    &= \sum_{u =1}^n \sum_{j' \in \Win^{w,\ell}(\bi^\vect_\CG(u))} |\Win^{w,\ell}(\bi^\vect_\CG(u,v))|^{-1} \cdot \AdjG(u,v) \tag{def. of $\nu$} \\
    &= |\Win^{w,\ell}(\bi^\vect_\CG(v))|^{-1} \sum_{u=1}^n \AdjG(u,v) \tag{def. of 2D windows}\\
    &= |\Win^{w,\ell}(\bi^\vect_\CG(v))|^{-1} \; \dout_\CG(v)\tag{def. of $\dout$}.
\end{align*}

Similarly, $\din_\CK(v,i') = \din_\CG(v)\cdot|\Win^{w,\ell}(\bi^\vect_\CG(v))|^{-1}$. Thus, their sum $\dout_\CK(v,i') = \dout_\CG(v)\cdot|\Win^{w,\ell}(\bi^\vect_\CG(v))|^{-1}$, and their normalized difference $\bias_\CK(v,i') = \bias_\CG(v)$.

Next, we observe that
\begin{align*}
    m_\CK &= \sum_{(v,i') \in \CV} \dout_\CK(v,i') \\
    &= \sum_{v =1}^n \sum_{i' \in \Win^{\sim w,\ell}(\bi^\vect_\CG(v))} \dout_\CG(v) \cdot |\Win^{\sim w,\ell}(\bi^\vect_\CG(v))|^{-1}\tag{defs. of $\CK$ and $\CV$} \\
    &= \sum_{v =1}^n \dout_\CG(v) \\
    &= m_\CG.
\end{align*}

Now, we claim that $\val_\CG \leq \val_\CH$. Indeed, let $\vecx$ be any assignment for $\CG$. We can construct an assignment $\vecy \in \{0,1\}^{\CV}$ to $\CK$'s vertices by simply giving each copy of a vertex in $\CK$ its assignment in $\CG$, i.e., $y_{(v,i')} = x_v$. This assignment has value
\begin{align*}
    \val_\CK(\vecy) &= \sum_{(u,i'),(v,j') \in \CV} \AdjK(u,i',v,j') y_{(u,i')} (1-y_{(v,j')}) \tag{def. of $\val$} \\
    &= \sum_{u,v=1}^n \sum_{(i',j') \in \Win^{w,\ell}(\bi^\vect_\CG(u,v))} \AdjK(u,i',v,j') \; y_{(u,i')} (1-y_{(v,j')}) \tag{def. of $\CV$} \\
    &= \sum_{u,v=1}^n \sum_{(i',j') \in \Win^{w,\ell}(\bi^\vect_\CG(u,v))} \nu^{\sim w,\ell}(\bi^\vect_\CG(u,v)) \cdot \AdjG(u,v) \; y_{(u,i')} (1-y_{(v,j')}) \tag{def. of $\CK$} \\
    &= \sum_{u,v=1}^n \sum_{(i',j') \in \Win^{w,\ell}(\bi^\vect_\CG(u,v))} \nu^{\sim w,\ell}(\bi^\vect_\CG(u,v)) \cdot \AdjG(u,v) \; x_u (1-x_v) \tag{def. of $\vecy$} \\
    &= \sum_{u,v=1}^n \AdjG(u,v) \; x_u (1-x_v) \tag{def. of $\nu^{\sim w,\ell}$} \\
    &= \val_\CG(\vecx) \tag{def. of $\val$}.
\end{align*}

Conversely, let $\vecy$ be any assignment for $\CK$. We construct a \emph{probabilistic} assignment $\vecx$ by, for each vertex $v \in [n]$, sampling $i' \in \Win^{w,\ell}(\bi^\vect_\CG(v))$ uniformly and independently at random, and setting $x_v = y_{(v,i')}$. This assignment has expected value

\begin{align*}
    \Exp[\val_\CG(\vecx)] &= \Exp\left[\sum_{u,v=1}^n \AdjG(u,v) x_u (1-x_v) \right] \tag{def. of $\val$} \\
    &= \sum_{u,v =1}^n \AdjG(u,v) \Exp[x_u (1-x_v)] \tag{lin. of $\Exp$} \\
     &= \sum_{u=1}^n \sum_{v \neq u} \AdjG(u,v) \Exp[x_u (1-x_v)] \tag{no self-loops in $\CG$} \\
    &= \sum_{u=1}^n \sum_{v \neq u} \AdjG(u,v) \Exp[x_u] \Exp[1-x_v] \tag{independence} \\
    &= \sum_{u,v=1}^n \sum_{(i',j') \in \Win^{w,\ell}(\bi^\vect_\CG(u,v))} \nu^{\sim w,\ell}(\bi^\vect_\CG(u,v)) \cdot \AdjG(u,v) \; y_{(u,i')} (1-y_{(v,j')}) \tag{def. of $\vecx$}, \\
\end{align*}
which equals $\val_\CK(\vecy)$ by the previous calculation.
\end{proof}

%% file: Files/AA-appendix-proofs.tex
\section{Proofs of preliminary lemmas}\label{app:prelim-proofs}

\begin{proof}[Proof of \cref{lem:cut-sparsifier}]
For convenience, let $m' = m_{\Gsparse}$ denote the number of edges in $\Gsparse$, and let $p = \pspar$; for an assignment $\vecx \in \{0,1\}^V$, let $y(\vecx) \eqdef m \cdot \val_\CG(\vecx)$ and $y'(\vecx) \eqdef m'\cdot \val_{\Gsparse}(\vecx)$ denote the (unnormalized) numbers of constraints satisfied by $\vecx$ in $\CG$ and $\Gsparse$, respectively; and $g \eqdef m \cdot \val_\CG(\vecx) = \max_{\vecx \in \{0,1\}^n} y(\vecx)$ and $g' \eqdef m' \cdot \val_{\Gsparse}(\vecx) = \max_{\vecx \in \{0,1\}^n} y'(\vecx)$ denote the (unnormalized) values of $\CG$ and $\CG$', respectively.

Enumerate the edges of $\CG$ as $e_1=(u_1,v_1),\ldots,e_m=(u_m,v_m)$. For each $j \in [m]$, let $I_j$ be the indicator random variable for the event that $e_j$ survives in $\Gsparse$. Thus, we can write $m'$ and $y'(\vecx)$ as sums of independent $\Bern_p$ random variables: \[ m' = \sum_{j=1}^m I_j \text{ and } y'(\vecx) = \sum_{j=1}^m I_j x_{u_j} (1-x_{v_j}). \] Our first goal is to show that w.h.p., $|g' - pg| \leq \epssparse pg/2$.

\paragraph{Lower-bounding $g'$.} Consider any fixed assignment $\vecx^* \in \{0,1\}^V$ maximizing $\CG$'s $\mdcut$ value, i.e., $g = y(\vecx^*)$. In particular, $y(\vecx^*) \geq \frac14 m$. We have $\Exp[y'(\vecx^*)] = py(\vecx^*) = pg$ by linearity of expectation, so by the Chernoff bound (\cref{lem:chernoff-twosided}), \[ \Pr[y'(\vecx^*) \leq (1-\epssparse/2) pg] \leq \exp(-\epssparse^2 pg/8). \] Since $g \geq m / 4$, by our assumed bound on $p$, we get a bound of $\leq \exp(-\Csparse n/32) = o(1)$. Since $g' \geq y'(\vecx^*)$, we conclude $g' \geq (1-\epssparse/2) pg$ except with probability $o(1)$.

\paragraph{Upper-bounding $g'$.} To prove an upper bound on $g'$, we take a union bound over all assignments $\vecx \in \{0,1\}^n$, and prove an $o(2^{-n})$ bound on the probability that $y'(\vecx) \geq (1+\epssparse/2) pg$.

Fix an assignment $\vecx \in \{0,1\}^V$. Again, $\Exp[y'(\vecx)] = py(\vecx)$. We consider two cases depending $y(\vecx)$. If $y(\vecx) \leq m/12$, we have \[ \Pr[y'(\vecx) \geq (1+\epssparse/2)pg] \leq \Pr[y'(\vecx) \geq pm/4] \leq \exp(-pm/8) \] using $g \geq m/4$ and \cref{lem:chernoff-highdev}. By our assumed bound on $p$ this is $\exp(-\Csparse n/(8\epssparse^2)) = o(2^{-n})$. Otherwise, we have \[ \Pr[y'(\vecx) \geq (1+\epssparse/2)pg] \leq \Pr[y'(\vecx) \geq (1+\epssparse/2) p y(\vecx)] \leq \exp(-\epssparse^2 p y(\vecx) / 12) \] using $g \geq y(\vecx)$ and \cref{lem:chernoff-ub} (with $\epssparse \leq 1$). This is $\leq \exp(-\Csparse n/144) = o(2^{-n})$ by assumption on $y(\vecx)$ and $p$. Thus, even taking a union bound over all $2^n$ assignments, we conclude that $g' \leq (1+\epssparse/2)pg$ except with probability $o(1)$.

\paragraph{Bounding $m'$.} Finally, by the Chernoff bound (\cref{lem:chernoff-twosided}), we know \[ \Pr[|m'-pm| \geq \epssparse p m / 2] \leq 2\exp(-\epssparse^2 p m / 12). \] By assumption on $p$, this is at most $2\exp(-\Csparse n/12)$ which is $o(1)$.

\paragraph{Putting the bounds together.} It remains to prove a bound on $|\val_\CG - \val_{\Gsparse}|$. As in \cref{sec:graph-smoothing-analysis}, we split into ``real'' and ``normalization'' error terms: \[ |\val_\CG - \val_{\Gsparse}| = \abs*{\frac1m g - \frac1{m'}g'} \leq \abs*{\frac1m g - \frac 1{pm} g'} + \abs*{\frac1{pm} - \frac1{m'}} g'. \] The first term can be bounded by $\epssparse / 2$ since w.h.p. $|g' - pg| \leq \epssparse p g / 2$ and $g \leq m$. For the second, using $g' \leq m'$ the term becomes $|m'/pm - 1| = |m' - pm|/(pm)$, which is at most $\epssparse / 2$ since $|m' - pm| \leq \epssparse pm /2$.
\end{proof}